\documentclass{article}
\usepackage{template}
\usepackage{mathrsfs,mathtools}
\usepackage{cleveref}
\usepackage{enumitem}

\title{Agnostic Product Mixed State Tomography via Robust Statistics}
\author{Alvan Arulandu\thanks{Harvard University, aarulandu@college.harvard.edu. Part of this work was performed while the author was a student at the Quantum@UW REU. } \and Ilias Diakonikolas\thanks{University of Wisconsin, Madison, ilias@cs.wisc.edu. Supported by NSF Medium Award CCF-2107079 and an H.I. Romnes Faculty Fellowship.} \and Daniel Kane\thanks{University of California, San Diego, dakane@ucsd.edu. Supported by NSF Medium Award CCF-210754.} \and Jerry Li\thanks{University of Washington, jerryzli@cs.washington.edu}.}

\begin{document}

\maketitle

\begin{abstract}
We study the complexity of two qualitatively related learning problems---one quantum and one classical. 
In the quantum setting, we consider the task of
agnostic tomography 
for the natural class of product mixed states. 
Specifically, given $N$ copies of an $n$-qubit state $\rho$, 
the goal is to output a nearly-optimal product mixed state approximation of $\rho$ 
in trace distance. While there has been a flurry of recent work on agnostic tomography for
pure state ansatz such as product states or stabilizer states, 
no polynomial-time guarantees were previously known for \emph{mixed state} ansatz. 
In the classical setting, we consider the task of robustly learning binary product distributions. Specifically, given $N$ samples from an unknown distribution $p$ 
on $\{0,1\}^n$, 
the goal is to output a nearly-optimal binary product approximation to $p$. 
This is a basic problem in robust statistics 
with a significant gap between the error guarantee of known efficient algorithms and the information-theoretic minimum. 

As our main contributions, we establish the following new results on the complexity of these tasks:
\begin{itemize}[leftmargin=*]
\item We give a semi-agnostic tomography algorithm for product   
      mixed states with polynomial copy and computational 
      complexity that achieves an error of 
      $O(\mathrm{opt} \log 1 / \mathrm{opt})$, 
      where $\mathrm{opt}$ is the trace distance of the best fit product   
      mixed state.
      This is the first efficient algorithm that achieves any non-trivial agnostic tomography guarantee for any class of mixed state ansatz, and we do so using only single-qubit, single-copy measurements. As a corollary, we obtain a new semi-agnostic tomography algorithm for {\em pure} product states.
      We complement our upper bound for product mixed states with a Quantum Statistical Query lower bound, providing formal evidence that the error guarantee achieved by our algorithm is 
      near-optimal among computationally efficient algorithms. We also establish an unconditional lower bound demonstrating that adaptivity is necessary for our agnostic tomography task, 
      so long as the algorithm only uses single-qubit two-outcome projective measurements.

\item We give a semi-agnostic algorithm for robustly learning binary product distributions with polynomial sample and computational complexity that achieves an error of $O(\mathrm{opt} \log 1 / \mathrm{opt})$, where $\mathrm{opt}$ is the total variation distance of the best fit product distribution. We complement our upper bound with a Statistical Query lower bound, providing 
evidence that the attained error guarantee is nearly optimal 
for efficient algorithms. This essentially resolves the efficient robust learnability of product distributions, marking 
the first algorithmic improvement since the initial work of (Diakonikolas et al. 2016). 
\end{itemize}
A central conceptual contribution of our work 
is an efficient black-box 
reduction from agnostic tomography of product mixed states 
to the robust learning of binary product distributions.  
We believe that this connection between quantum learning and classical robust statistics is of independent interest and may have broader implications.
As a corollary, we establish that these two learning problems are essentially equivalent.
Our new robust learner for binary product distributions 
introduces several technical innovations that may be useful in 
other contexts. These include a new measure that tightly characterizes the total variation distance between two binary product distributions in terms of their means, as well as a novel method for bounding the sample complexity of stability conditions arising in robust statistics.
\end{abstract}

\thispagestyle{empty}

\newpage

\tableofcontents
\setcounter{page}{0}
\thispagestyle{empty}

\newpage

\section{Introduction} \label{sec:intro}

\paragraph{Background and Motivation} 
In this paper, we consider two qualitatively similar learning tasks: one in the quantum setting and one in the classical setting. An informal description of each task is given below.
\begin{itemize}[leftmargin=*]
    \item \textbf{Agnostic Tomography:} Given $N$ copies of an $n$-qubit mixed state $\rho$, 
    can we efficiently approximate the best description of the state $\rho$ 
    within a given ``nice'' quantum ansatz class?
    \item \textbf{Robust Distribution Learning:} Given $N$ i.i.d.
    samples from an $n$-dimensional distribution $p$, 
    can we efficiently approximate the best fit to $p$ within a given 
    ``nice'' (classical) distribution family?
\end{itemize}
Both these tasks are of fundamental importance within their respective fields, and indeed, share very similar motivations.
In real-world applications---both quantum 
and classical---complex phenomena are typically modeled 
using simplifying assumptions.  As a result, our ansatz class 
(or distribution family in the classical setting) 
will almost surely fail to precisely 
capture the target quantum state (or data distribution).
Hence, it is important to develop efficient 
learning algorithms that are 
able to tolerate some degree of \emph{model misspecification}. 

In the classical setting, there exists a rich history 
of learning in the presence of adversarial noise, both in the 
supervised~\cite{Valiant:85,Haussler:92,KearnsLi:93,KSS:94} and the 
unsupervised settings~\cite{DKKLMS16, LaiRV16}. The reader is 
referred to~\cite{diakonikolas2023algorithmic} for a recent book
on the topic.

The task of agnostic tomography has received a wave of recent interest from the quantum computing community, see, e.g.,~\cite{grewal2024agnostic, bakshi2025learning,chen2025stabilizer}.
Beyond the model misspecification motivation, 
an additional motivation, largely unique to the quantum setting, 
is that agnostic tomography algorithms may allow us 
to verify the effectiveness of popular empirical approximations 
arising from mean-field theories, 
such as those underlying Hartree-Fock 
theory~\cite{hartree1928wave,fock1930naherungsmethode,slater1928self,bardeen1957theory} and 
density functional theory~\cite{hohenberg1964inhomogeneous,levy1979universal,vignale1987density}.

In contrast to the classical setting, our understanding of agnostic tomography remains quite limited. A particularly important gap is the case of \emph{mixed-state ansatz classes}. Indeed, all prior efficient algorithms for agnostic tomography apply 
only to structured classes of {\em pure-state} ansatz, such as product states~\cite{bakshi2025learning}, stabilizer states~\cite{chen2025stabilizer}, 
and product stabilizer states~\cite{grewal2024agnostic}. Moreover, the algorithmic ideas underlying these results do not appear to extend to the mixed-state setting. 
This gap is especially striking given the central role of mixed states in quantum information theory and the fact that many physically relevant quantum systems are naturally modeled by mixed states. Consequently, obtaining computationally efficient agnostic tomography algorithms for mixed-state ansatz classes has emerged 
as a central open problem. For example, Gibbs states at finite temperature are inherently mixed, and thus existing methods cannot in general be expected to output accurate approximations even for thermal states of simple Hamiltonians, despite the fact that this is among the most basic quantum estimation tasks.

\paragraph{Summary of Contributions}
In this work, we take a first step toward a theory of efficient agnostic tomography for 
mixed-state ansatz classes. Specifically, we give efficient algorithms for (semi-)agnostic 
tomography of {\em mixed product states}, arguably the most basic and fundamental class of mixed-
state ansatz. Our approach is based on a formal reduction from this quantum tomography 
task to a well-studied problem in classical robust distribution learning. To the best of 
our knowledge, this is the first use of a connection between agnostic tomography and 
robust statistics to obtain computationally efficient algorithms for agnostic 
tomography,\footnote{We note that independent work of~\cite{aliakbarpour2025adversarially} 
draws a connection between certain {\em exponential-time} quantum learning tasks under 
worst-case measurement noise and robust learning.} and we believe that this connection may 
be of broader conceptual interest.

As further evidence for the power of this connection, we show that tools from 
classical robust statistics also yield new guarantees for closest {\em pure} product-
state approximation. In particular, we obtain nearly-optimal error guarantees for the 
product-state approximation problem studied in~\cite{bakshi2025learning} in fully-
polynomial time, while using substantially simpler measurements. More broadly, this 
connection reveals an important conceptual distinction between the pure- and mixed-state 
settings in agnostic tomography. Prior work on pure-state agnostic tomography typically 
obtained error guarantees of the form $\opt + \eps$, where $\opt$ denotes the error of 
the best approximation to the target state within the ansatz class. In contrast, by 
leveraging classical lower bounds from robust estimation, we provide strong evidence that 
such guarantees are computationally intractable for mixed-state ansatz classes. Instead, 
in the mixed-state setting, the appropriate goal is a nearly-optimal (or semi-agnostic) 
guarantee, in which the error scales as $f(\opt)$ for some suitably well-behaved function 
$f$.

On the technical side, the main obstacle to efficient agnostic tomography turns out to be 
the design of an improved algorithm for the classical problem of robustly learning 
product distributions over the hypercube, a problem of independent interest in robust 
statistics. This problem already appeared in some of the earliest algorithmic works that 
initiated the field~\cite{DKKLMS16, LaiRV16}. Yet despite sustained interest in this question and its generalizations to other discrete distribution families (see, e.g.,~\cite{cheng2018robust,cheng2021robust,DiakonikolasKSS21}), 
prior to our work no efficient algorithm with nearly-optimal error guarantees was known.

\subsection{Our Results} \label{ssec:results}

\subsubsection{Agnostic Tomography of Product Mixed States}   \label{ssec:quantum-results}

As discussed above, we develop an efficient agnostic tomography algorithm for the class 
of \emph{product mixed states}, i.e., states of the form 
$\pi_1 \otimes \pi_2 \otimes \cdots \otimes \pi_n$, 
where $\pi_1, \ldots, \pi_n$ are arbitrary one-qubit mixed states. 
Beyond being a fundamental class, efficient tomography for product mixed states 
is also relevant to testing widely used  nonzero-temperature variants of mean-field 
approximations, including spin-glass versions of Hartree--Fock--Bogoliubov 
theories~\cite{bardeen1957theory,bogoljubov1958new,valatin1961generalized,bach1994generalized} 
and Kohn--Sham density functional theory~\cite{kohn1965self}. 
Product mixed states also arise as a special case of popular empirical ansatz classes, 
such as neural-network generative ansatz~\cite{carrasquilla2019reconstructing} 
and the product spectrum ansatz~\cite{martyn2019product}. 
More broadly, if one ultimately aims to develop a general theory 
of agnostic tomography for such richer classes, 
then product mixed states constitute a natural and necessary starting point.

This discussion leads to the following open question: 
\begin{center}
    {\it Does there exist a polynomial-time algorithm 
    for agnostic tomography\\ of product mixed states 
    with near-optimal error guarantees?}
\end{center}
As our first main result, 
we essentially resolve this question. Specifically, we 
develop the first efficient learning algorithm for this task 
with near-optimal error tolerance, 
and also establish a nearly-matching computational lower bound.

\paragraph{Formal Setup}
Before we state our main results for the quantum setting,
we provide a definition of the underlying estimation task.

\begin{definition}[(Semi-)Agnostic Learning for Product Mixed States]\label{prob:agnostic-pmix-intro}
    Let $\cM_n=\{\pi_1 \otimes \cdots\otimes \pi_n:\pi_i\in \C^{2\times 2}\}$ 
    be the family of product mixed states over $n$ qubits. 
    Given copies of an arbitrary quantum state $\rho\in \C^{2^n\times 2^n}$,
    and a desired accuracy $\eps>0$, the goal of the learner is to output 
    $\hat{\pi} \in \cM_n$ such that  $\trd(\hat{\pi},\rho)\leq f(\opt)+\eps$,
    where $\trd$ is the trace distance and $\opt = \opt(\cM_n) \eqdef \inf_{\pi \in \cM_n} \trd(\pi,\rho)$.
   Here $f: \R_+ \mapsto \R_+$ is a monotone nondecreasing function such that  
   $\lim_{t \rightarrow 0}f(t)=0$.
\end{definition}

\noindent Definition~\ref{prob:agnostic-pmix-intro} is the natural 
generalization of the standard notion of (semi)-agnostic learning for classical distributions. 
The (exact) agnostic setting corresponds to the special case where $f(\opt) = O(\opt)$. 
We use the term ``semi-agnostic'' for the setting where $f(\opt)$ 
is a nondecreasing function only of $\opt$---independent of the dimension $n$, 
which satisfies $\lim_{t \to 0} f(t) = 0$.

A few additional remarks are in order. 
First, prior work on agnostic tomography~\cite{grewal2024agnostic,chen2025stabilizer,bakshi2025learning} largely focused on obtaining optimal approximations in (in)fidelity.
However, in this mixed state setting, trace distance is in many ways the more natural measure of closeness between states.
For instance, while (in)fidelity governs the optimal distinguishability between pure states, trace distance governs the optimal distinguishability between mixed states.
Second, while obtaining a constant-factor optimal agnostic guarantee 
(i.e., with error $2 \opt+\eps$) is feasible with $\poly(n/\eps)$  
many copies,  
the standard methods to achieve this~\cite{buadescu2021improved} 
require exponential time in general. 
Moreover, such a computational limit may be inherent (as we show for our problem). 
This motivates research into algorithms that, like the one we present, 
have higher error rates 
and run in polynomial time. 
Third, we note that efficiently achieving error $n \, \opt+\eps$ 
is straightforward, as we can consider each qubit independently. As in the relevant classical 
learning theory literature, the challenge is to obtain {\em dimension-independent} error 
guarantee in polynomial time. 

\medskip

We are now ready to state our main positive result in this setting (see  Corollary~\ref{cor:quantum-main}). 

\begin{theorem}[Computationally Efficient Semi-Agnostic Learner for Product Mixed States] \label{thm:q-main-inf}
There exists a semi-agnostic learner for product mixed states 
that draws $N = \poly (n, 1/ \eps)$ copies, 
uses only single-qubit, unentangled measurements, 
runs in $\poly(N)$ time, and outputs a $\widehat{\pi} \in \cM_n$ 
so that with high probability, $\trd(\widehat{\pi}, \rho)  = O(\opt \log(1/\opt))+\eps, $
where we denote $\opt =  \opt(\cM_n)$. 
\end{theorem}

Theorem~\ref{thm:q-main-inf} gives the first computationally efficient semi-agnostic learner for product mixed states with dimension-independent error guarantees. Additionally, as 
a feature of our approach, our algorithm only uses very simple classes of measurements. 
This stands in contrast to prior work~\cite{bakshi2025learning,chen2025stabilizer}  
which leverages highly entangled measurements across the different qubits. 

It is natural to ask whether the relaxed, semi-agnostic error guarantee attained in Theorem~\ref{thm:q-main-inf} 
is inherent for computationally efficient algorithms. 
In Theorem~\ref{thm:qsq-intro}, we provide formal evidence that this is indeed the case for the class of Quantum SQ algorithms. 

\vspace{-0.3cm}

\paragraph{Discussion} 
Beyond the statement itself, the underlying approach to establish Theorem~\ref{thm:q-main-inf} is of independent interest and we believe may have broader implications. Specifically, 
to design our efficient quantum tomography algorithm, 
we establish a formal connection between agnostic tomography 
and classical robust statistics. 
In more detail, we show that agnostic tomography 
of product mixed states is \emph{equivalent}, up to constant factors, to the classical task of robustly learning 
a binary product distribution. 
Importantly, this equivalence preserves sample/copy complexity and computational complexity. 
To establish our quantum upper bound (Theorem~\ref{thm:q-main-inf}), 
we leverage this connection together with a new near-optimal algorithm 
for robustly learning product distributions (see Theorem~\ref{thm:robust-main-inf}). Interestingly, we show that this connection also holds in the opposite direction. 
This gives us our Quantum SQ lower bound for agnostic tomography, by leveraging
a classical SQ lower bound for robustly learning binary products that we establish (Theorem~\ref{thm:sq-informal}). 

\medskip

As an additional interesting implication of Theorem~\ref{thm:q-main-inf}, we obtain 
a new fully-polynomial time semi-agnostic learner for {\em pure} product states. 
Specifically, we establish the following (see Theorem~\ref{thm:agnostic-pure}).

\begin{theorem}[Computationally Efficient Semi-Agnostic Learner for Pure Product States]\label{thm:pure-intro} 
There is a non-adaptive algorithm, using only single-copy, single-qubit measurements, with the following performance guarantee. The algorithm draws $N=\tilde{O}(n/\eps^2)$ copies of an arbitrary quantum state $\rho\in \C^{2^n\times 2^n}$, runs in $\poly(N)$ time, 
and with high probability outputs a description of a pure product state $\ket{\hat \pi}$ such that $\trd(\rho,\ketbra{\hat \pi}{\hat \pi}) = O(\opt \sqrt{\log(1/\opt)} )+\eps$,  
where $\opt = \opt(\Pi_n) \eqdef \inf_{\pi \in \Pi_n} \trd(\pi,\rho)$ and $\Pi_n$ is the class of pure product states.
\end{theorem}

\vspace{-0.3cm}

\paragraph{Discussion}
The most direct comparison is the guarantee of~\cite{bakshi2025learning}, 
which gives an algorithm that outputs a pure product state $\ket{\hat \pi}$ 
so that the {\em fidelity} between $\ket{\hat \pi}$ and $\rho$ is at least $\opt_F - \eta$, 
where $\opt_F$ is the best fidelity achievable by any pure product state. 
That algorithm runs in polynomial time (in $n$ and $1 / \eta$), 
but requires that $\opt_F$ is sufficiently large.
In contrast, our guarantees are for trace distance, 
and we achieve a weaker semi-agnostic guarantee.
In the special case where $\rho$ is also promised to be pure, 
we can obtain a fidelity guarantee as well. 
Namely, if $\opt_F = 1 - \delta$, then our algorithm outputs 
a pure state $\ket{\widehat{\pi}}$ so that the fidelity 
between $\rho$ and $\ket{\widehat{\pi}}$ is at least $1 - O(\delta \sqrt{\log 1 / \delta})$.
In other words, our results match those of~\cite{bakshi2025learning} 
as long as $\eta = \Omega(\delta \sqrt{\log 1 / \delta})$, 
but cannot go below this threshold. 
While our error guarantee is quantitatively somewhat worse, 
a major advantage of our approach is that 
the measurements we need are simpler 
than the ones required by the algorithm of~\cite{bakshi2025learning}, 
which requires many rounds of adaptively chosen measurements.
In contrast, our algorithm only uses non-adaptively chosen, single-qubit measurements.
We believe this algorithm is mostly of conceptual interest, as it gives a completely different approach to agnostic tomography of pure states than previous methods, and we believe it can generalize to different settings beyond the ones we currently understand to date.

We now move to present our lower bounds. We start by establishing an information-computation tradeoff in the Quantum SQ model, providing rigorous evidence that the error guarantee achieved
by our algorithm of Theorem~\ref{thm:q-main-inf} is nearly best possible within the class of polynomial-time algorithms. For concreteness, we define the Quantum Statistical Query (QSQ) model, introduced in~\cite{arunachalam2020quantum}. A Quantum SQ algorithm 
is an algorithm whose objective is to learn some information about an unknown state $\rho$, 
by making adaptive calls to the following QSTAT oracle.

\begin{definition}[QSTAT Oracle] \label{def:qsq-oracle}
Let $\rho$ be an $n$-qubit mixed state.
A Quantum Statistical Query is an observable $O \in \C^{2^n} \times \C^{2^n}$ satisfying $\norm{O}_2 \leq 1$.
For $\tau > 0$, the $\QSTAT (\tau)$ oracle responds to the query $O$ with a value $v$ such that
$\left| v - \tr (O \rho) ] \right| \leq \tau$.
We call $\tau$ the tolerance of the Quantum Statistical Query.
\end{definition}
As in the classical setting, the parameter $\tau$ plays the role 
of the proxy for the copy complexity of the problem, 
and the number of calls to the QSQ oracle plays the role of the runtime of the algorithm.

\smallskip

With this setup, our QSQ lower bound is stated below (see Theorem~\ref{thm:qsq-body})
for a more detailed statement).

\begin{theorem}[QSQ Lower Bound for Agnostic Product Mixed State Tomography] \label{thm:qsq-intro}
Any Quantum SQ algorithm that learns a product mixed state $\pi$ on $n$ qubits 
to trace distance $o(\opt \log (1/ \opt) / \log \log (1 / \opt))$, 
given QSQ access to a state $\rho$ satisfying $\trd (\rho, \pi) = \opt$, 
where $\opt$ is known to the learner, 
either requires $n^{\omega (1)}$ many Quantum SQs, or must make at least one query 
of tolerance $n^{-\omega(1)}$.
\end{theorem}
In summary, any QSQ algorithm for agnostic product mixed state tomography that achieves error 
slightly smaller than the error guarantee achieved by our algorithm requires super-polynomial 
complexity. This in particular implies that even achieving any constant factor approximation to the optimal error of $\opt$ requires super-polynomial time in this model.

\vspace{-0.3cm}

\paragraph{Discussion}
While our formal lower bound result is for the specific QSQ model, 
there is a strong sense in which the computational complexity 
of agnostic product mixed state tomography is closely tied 
to the complexity of classical robust estimation.
This is because we can always embed a classical distribution over the hypercube as a diagonal mixed state, and a product distribution becomes a diagonal product mixed state.
Therefore, any better algorithm for the quantum estimation problem immediately implies an improved (quantum) algorithm for robustly learning product distributions.
Unfortunately, the literature on information-computation gaps for classical statistical tasks 
does not typically provide strong evidence of hardness against quantum algorithms\footnote{For some notable exceptions, see~\cite{bruna2021continuous,gupte2022continuous,diakonikolas2022cryptographic,tiegel2023hardness,bangachev2025near} for lattice-based cryptographic hardness.}. That said, 
for the broad class of QSQ algorithms, we formally establish that clasical SQ 
lower bounds directly translating to our setting, yielding Theorem~\ref{thm:qsq-intro}.

\medskip

Our second lower bound concerns the necessity for adaptivity in Theorem~\ref{thm:q-main-inf}. We remark that our algorithm establishing Theorem~\ref{thm:q-main-inf} 
crucially uses a single step of adaptivity to alter its measurement basis for every qubit.
We conjecture that this is in fact necessary for any efficient algorithm that 
only uses single-copy measurements. As a first step towards showing this, 
we demonstrate that fewer than sub-exponentially many non-adaptive, $2$-outcome, single-qubit measurements---like the ones considered in~\cite{chen2025information}---\emph{information-theoretically} do not suffice for this problem. Namely, we show the following (see  Theorem~\ref{thm:lb-product-basis-meas} for a more detailed statement). 

\begin{theorem}[Lower Bound against Nonadaptive Algorithms] \label{thm:quantum-na-lb-inf}
    Any algorithm that solves the agnostic tomography for product mixed states problem to non-trivial error using only non-adaptively chosen single-qubit, two-outcome projective measurements requires a sub-exponential number of copies. 
\end{theorem}

Interestingly, in contrast to the result of~\cite{chen2025information}, which only 
proved {\em computational} lower bounds for algorithms 
using these types of measurements (based on the low-degree likelihood heuristic~\cite{barak2019nearly,hopkins2018statistical,kunisky2019notes,wein2025computational}), 
our lower bound is {\em unconditional}. 
Prior to our work, the only other information theoretic lower bound of this sort was for state certification~\cite{gupta2025few}, a natural quantum testing problem; we show that such measurements are also insufficient for this natural learning task.

\subsubsection{Robustly Learning Product Distributions} \label{ssec:classical-results}

As already mentioned, 
a key ingredient of our agnostic tomography upper bounds is 
a new efficient algorithm, with near-optimal error guarantee,
for the classical task of robustly learning a binary product distribution. 
In this task, we are given samples from a (potentially arbitrary) 
distribution $p$ and the goal is to compute a product distribution $\widehat{q}$  
(i.e., a distribution whose coordinates are mutually independent) 
 whose total variation distance to $p$ is competitive 
to that of the ``best fit'' product distribution. Formally, we have the following 
definition of semi-agnostic learning of product distributions.

\begin{definition}[(Semi)-Agnostic Learning of Product Distributions] \label{def:agnostic-product-intro} 
Let $\cP_n$ be the class of product distributions over $\{0, 1\}^n$. Given access to i.i.d.\ samples from an arbitrary distribution $p$ over $\{0, 1\}^n$ and a desired accuracy $\eps>0$, the goal of the learner is to output $\widehat{q} \in \cP_n$ such that 
$\tvd(\widehat{q}, p) \leq f(\opt)+\epsilon$, where $\tvd$ is the total variation distance and 
$\opt = \opt(\cP_n) \eqdef \inf_{q \in \cP_n} \tvd(q, p)$. 
\end{definition}

The first algorithmic work in high-dimensional robust statistics~\cite{DKKLMS16} gave 
a polynomial sample and time algorithm for this task 
with error guarantee $\tilde{O}(\sqrt{\opt})+\eps$. 
Perhaps surprisingly, despite extensive work on robust statistics 
over the past decade, this error bound
had remained the best known. We note that a near-optimal error robust algorithm for this task is 
a prerequisite to obtain similarly optimal robust 
algorithms for broader models of interest, including mixtures of product distributions 
and graphical models.
In summary, we ask the following open question: 
\begin{center}
    {\it Does there exist a polynomial-time algorithm 
    for robustly learning\\ product distributions 
    with near-optimal error guarantees?}
\end{center}
As our second main algorithmic contribution, we resolve
this question in the affirmative 
(see Theorem~\ref{thm:robust-main} for a more detailed statement). 

\begin{theorem}[Computationally Efficient Semi-Agnostic Learner for Product Distributions] \label{thm:robust-main-inf}
There exists a semi-agnostic learner for binary product distributions 
with sample complexity $N = \poly (n, 1/ \eps)$, computational complexity $\poly(N)$, 
that outputs a $\widehat{q} \in \cP_n$ so that with high probability,
$\tvd(p, \widehat{q}) \leq O(\opt \log(1/\opt))+\eps$, 
where  $\opt = \opt(\cP_n) \eqdef \inf_{q \in \cP_n} \tvd(q, p)$ .
\end{theorem}

We note that our algorithm also works 
in the stronger \emph{$\eps$-corruption model} from robust statistics, 
where an $\eps$-fraction of the samples (where the value of $\eps$ is unknown to the algorithm)
are adversarially corrupted post-hoc; see Definition~\ref{def:corruption}.

Recall that the information-theoretically optimal error for 
robustly estimating a binary product in total variation distance is $O(\opt)+\eps$, 
while our algorithm achieves the weaker semi-agnostic error guarantee of 
$O(\opt \log(1/\opt))+\eps$. We provide rigorous evidence of an information-computation tradeoff, namely showing that the extra logarithmic 
factor is essentially best possible 
within the class of efficient Statistical Query (SQ) 
algorithms. 

For concreteness, we define the family of SQ algorithms below.
The Statistical Query (SQ) model~\cite{Kearns:98,FGR+13} 
considers algorithms that, 
instead of drawing individual samples from the target distribution, 
have indirect access to the distribution using the following oracle.

\begin{definition}[STAT Oracle] \label{def:stat-oracle-intro}
Let $D$ be a distribution on $\R^n$. 
A Statistical Query is a bounded function $F: \R^n \to [-1,1]$. 
For $\tau > 0$, the $\mathrm{STAT}(\tau)$ oracle responds to the query $F$ 
with a value $v$ such that $|v - \E_{X \sim D}[F(X)]| \leq \tau$. 
We call $\tau$ the tolerance of the statistical query.
\end{definition}

We note that the parameter $\tau$ is a proxy of the algorithm's
simulation sample complexity while the total number of queries 
is viewed as a measure of the algorithm's running time. 

\medskip

With this setup, we establish the following SQ lower bound (see Section~\ref{sec:sq}):

\begin{theorem}[SQ Lower Bound for Robustly Learning Binary Products] \label{thm:sq-informal}
Any SQ algorithm that learns a product distribution over $\{0, 1\}^n$, 
given SQ access to a distribution $p$ with total variation distance $\opt$ to an unknown product distribution, within total variation error $o(\opt \log(1/\opt)/\log\log(1/\opt))$, even if $\opt$ is known to the learner,  
either requires $n^{\omega(1)}$ many Statistical Queries or must make at least one query 
of tolerance $n^{-\omega(1)}$.
\end{theorem}

As an immediate corollary, it follows that 
the error guarantee achieved by our 
algorithm of Theorem~\ref{thm:robust-main-inf} 
is essentially optimal within the class of efficient 
SQ algorithms. 
An interesting conceptual implication of 
Theorem~\ref{thm:sq-informal} is that robustly learning 
a binary product 
under total variation distance is computationally harder than 
robustly learning with respect to the $\ell_2$-norm (where the 
best error rate achievable in polynomial time is $\Theta(\opt \sqrt{\log(1/\opt)})$~\cite{DKKLMS16, diakonikolas2022optimal}). 
This is in sharp contrast to the related task of robustly learning a spherical Gaussian, where the two 
notions of learning are equivalent up to constant factors. 

\subsection{Our Techniques} \label{ssec:techniques}

We now give a high-level technical overview of our results. 

\subsubsection{From Agnostic Tomography to Robust Statistics and Proof of Theorem~\ref{thm:q-main-inf}} \label{ssec:tech-reduction}

We first describe a formal reduction from
the problem of agnostic tomography 
of product mixed states to that of robustly learning a binary product distribution.
In fact, we give a \emph{black-box} reduction, i.e., 
we show how to take \emph{any} classical efficient algorithm 
that achieves non-trivial 
statistical rates for robust density estimation of a binary product distribution
and use it as a subroutine to obtain 
an efficient algorithm for agnostic tomography 
{\em with the same error guarantees}, within a constant factor.
We do this in two steps. 

Our first step is motivated by the following observation: 
if we measure a product state $\pi = \pi_1 \otimes \ldots \otimes \pi_n$ 
in any Pauli basis---i.e., we measure each qubit using the POVM 
$\{\tfrac{I + P}{2}, \tfrac{I - P}{2}\}^{\otimes n}$ 
for $P \in \{X, Y, Z\}$---then the resulting distribution 
is a binary product distribution whose mean allows 
us to recover the Bloch coefficients of $\pi_1, \ldots, \pi_n$.
For instance, if the state $\pi$ was diagonal and 
we measured in the computational basis, the resulting outcome 
would be a sample from a binary product distribution 
whose mean exactly specifies the diagonal entries of $\pi$.
But since we are measuring a state $\rho$ that has trace distance 
at most $\opt$ from some product mixed state, 
when we measure in this Pauli basis, we obtain samples from a classical distribution that has total variation distance at most $\eps$ 
from this binary product distribution (Lemma~\ref{lemma:tvd-leq-trd}).
Therefore, running a classical robust mean estimation algorithm allows us to achieve a fairly high quality approximation of the Bloch coefficients of the best product mixed state approximation!

Unfortunately, this step alone is insufficient for the following reason. 
At this stage, the best guarantee that any robust mean estimation algorithm can 
provide is an approximation 
$\widehat{\pi} = \widehat{\pi}_1 \otimes \ldots \otimes \widehat{\pi}_n$ 
such that $\sum_{i = 1}^n \norm{\pi_i - \widehat{\pi}_i}_F$ is small, 
where $\norm{\cdot}_F$ denotes the Frobenius norm.
However, if some of the $\pi_i$'s are close to pure, 
this sort of approximation is insufficient to ensure any 
nontrivial bound in trace distance---the natural and standard notion 
of distance in agnostic tomography.
Along such nearly-pure qubits, it turns out that one needs to 
learn to good \emph{relative} error.
This issue is a quantum manifestation of
the main difficulty in robustly learning binary product distributions in total variation distance (in the classical setting), 
where the key technical challenge arises from coordinates 
whose true means $p_i$ are very close to $0$ or $1$.

However, not all hope is lost. This is because while this initial approximation 
is insufficient for learning the $\pi_i$'s, 
we demonstrate (see Lemma~\ref{lemma:F-rho-rhodiag-taylor}) 
that it does yield a sufficiently high quality approximation 
to the \emph{eigenvectors} of each $\pi_i$.
Specifically, we show that the best approximation, $\pi$, 
is approximately diagonal in the product basis formed by these eigenvectors.
Therefore, it suffices to learn the measurement outcomes 
of $\rho$ when we measure in this learned basis!
Since the measurement outcome distribution is $\opt$-close in total variation distance 
to a binary product distribution that would exactly determine the coefficients 
of $\pi$ in the same basis, it suffices to do a second round 
of robust estimation (in total variation distance) 
to compute the best product approximation in this learned basis. This second round concludes our black-box reduction.

\subsubsection{Semi-Agnostic Tomography of Pure Product States: Proof of Theorem~\ref{thm:pure-intro}} \label{ssec:tech-pure}

In the special case
of 
pure product state approximation to our unknown state, 
we demonstrate that a slight---arguably even simpler---variant of the aforementioned 
also yields a black-box reduction to robustly learning a binary product distribution.
The key insight is that, in the previous reduction, we already demonstrated 
that by measuring in the Pauli basis and using robust estimation to robustly learn the mean of the resulting distribution, we can identify the eigenvectors 
of the best mixed product state approximation to sufficiently good accuracy.
The same property remains true if the product state is pure.
Previously, we then had to perform a second round of estimation 
to learn the eigenvalues of the mixed state.
But now, if the state is pure, we do not need to estimate the eigenvalues!
Instead, we show that it suffices to simply take the qubit-wise estimated mixed states, 
and round them to be pure states. It turns out that this will only incur a constant factor loss in trace distance.

Crucially, this reduction only needs the first round of robust estimation 
(used in the previous reduction).
This presents two conceptual advantages over the previous reduction.
First, now we only need an $\ell_2$-accurate estimate of the mean of the product distribution, 
which is a much simpler robust estimation task. In Appendix~\ref{app:sample-opt-l2}, 
we show that this task can be solved with slightly better accuracy (namely, 
$O(\opt \sqrt{\log 1 / \opt})$ rather than $O(\opt \log 1 / \opt)$), 
and with a {\em nearly-linear} sample complexity.
Consequently, our semi-agnostic tomography algorithm for pure product states 
achieves better accuracy and copy complexity as well.
Second, our measurements can be chosen {\em fully non-adaptively}, whereas
the measurements for the previous setting 
are chosen adaptively---and, indeed, as our lower bound (Theorem~\ref{thm:quantum-na-lb-inf}) shows  
{\em must} be chosen adaptively.

Although Theorem~\ref{thm:pure-intro} follows as a simple corollary of our more general reduction, we 
believe it is conceptually significant. Specifically, it shows that techniques from robust estimation can 
also be used to derive new algorithmic results for agnostic tomography of pure states. This naturally 
leads to the question of whether analogous methods can be used to obtain interesting guarantees for 
agnostic tomography of stabilizer states.

\subsubsection{Quantum SQ Lower Bound: Proof of Theorem~\ref{thm:qsq-intro}} \label{ssec:tech-qsq}

We now turn to our QSQ lower bound for agnostic tomography of product mixed states. 
As discussed earlier, there is a straightforward but important reduction in the reverse direction, 
showing that learning product distributions over the hypercube reduces to agnostic tomography 
of product mixed states via an embedding into diagonal density matrices. 
Because all relevant states in this reduction are diagonal, 
it suffices to consider diagonal measurements as well. 
It follows that any QSQ algorithm immediately induces an SQ algorithm 
for the corresponding classical learning problem, 
by viewing each diagonal measurement as a function on the hypercube in the standard way.
This allows us to essentially transfer the SQ lower bound of Theorem~\ref{thm:sq-informal} 
to the quantum setting. 

\subsubsection{Lower Bound for Non-Adaptive Agnostic Tomography: Proof of Theorem~\ref{thm:quantum-na-lb-inf}} \label{ssec:tech-quantum-lb}
We now describe our lower bound against non-adaptive, single-qubit measurements that are two-element projection-valued measures (PVMs).
The high level intuition in our proof of Theorem~\ref{thm:quantum-na-lb-inf} is that for such a measurement to succeed, it must effectively guess the dominant eigenvector in all qubits where the best product mixed state approximation is very close to pure. However, as long as we take this direction to be random, this event is
very unlikely.

The key idea is to embed a \emph{moment-matching} construction into the product mixed state tomography problem. In Proposition~\ref{prop:main-lower-bound-prop}, we show that there exists an ensemble of pairs of $n$-qubit mixed states that are each close to product mixed states, constantly far away from each other in trace distance (see Lemma~\ref{lemma:lb-states}), yet have random eigenvectors for each qubit with the distributions over their eigenvalues matching many moments (see Lemma~\ref{lemma:lb-moment-match}).
The point is that for any fixed measurement of consideration, the measurement will not align with the true eigenvector in all but a small fraction of the qubits, which we show can be ignored via Lemma~\ref{lemma:lb-bad}. In Lemma~\ref{lemma:lb-good} and the proof of Proposition~\ref{prop:main-lower-bound-prop}, we proceed to show that this failure in alignment in the remaining qubits causes the likelihood of any measurement outcome to be close to a low-degree polynomial in the eigenvalues. However, since the moments of the distributions of the eigenvalues match, this shows that the distribution of the measurement outcomes under the two product states is statistically indistinguishable.

\subsubsection{Near-Optimal Robust Learner for Product Distributions: Proof of Theorem~\ref{thm:robust-main-inf}}

\label{ssec:tech-product-alg}

As mentioned in the earlier discussion, 
\cite{DKKLMS16} gave an algorithm for robustly learning product distributions that 
achieves total variation error $\tilde{O}(\sqrt{\opt})$. 
Whether this error bound can be improved to a near-optimal 
bound of $\tilde{O}(\opt)$ has remained a basic open 
question in robust statistics 
that we fully resolve here. To 
achieve this, we need to develop several new technical 
ingredients that we summarize in the subsequent discussion.

\vspace{-0.3cm}

\paragraph{New TV Distance Characterization between Products}
The first fundamental obstacle 
lies in developing a tighter characterization of the total variation 
distance between two product distributions. 
Such a step is necessary but not sufficient, as
once we have such a characterization, we need a 
method to exploit it algorithmically.   
In particular, suppose that our algorithm approximates 
a product $p$ by another product $q$. Then, 
we will need our analysis to {\em certify} 
that the total variation distance between them is small. 
To set up notation, let $\mu$ be the mean of $p$ 
and $\nu$ be the mean of $q$. 
By flipping any coordinates with mean close to $1$, we will assume without loss of generality throughout that each coordinate 
$\mu_i$ and $\nu_i$ is bounded away from $1$ 
for all $i$. 
A simple and standard way to bound 
the total variation distance between $p$ and $q$ is via
the Hellinger distance:
$\Theta(\sqrt{\sum_i |\mu_i - \nu_i|^2 / (\nu_i+ \mu_i)})$. 
In fact, the prior algorithm of~\cite{DKKLMS16} relies on this upper bound to achieve the weaker error guarantee of $\tilde{O}(\sqrt{\opt})$.
Unfortunately, just relying on the Hellinger distance cannot provide a better guarantee 
due to an {\em integrality gap}. 
In particular, there exist pairs of product 
distributions $p$, $q$ with total variation distance $\opt$, whose Hellinger distance is on the order of $\sqrt{\opt}$. 
Since $q$ can be thought of as a version of $p$ 
with $\opt$-corruption, 
an algorithm observing samples from $p$ 
could not tell whether $p$ or $q$ 
is the ``true'' distribution. Thus, we
cannot hope to robustly learn 
to Hellinger distance better than about $\sqrt{\opt}$.

To circumvent this obstacle, we develop 
a new measure tightly characterizing 
the total variation distance between two product distributions in terms of their means (Theorem~\ref{thm:product-tv-bound}). 
Intuitively, our new measure smoothly interpolates between 
the $\ell_1$-distance and a $\chi^2$-divergence-type object, 
allowing us to tightly witness the contribution 
to the total variation distance both on balanced and unbalanced coordinates.
We view this as a basic structural result 
of broader applicability. We start by observing 
that there is another useful upper bound on the total variation distance between products, 
given by $\|\mu - \nu\|_1$, 
the $\ell_1$-norm between the mean vectors. 
In fact, we can combine this bound with the Hellinger bound 
to show the following: 
for any subset $A \subseteq [n]$ of coordinates, 
the total variation distance between $p$ and $q$ 
is bounded above by 
$O\left(\sum_{i \in A}  |\mu_i - \nu_i| + \sqrt{\sum_{i \not\in A}  |\mu_i -\nu_i|^2/\mu_i}\right)$.
It turns out that {\em this} characterization 
is nearly tight. 
In particular, we define a convex body $\cT_\mu$ and corresponding dual norm $\|\mu-\nu\|_{\mu}$ along these lines (Definition~\ref{def:norm}) and prove that it bounds $\tvd(p,q)$ from above.

\vspace{-0.3cm}

\paragraph{Filtering Algorithm via Novel Convex Relaxation}
Armed with this characterization, 
we can attempt a filtering style algorithm~\cite{DKKLMS16,dong2019quantum,diakonikolas2019recent,diakonikolas2023algorithmic} to solve our problem.
In more detail, given a sample set $S$, let $\nu$ 
be the empirical mean 
and let $q$ be the corresponding product distribution. 
We would like to show that the true distribution $p$ satisfies $\tvd(p,q) = \tilde{O}(\opt)$. 
By our above characterization, 
it is sufficient to show that $x \cdot (\mu-\nu) = \tilde{O}(\opt)$ for every $x \in \cT_\mu$ . 
On the other hand, it is not hard to show 
that for any such $x$ the distribution of $x \cdot p$ 
is tightly concentrated (Lemma~\ref{lem:quadratic-tail-bound}). 
This means that if a small fraction of outliers 
change the mean by a lot, 
it will cause the variance of $x \cdot p$ to be substantially more than expected. 
This scenario is detectable, as follows: 
if we let $\Sigma_0$ be the empirical covariance matrix (with the diagonal zeroed out to account for the anticipated covariance), 
then we only need to worry about points $x$ 
with $\|x\|_{\nu}^* \leq 1$ and $x^{\top}(\Sigma_0) x$ large. 
Furthermore, if we can find such an $x$ we can find samples 
that are most likely outliers in the $x$-direction 
and remove them via an appropriate filtering step.

The catch is that the method we just described 
does not lead to a polynomial-time algorithm, as the computational problem 
of maximizing $x^{\top} A x$ 
subject to $\|x\|_{\nu}^* = 1$ is computationally intractable in general. To obtain a polynomial-time algorithm, we instead consider a natural convex relaxation of this objective (\Cref{eq:st-mu}).
Note that the above is equivalent 
to maximizing $\tr(A(xx^{\top}))$. 
We relax this objective by replacing 
$xx^{\top}$ by any positive semi-definite matrix $H$ that 
has $\|H\|_{\nu}^* \leq 1$, 
where $\|H\|_{\nu}^*$ is obtained 
from $H$ by first taking the $\|\cdot \|_{\nu}^*$ 
norm of each row 
and then taking the $\|\cdot\|_{\nu}^*$ norm of the resulting vector. 
It turns out that this convex relaxation can be optimized 
efficiently. Moreover, we can show 
that for any such $H$, $p^{\top}Hp$ 
is tightly concentrated about its mean. 
This allows us to build a somewhat more complicated 
filtering algorithm that suffices for our purposes.

To establish the correctness of our algorithm, 
we additionally need to establish 
tight tail bounds for the types of quadratic polynomials 
encountered by our filtering procedure, which  
we handle
by leveraging certain decoupling lemmas from 
the Boolean analysis literature. 
An additional hurdle is that one must first 
perform a number of pre-processing steps 
to ensure that the product distribution in question 
is of the right form. 

\vspace{-0.3cm}

\paragraph{New Tools for Stability and Near-optimal Sample Complexity}
Finally, it is worth highlighting an additional  
technical aspect of our contributions, leading 
to a nearly tight sample complexity bound. 
In order for our analysis to work, 
we need that with high probability 
our sample set $S$ satisfies the following {\em stability condition}:  
For every subset $S'\subseteq S$ of size $(1-\eps)|S|$, 
if $\Sigma_0(S')$ is the empirical expectation 
of $(x-\mu)(x-\mu)^{\top}$ 
with zeroed-out diagonal, then for all $H$ 
satisfying the desired properties (outlined above) 
we have $|\tr(H \Sigma_0(S'))| = \tilde{O}(\eps)$. 
Prior work in robust statistics typically establishes such stability 
conditions by finding an approximate 
cover of such $H$ and proving high probability bounds 
for each. While this approach is feasible in our 
setting, it would inherently lead to a sample complexity 
bound of $\Omega(n^4)$,
due to the relatively weak 
tail bounds on $p^{\top}Hp$
. While such a sample bound is 
still polynomial, it is unsatisfying as it would lead to a 
highly impractical algorithm. Instead, 
we develop a new technique for analyzing the sample complexity 
that leads to a bound of $\tilde{O}(n^2/\epsilon^2)$ and
may be applicable in other contexts.

Our new analysis works in two steps:
To prove an upper bound on $\tr(H \Sigma_0(S'))$, 
we first show that since $\Sigma_0(S)$ is likely 
close to $0$, 
we can show that  $|\tr(H \Sigma_0(S))|$ 
is likely small for all $H$. 
Since removing elements from $S$ 
cannot make $\Sigma_0(S)$ much larger, 
this proves the lower bound. 
For the upper bound, we use the VC-inequality 
to show that with high probability over $S$ 
the empirical distribution of $p^{\top}Hp$ 
is close to the true distribution 
in Kolmogorov distance for all $H$. 
This suffices to show that the empirical expectation 
of $\tr(H(x-\mu)(x-\mu)^T)$ 
with its $\epsilon$-tails removed, 
which is the smallest that  $\tr(H \Sigma_0(S'))$ could be,  
is close to the population average 
with the same truncated tails.

We put all these pieces together to complete the proof of Theorem~\ref{thm:robust-main-inf} 
in Section~\ref{sec:robust-product}.

\subsubsection{SQ Lower Bound for Robustly Learning Binary Products: Proof of Theorem~\ref{thm:sq-informal}} \label{ssec:tech-sq}

Recall that Theorem~\ref{thm:sq-informal} establishes 
an SQ lower bound 
providing rigorous evidence that no polynomial-time classical algorithm 
can obtain total variation distance error 
$\delta := o(\opt \log(1/\opt)/\log\log(1/\opt))$
for robust learning of binary products. 

As is standard~\cite{FGR+13}, establishing an SQ lower bound in our setting 
essentially boils down to constructing large families of 
$\opt$-corrupted 
binary product distributions that have pairwise small correlation, 
i.e., $\chi^2$-squared inner product,  
with respect to some given base distribution. To achieve this,  
we construct a pair of product distributions $D$ and $D_0$,
supported in a lower dimensional space, 
that are $\delta$-far from each other in total variation distance, 
but for which an $\opt$-corrupted version $D'$ (in total variation distance)
of $D$ matches many low-degree moments with $D_0$. 
We then embed this instance into a higher dimensional product distribution 
(with marginal probabilities agreeing with those in $D_0$) 
and roughly speaking show that it is SQ-hard to find these hidden coordinates.
This high-level approach to establish SQ lower bound has been leveraged 
in prior work 
for Gaussian-like settings~\cite{DKS17-sq, DKRS23} and for a discrete setting~\cite{diakonikolas2022optimal},   
as the one we consider here. These similarities notwithstanding, 
there are two major aspects that require substantial conceptual and technical innovations.

The first issue concerns the choice of reference distribution $D_0$. 
A generic SQ lower bound of the aforedescribed form 
is only known when $D_0$ is the {\em uniform distribution} on the hypercube~\cite{diakonikolas2022optimal}. 
Indeed, the uniform distribution was used as the reference distribution 
in~\cite{diakonikolas2022optimal} to prove a super-polynomial 
SQ lower bound for robustly learning binary products to error $o(\opt \sqrt{\log(1/\opt)})$ 
{\em with respect to the $\ell_2$-norm}. This prior SQ lower bound already implies a similar SQ lower bound
under the total variation distance, since the total variation distance between two binary products is at least 
a constant multiple of their $\ell_2$ distance. Interestingly,~\cite{DKKLMS16} gave an efficient algorithm
matching this $\ell_2$ distance bound. Hence, to prove our near-optimal SQ lower bound under the total variation distance, 
a new approach is needed.

This obstacle can be circumvented by selecting $D_0$ to be a ``highly unbalanced'' 
product distribution. This requires a {\em new generic} SQ lower bound result which translates 
moment matching to SQ hardness. We establish such a result (Proposition~\ref{prop:generic-SQ}) 
for any reference distribution $D_0$ (independent of the bias of its coordinates)
via Fourier-analytic ideas. While our Fourier approach 
can be viewed as a generalization of the analogous analysis in~\cite{diakonikolas2022optimal} 
for the special case of the uniform distribution, the fact that the approach can be carried through
for any biased product {\em with no dependence on the bias} is a notable new result.

The second issue has to do with our moment-matching construction in low dimensions.
This is the technically most novel and challenging aspect of our proof. 
In more detail, to carry out our overall strategy, 
we need to show that there is a distribution $A$ on $\{0, 1\}^m$---for a carefully selected 
value of $m \ll n$---that is $\opt$-close in total variation distance to $U_{(1+\delta)/m}^m$ 
yet also matches 
its first $\omega(1)$ many moments with $U_{1/m}^m$, where $U_{b}^m$ denotes the binary product distribution on $\{0,1\}^m$ with bias $1/m$ on each coordinate (Proposition~\ref{prop:final-A}).  
As both these products are symmetric, it suffices to consider the distribution over weights 
(i.e., the one-dimensional distributions corresponding to the sums of the coordinates). 
Related to this, note that the $k$th moments of $\Bin(m,(1+\delta)/m)$ and $\Bin(m,1/m)$ 
differ by $O_k(\delta)$. Furthermore, as both distributions have pmf values 
at least $\opt^{1/2}$ over $\{0,1,\ldots,T\}$ for some $T \gg \log(1/\opt)/\log\log(1/\opt)$,  
we will attempt to modify the distribution $\Bin(m,(1+\delta)/m)$ on $\{0,1,\ldots,T\}$, 
so that it matches its $k$ low-degree moments with $\Bin(m,1/m)$. 

Recalling that the pmf values of our two products are non-trivially large on $\{0,1,\ldots,T\}$, 
we should expect that we will only need to change the distribution by roughly 
$O_k(\delta/T) \ll \opt$ in order to match $k$ moments. Now,   
if these were {\em continuous} distributions supported on $[0,T]$, 
established techniques (see, e.g., Chapter 8 of~\cite{diakonikolas2023algorithmic}) 
should suffice to deal with this. 
Unfortunately, in our discrete setting, we have to modify our original distribution 
to match moments {\em while only changing its density at integer values}. 
Intuitively speaking, if the parameter $T$ is large enough relative to $k$, 
these integers are sufficiently finely spaced that this restriction should not matter. 
Formally, we can prove this statement (Proposition~\ref{prop:int-mm}) by formulating 
the moment-matching requirements as a linear programming problem 
and comparing {\em the dual versions} of the integer and continuous versions. 
At a high-level, we can show that a feasible solution can be transferred from one to the other, 
so long as for any low-degree polynomial $p$, the maximum of $|p|$ on $[0,T]$ 
is not too far from its maximum on $\{0,1,\ldots,T\}$. 
Fortunately, this statement (Lemma~\ref{lem:disc-cont}) can be shown by leveraging 
a technical result in~\cite{KaneKP17}. 

\vspace{-0.2cm}

\subsection{Discussion and Open Problems} \label{ssec:open}

Our work leaves open a number of interesting directions for future research. 
The connection we uncover between agnostic tomography and robust statistics 
appears quite promising, and we expect that it extends well beyond the settings studied here. 
On the one hand, it is plausible that tools from robust statistics can be used to obtain 
(semi-)agnostic guarantees for many other physically relevant classes of quantum states. 
On the other hand, the quantum estimation problems considered here 
may in turn motivate the development of new robust estimation algorithms 
in regimes that have not previously been studied in the classical literature. 
We conclude by highlighting a few concrete open problems.

First, our algorithm applies in the regime where the optimal error 
is relatively small, namely when $\opt$ is at most a sufficiently small constant. 
It is an interesting open question whether these results can be extended 
to the more general setting in which $\opt$ is close to $1$. 
We conjecture that such an extension may be possible 
using techniques from classical \emph{list learning}~\cite{charikar2017learning,DiakonikolasKS18}.

Second, it is natural to ask whether our results can be extended 
to other classes of mean-field approximations, 
for example those arising in fermionic or bosonic systems, 
i.e., agnostic versions of algorithms such as~\cite{aaronson2021efficient,bittel2025optimal} in the fermionic setting, or~\cite{mele2025learning,bittel2025energy,bittel2025optimal,chen2026towards} in the bosonic setting, among many others. 
Progress in this direction would likely require new robust estimation algorithms 
for distribution families such as determinantal point processes~\cite{kulesza2012determinantal} and their less well-understood bosonic analogues.

Finally, as discussed in the introduction, mixed-state ansatz classes 
arise naturally in quantum estimation problems involving thermal states. 
Recent work has given new non-agnostic algorithms for this problem for thermal states 
of structured classes of Hamiltonians, such as geometrically local Hamiltonians~\cite{chen2025learning}.
Given that these ansatz are unlikely to be exact in most realistic settings, it would be particularly interesting if one obtain similar results in the agnostic regime.

\vspace{-0.2cm}

\subsection{Related Work}
\label{sec:related-work}

\paragraph{Independent Work} Prior to the dissemination of this work, we were made aware of independent work 
\cite{aliakbarpour2025adversarially} which draws a similar conceptual connection to the one we make here between quantum learning with outliers and robust statistics.
In~\cite{aliakbarpour2025adversarially}, they consider the problem of learning an arbitrary $n$-qubit mixed state with single copy measurements, but where an $\epsilon$-fraction of these measurements are potentially corrupted.
They demonstrate an inefficient algorithm which achieves error $O(2^{n / 2} \eps)$ with non-adaptive measurements and 
show that this error is optimal for algorithms with non-adaptive measurements.
In contrast, we demonstrate efficient algorithms for agnostic tomography of $n$-qubit product mixed states that do not suffer {\em any} dimension-dependent loss, but which are necessarily adaptive via an
$\exp (n)$ sample complexity lower bound for $2$-outcome, non-adaptive measurements. Beyond the conceptual connection to robust statistics,
these works use entirely different techniques at the technical level.
We view these contributions as complementary to each other: 
the result of \cite{aliakbarpour2025adversarially} demonstrates that without any structure on the mixed state, agnostic tomography is information-theoretically hard. 
In contrast, we show that under natural structural assumptions, we can circumvent these lower bounds and obtain dimension-independent error in polynomial time.

\vspace{-0.3cm}

\paragraph{Agnostic Tomography}
Agnostic tomography was introduced by~\cite{grewal2024agnostic}, although qualitatively similar notions were considered previously by~\cite{buadescu2021improved} and in the PAC learning setting by
\cite{anshu2024survey}.
Subsequently, efficient algorithms for agnostic tomography were developed for product states~\cite{bakshi2025learning} and stabilizer states~\cite{chen2025stabilizer}. 
Prior to our work, no efficient agnostic tomography algorithms were known for any class of 
mixed state ansatz.

\vspace{-0.3cm}

\paragraph{Robust Statistics}
In a range of machine learning 
scenarios, the standard i.i.d.\ assumption 
does not accurately represent the underlying 
phenomenon. To address such settings, robust statistics \cite{Huber09,diakonikolas2023algorithmic} 
aims to develop accurate 
estimators in the presence of adversarial outliers 
or model misspecification. 
The field originates from the pioneering 
works of Tukey and Huber \cite{tukey1960survey, Huber64} in the 1960s. 
Early work in statistics determined the sample complexity of 
robust estimation for various basic tasks, 
including mean estimation. 
Alas, the multivariate versions of these estimators incurred exponential 
runtime in the dimension. A recent line of work in computer science, starting 
with \cite{DKKLMS16, LaiRV16}, has led to a revival of robust statistics from 
an algorithmic standpoint, 
providing the first robust estimators 
in high dimensions with polynomial sample and time complexity. Since the dissemination of these works, 
there has been an explosion of results 
providing computationally 
efficient robust estimators and associated statistical-computational 
tradeoffs for a wide range of tasks. 
See \cite{diakonikolas2023algorithmic} for a textbook overview of this field. 

The task of robustly learning binary product distributions was one of the 
first problems studied in the field. Specifically,~\cite{DKKLMS16} 
gave an efficient 
algorithm that approximates the underlying distribution within error 
$\tilde{O}(\sqrt{\opt})$ in total variation distance. While one can achieve $O(\opt)$ error information-theoretically,
known Statistical Query lower bounds~\cite{diakonikolas2022optimal} 
rule out efficient algorithms with error better than 
$\Omega(\opt \sqrt{\log(1/\opt)})$. 
This near-quadratic gap between the known upper and lower 
bounds has remained a basic open question in the field. 
Our \Cref{thm:robust-main-inf} 
gives an efficient algorithm which matches the SQ lower bound 
up to a $\log\log(1/\opt)$ factor.

\vspace{-0.2cm}

\subsection{Organization} 
In Section~\ref{sec:bg}, we provide the necessary technical background. 
Section~\ref{sec:reduction} presents our efficient reduction from agnostic tomography to robust estimation. Section~\ref{sec:robust-product} gives 
our near-optimal efficient algorithm for robustly learning binary products, thereby establishing Theorem~\ref{thm:robust-main-inf}. 
Theorem~\ref{thm:q-main-inf} follows by combining the results of these two sections. 
Section~\ref{sec:na-lb} proves our information-theoretic lower bound (Theorem~\ref{thm:quantum-na-lb-inf}) for quantum tomography with non-adaptive measurements. 
Section~\ref{sec:sq} establishes our super-polynomial SQ lower bound, showing that the error guarantee of our robust product learner is essentially the best possible. 
Finally, Section~\ref{sec:qsq} establishes our quantum SQ hardness for agnostic tomography of product mixed states.

\section{Preliminaries}\label{sec:bg}

\paragraph{Notation}
We use the notation $f \lesssim g$ to indicate that $f \leq C g$ 
for some universal constant $C$.
For two $n$-qubit mixed states $\rho,\sigma\in \C^{2^n\times 2^n}$, we let 
$\trd (\rho, \sigma) = \tfrac{1}{2} \norm{\rho - \sigma}_1$ 
denote the trace distance and
$F(\rho, \sigma) = \tr \left(\sqrt{ \sqrt{\rho} \sigma  \sqrt{\rho} } \right)^2$ 
denote the fidelity between the two states.
For two classical distributions $p, q$, we use 
$\tvd (p, q) = \frac{1}{2} \norm{p - q}_1$ to denote the
total variation distance between them.

For a distribution $D$ and a function $f$, 
we let $\E [f(D)] = \E_{X \sim D} [f(X)]$; 
for a multiset $S$, we let $\E [f(S)]$ denote the expectation of $f$ 
over the uniform distribution of the points in $S$.
We also let $\mu(S) = \E [S]$ denote the empirical mean of $S$.

\paragraph{(Semi-)Agnostic Tomography}

Formally, we study the following problem.

\begin{problem}[Agnostic Learning for Product Mixed States]\label{prob:agnostic-pmix}
    Let $\cM_n=\{\pi_1 \otimes \cdots\otimes \pi_n:\pi_i\in \C^{2\times 2}\}$ 
    be the family of product mixed states over $n$ qubits. 
    Given copies of an arbitrary quantum state $\rho\in \C^{2^n\times 2^n}$
    output $\ket{\hat{\pi}} \in \cM_n$ such that with probability $1 - \delta$, we have
    \[
    \trd (\ketbra{\hat{\pi}}{\hat{\pi}}, \rho) \leq f (\opt) + \eps \; ,
    \]
    where $\opt = \inf_{\ket{\pi} \in \cM_n} \trd (\ketbra{\pi}{\pi}, \rho)$.
\end{problem}

\noindent A couple of remarks are in order: First, as is standard in this literature, 
the desired accuracy guarantee is measured with respect to the trace distance. 
Second, the function $f: \R_+ \mapsto \R_+$ quantifying the final error ought to 
satisfy $\lim_{t \rightarrow 0}f(t)=0$. While information-theoretically one can 
achieve $f(\eps) = O(\eps)$ with polynomial number of copies 
(e.g., via shadow tomography~\cite{buadescu2021improved}), 
no non-trivial error guarantee was previously known for polynomial-time algorithms.

\paragraph{Robust Statistics}
We now record the basic setup of clasical robust statistics.
We will restrict ourselves to the background that is necessary for this work.
The interested reader is referred to~\cite{li2018principled,diakonikolas2019recent,diakonikolas2023algorithmic} 
for an in-depth treatment of the topic.

We first recall the standard contamination model of $\eps$-corruption 
from robust statistics:
\begin{definition}[$\eps$-corruption,~\cite{DKKLMS16}] 
\label{def:corruption}
    We say that a multi-set $S$ of $n$ points is an \emph{$\eps$-corrupted} set of samples from a distribution $p$ if we can write $S = S_g \cup S_b \setminus S_r$, where:
    \begin{itemize}
        \item $S_g$ is a set of $n$ i.i.d. samples from $p$,
        \item $S_r \subset S_g$, and $|S_r| = |S_b| = \eps n$. 
    \end{itemize}
\end{definition}
The above contamination model is closely related to the more traditional 
statistical notion of gross (total variation distance) corruption:
\begin{definition}[$\eps$-general, non-adaptive contamination]
   We say that a set $S$ of $n$ points is an \emph{$\eps$-contaminated} 
   set of samples from a distribution $p$ if $S$ consists of 
   $n$ i.i.d.\ samples from some distribution $q$ satisfying $\tvd(p, q) \leq \eps$.
\end{definition}

The following standard fact (see, e.g.,~\cite{DKKLMS16}) 
relates the two models:
\begin{fact}
\label{fact:contamination-to-corruption}
    Let $S$ be an $\eps$-contaminated set of samples from $p$.
    Then, for any $c > 0$, with probability $1 - \exp (- O(c \eps n))$, we have that $S$ is an $(1 + c) \eps$-corrupted set of samples from $p$.
\end{fact}
In other words, up to sub-constant factors in $\eps$ 
(which do not affect our guarantees), 
the setting of $\eps$-corruption is strictly more general 
than general non-adaptive contamination.

The class of distributions we will be concerned with is the 
set of product distributions over the binary, $n$-dimensional hypercube $\{0,1\}^n$.
Denote the set of such distributions $\cP_n$.
Note that any such distribution is determined by its mean vector.
Consequently, there are two natural choices for estimands: 
the mean and the underlying density.
\begin{problem}[Robust Mean Estimation for Binary Product Distributions]\label{prob:robust-mean}
    Given $N$ $\eps$-contaminated samples from a distribution $p \in \cP_n$ 
    with mean $\mu$, output an estimate $\hat{\mu}\in \R^n$ such 
    that $\|\hat{\mu}-\mu\|_2 \leq f_{\rm mean}(\epsilon)$ 
    with probability $1 - \delta$.
\end{problem}

\begin{problem}[Robust Density Estimation for Binary Product Distributions]\label{prob:robust-density}
Given $N$ $\eps$-contaminated samples from a distribution $p \in \cP_n$, 
a binary product distribution $\hat{p}$ such that $\tvd(\hat{p},p)\leq f_{\rm density}(\epsilon)$ with probability $1 - \delta$. 
\end{problem}

It is worth mentioning that the total variation distance guarantee in 
Problem~\ref{prob:robust-density} is stronger than the $\ell_2$ mean estimation 
guarantee in Problem~\ref{prob:robust-mean}. Indeed, since the total variation 
distance between two binary products is at least proportional to the $\ell_2$-distance between their means, any algorithm for Problem~\ref{prob:robust-density} 
immediately gives an algorithm for Problem~\ref{prob:robust-mean} with essentially 
the same error. Unfortunately, the other direction does not hold, as 
the $\ell_2$-distance between the means of two binary products can be much smaller than their total variation distance, particularly when the distribution is unbalanced.

Here we are interested in developing an algorithm for 
Problem~\ref{prob:robust-density} that incurs polynomial sample and computational 
complexity, and yields a near-optimal error rate. The best known algorithmic guarantee from \citep{DKKLMS16} achieves error 
$f_{\rm density}(\epsilon)=O( \sqrt{\epsilon \log(1/\epsilon)})$. On the lower bound side, the SQ lower bound of \cite{diakonikolas2022optimal} still applies. Consequently, prior to this work, there was a near-quadratic gap between the best known upper and lower bounds for Problem~\ref{prob:robust-density}. 

\paragraph{Robust Hypothesis Selection}

We will require the following robust hypothesis selection routine \cite{DKKLMS16}:
\begin{lemma}[Robust Hypothesis Selection]
\label{lem:robust-hypothesis-selection}
    Let $\mathcal{C}$ be a class of distributions, and let $\mathcal{M}$ be a finite set of distributions.
    Suppose that for some $N$ and $\eps > 0$, there is an algorithm which, given a set of $\eps$-corrupted samples from a distribution $p$ of size $N$, outputs in time $T_{\mathrm{alg}}$ a list of $M$ distributions $q_1, \ldots, q_L$ so that all the $q_i$ can be sampled from in time $T_{\mathrm{samp}}$.
    Suppose further that these $q_i \in \mathcal{M}$ for all $i = 1, \ldots, L$ deterministically, and that $\tvd (q_i, p) \leq \eps$ for some $i = 1, \ldots, L$.
    Then, there is an algorithm which takes $O(N + \tfrac{\log |\mathcal{M}| + \log 1 / \delta}{\eps^2})$ samples, and which outputs $q_j$ so that $\tvd (q_j, p) \leq O(\eps)$.
\end{lemma}

\begin{remark}
As is standard, Lemma~\ref{lem:robust-hypothesis-selection} allows us to efficiently reduce to the case where the fraction of outliers $\eps$ (in Definition~\ref{def:corruption}), 
or the value of $\opt$ (denoting the total variation distance of the best fit product distribution   
in the formalism of Definition~\ref{def:agnostic-product-intro}), is known to the algorithm. 
See Section~\ref{ssec:robust-prep} for a formal proof of this translation.
\end{remark}

\section{Reduction from Agnostic Tomography to Robust Estimation} \label{sec:reduction}
Our main result in this section is the following efficient reduction:
\begin{theorem}\label{thm:reduction}
    Given algorithms for Problems~\ref{prob:robust-mean} and \ref{prob:robust-density} that run in time $T_{\rm mean}$ and $T_{\rm density}$, achieve error rates $f_{\rm mean}$ and $f_{\rm density}$, have sample complexities $N_{\rm mean}$ and $N_{\rm density}$, and failure probabilities $\delta / 2$, there exists an algorithm for Problem~\ref{prob:agnostic-pmix} that with probability $1 - \delta$ achieves error
    \[
    f(\opt)\leq \sqrt 6 \; f_{\rm mean}(\opt)+f_{\rm density}(\opt) \; .
    \]
    Moreover, the algorithm uses $3N_{\rm mean} + N_{\rm density}$ single-copy single-qubit measurements and runs in time $T_{\rm mean} +T_{\rm density}$.
\end{theorem}

\begin{remark} \label{rem:only-density}
{\em Since Problem~\ref{prob:robust-density} subsumes Problem~\ref{prob:robust-mean}, 
we could have alternatively simply used an algorithm 
for Problem~\ref{prob:robust-density}, incurring a final error 
$f(\opt)\leq (1+\sqrt 6)\; f_{\rm density}(\opt)$} (within 
a factor of $2$ of the above guarantee). We chose to phrase Theorem~\ref{thm:reduction} in this way to illustrate our two-phase reduction approach (see the pseudo-code given in Algorithm~\ref{alg:reduction}).
\end{remark}

\begin{algorithm}
\caption{Black-box Reduction from Agnostic Tomography to Robust Estimation}
\label{alg:reduction}
\KwIn{
$N=3N_{\rm mean}+N_{\rm density}$ copies of an $n$-qubit mixed state $\rho\in \C^{2^n\times 2^n}$. Oracle access to algorithms $\mathcal{A}_{\rm mean}$ and $\mathcal{A}_{\rm density}$ which solve Problems~\ref{prob:robust-mean} and \ref{prob:robust-density}. }
\KwOut{A product mixed state $\hat{\pi}$}

\ForEach{$P \in \{X, Y, Z\}$}{
    Define the POVM $\mathcal{M}_P = \bigotimes_{j=1}^{n}\{ (I_j+P_j)/2,\ (I_j-P_j)/2 \}$ where $I_j,P_j$ are Pauli operators on the $j$-th qubit. \\
    Take $N_{\rm mean}$ copies of $\rho$ and measure $\mathcal{M}_P$ on each copy to get outcomes $s_{i,P}\in \{0,1\}^{n}$ for $1\leq i\leq N_{\rm mean}$. \\
    $\Tilde{\mu}_P\leftarrow \mathcal{A}_{\rm mean}(\{s_{i,P}\}_{i=1}^{N_{\rm mean}})$ where $\hat{\mu}_P \in [0,1]^n$. 
}
\ForEach{$j\in [n]$}{
    Let $\Tilde c_{j, P}=1-2\Tilde{\mu}_{j,P}$ for $P\in \{X,Y,Z\}$ and $1\leq j\leq n$. Let $\Tilde{c}_{i}=(\Tilde c_{i,X}, \Tilde c_{i,Y},\Tilde c_{i,Z})$. \\
    Construct $\Tilde \pi_j=\frac{1}{2}(I+\Tilde{c}_{j}\cdot \vec{\sigma}_j)$ where $\sigma_j=(X_j,Y_j,Z_j)$ are the Pauli operators for the $j$-th qubit. \\
    Diagonalize $\Tilde \pi_j$ to get eigenvectors $\{\ket{\Tilde u_j},\ket{\Tilde v_j}\}$ ordered by decreasing eigenvalue magnitude. 
}
Define the POVM 
$\mathcal{M}_{\mathrm{learned}} = 
   \bigotimes_{j=1}^{n}\{ \ketbra{\tilde{u}_j}{\tilde{u}_j}, 
                           \ketbra{\tilde{v}_j}{\tilde{v}_j} \}$.

Take $N_{\rm density}$ copies of $\rho$ and measure $\mathcal{M}$ on each copy to get outcomes $s_i\in \{0,1\}^n$ for $1\leq i\leq N_{\rm density}$. \\ 
$\hat{\lambda}\leftarrow \mathcal{A}_{\rm density}(\{s_{i}\}_{i=1}^{N_{\rm density}})$ where $\hat{\lambda}\in [0,1]^n$. \\
\textbf{Return} $\hat{\pi}=\bigotimes_{j=1}^{n}\hat\pi_j$ where $\hat\pi_j=(1-\hat\lambda_j)\ketbra{\Tilde u_j}{\Tilde u_j}+\hat\lambda_j\ketbra{\Tilde v_j}{\Tilde v_j}$.
\end{algorithm}

Combining Theorem~\ref{thm:reduction} with Theorem~\ref{thm:robust-main-inf} (see Theorem~\ref{thm:robust-main} for a more detailed statement), we obtain the following corollary 
, which is a detailed statement of Theorem~\ref{thm:q-main-inf}.

\begin{corollary}
\label{cor:quantum-main}
Let $\eps, \delta > 0$.
    Let $\rho$ be an $n$-qubit density matrix.
    There is an algorithm, using only single-copy, single-qubit measurements, 
    which given $N \geq N_0$ copies of $\rho$, where 
    $N_0 = \widetilde{O} \left( \frac{n^2 \log(1/\delta)}{\eps^2} \right)$, 
    runs in $\poly(N)$ time, and outputs with probability at least $1-\delta$ 
    a description of a product state $\widehat{\pi}$ such 
    that $\trd(\rho, \hat{\pi}) \lesssim \opt \cdot \log 1 / \opt + \eps$, where $\opt = \inf_{\pi \in \cM_n} \trd (\pi, \rho)$.
\end{corollary}

\subsection{Setup} 
Before describing and analyzing our reduction, 
we first establish a formal connection 
between the corruption models in the quantum and classical settings. The following fact follows from the variational characterization of trace distance:
\begin{lemma}\label{lemma:tvd-leq-trd}
    Let $\rho,\sigma\in \C^{2^n\times 2^n}$ be two density matrices, and let $p$ and $q$ denote the corresponding distributions over measurement outcomes obtained by measuring $\rho$ and $\sigma$ with a POVM $\cM$. Then, 
    $\tvd(p,q)\leq \trd(\rho,\sigma) \;$.
\end{lemma}

Considering the trace distance guarantee in the setup of 
Problem~\ref{prob:agnostic-pmix}, if $\rho$ is $\epsilon$-close in 
trace distance to a product mixed state $\pi$, measuring $\rho$ 
with a POVM $\cM$ and seeing some outcome can be viewed as 
measuring $\pi$ with $\cM$ and seeing an $\epsilon$-contaminated 
outcome. Since applying any set of single-qubit measurements to 
$\rho$ gives us an $\epsilon$-contaminated draw from a product 
distribution, we are able to apply our robust learners from 
Problems~\ref{prob:robust-mean} and \ref{prob:robust-density} 
on the outcomes of such measurements. 

Naively, we begin by trying to directly learn the Pauli 
coefficients of each qubit by measuring in the $X^{\otimes n},Y^{\otimes n},$ and $Z^{\otimes n}$
Pauli bases and applying robust mean estimation for each of the 
three. Since Problem~\ref{prob:robust-mean} 
gives a guarantee in the $\ell_2$-norm, 
this allows us to construct a matrix of Pauli coefficients 
$\Tilde{c}\in \C^{n\times 3}$ such that 
$\|\Tilde{c}_{:,P}-c_{:,P}\|_2\lesssim f_{\rm mean}(\opt)$ 
for each coordinate $P\in \{X,Y,Z\}$, 
where $c_{:,P}=(c_{1,P},\dots,c_{n,P})$ 
are the respective Pauli coefficients of $\pi$. 
Then the corresponding product mixed state $\Tilde{\pi}=\bigotimes_{j=1}^{n}\Tilde{\pi}_j$ defined 
by these coefficients is close to the best product mixed 
state approximation in the sum of the Frobenius norm 
across each qubit, i.e., 
$\sum_{j=1}^{n}\|\Tilde\pi_j-\pi_j\|_F\lesssim f_{\rm mean}(\opt)$. 
If each qubit was sufficiently mixed, 
i.e., $\|\Tilde{c}_j\|_2\leq 1-\delta$ for all $1\leq j\leq n$ for some constant $\delta$, 
the naive approach would suffice to bound the trace distance. 
However, when certain qubits of the best product mixed state 
approximation are close to pure, that is $\|\Tilde{c}_j\|_2$ 
is close to $1$, the eigenvalues of $\pi_j$ are very unbalanced 
causing the 
bound to fail. 
This is analogous to the reason why $\ell_2$ control 
of the mean of a product distribution fails 
to bound the total variation distance 
when coordinates of the mean are very close to $0$ or $1$. 

Fortunately, we show in Section~\ref{sec:pf:thm:reduction} that our rough 
estimate in $\ell_2$-norm can be refined with another application of 
robust estimation. Specifically, we can salvage the estimate 
$\Tilde{\pi}$ by observing that $\pi$ is approximately diagonal 
in the eigenbasis of $\Tilde{\pi}$. 
Since $\sum_{j=1}^{n}\|\Tilde\pi_j-\pi_j\|_F\lesssim f_{\rm mean}(\opt)$ from the $\ell_2$ guarantee of robust mean estimation, 
the sum of the magnitudes of the off-diagonal terms of each $\pi_j$ 
in this eigenbasis is on the order of $f_{\rm mean}(\opt)$ as well. Because these off-diagonals are controlled, 
the second step of Algorithm~\ref{alg:reduction} learns 
the diagonal terms of $\pi$ using robust density estimation. 
Since Theorem~\ref{alg:main} holds for unbalanced product distributions, this handles the case when qubits of $\pi$ are close to pure. 

Of course, the analysis of Algorithm~\ref{alg:reduction} requires 
a more careful treatment of the off-diagonal terms 
to justify the fact that they can be ignored. 
In particular, we would like to control the off-diagonal terms 
of each qubit separately. While trace distance does not tensorize, 
we apply the Fuchs-van de Graaf inequality to upper bound 
by fidelity, which does tensorize. Then, it suffices to show that 
the off-diagonal contributions for each qubit only quadratically 
decrease the fidelity. We begin by doing so 
in Lemma~\ref{lemma:F-rho-rhodiag-taylor}, 
and then proceed to proving Theorem~\ref{thm:reduction} in Section~\ref{sec:pf:thm:reduction}.

\begin{lemma}\label{lemma:F-rho-rhodiag-taylor}
    Consider the following approximately-diagonal mixed state:
    $$\rho=\begin{bmatrix}
        \sigma_1 & a\\ \overline{a} & \sigma_2 
    \end{bmatrix}=\underbrace{\begin{bmatrix}
        \sigma_1 & 0\\ 0 & \sigma_2 
    \end{bmatrix}}_{\rho_{\rm diag}}+\begin{bmatrix}
        0 & a\\ \overline{a} & 0 
    \end{bmatrix} \;,$$
    where $t=|a|\ll 1$ and $\sigma_1>\sigma_2$. Then, $F(\rho,\rho_{\rm diag})\geq 1-2|a|^2$.
\end{lemma}
\begin{proof}
    Consider the definition of fidelity $F(\rho,\rho_{\rm diag})=(\tr\sqrt{\sqrt{\rho_{\rm diag}}\rho \sqrt{\rho_{\rm diag}}})^2$. Expanding, we can write: 
    $$\sqrt{\rho_{\rm diag}}\rho \sqrt{\rho_{\rm diag}}=\begin{bmatrix}
        \sigma_1^2 & a\sqrt{\sigma_1 \sigma_2} \\ \overline{a}\sqrt{\sigma_1 \sigma_2} & \sigma_2^2 \;.
    \end{bmatrix}$$
    Then, we have the characteristic equation
    \begin{align*}
    0&=(\sigma_1^2-\lambda)(\sigma_2^2-\lambda)-t^2\sigma_1\sigma_2=\lambda^2-(\sigma_1^2+\sigma_2^2)\lambda +\sigma_1^2\sigma_2^2-t^2\sigma_1\sigma_2
    \end{align*}
    which yields eigenvalues $$\lambda_{\pm}=\frac{\sigma_1^2+\sigma_2^2\pm \sqrt{(\sigma_1^2-\sigma_2^2)^2+4t^2\sigma_1\sigma_2}}{2} \;.$$ 
    Then, we obtain  
    \begin{align*}
    F(\rho,\rho_{\rm diag})&=\left(\tr\sqrt{\sqrt{\rho_{\rm diag}}\rho \sqrt{\rho_{\rm diag}}}\right )^2\\
    &=\left(\sqrt{\lambda_+(t)}+\sqrt{\lambda_-(t)}\right)^2\\
    &=\lambda_+(t)+\lambda_-(t)+2\sqrt{\lambda_+(t)\lambda_-(t)}\\
    &=\sigma_1^2+\sigma_2^2+\sqrt{(\sigma_1^2+\sigma_2^2)^2-((\sigma_1^2-\sigma_2^2)^2+4t^2\sigma_1\sigma_2)}\\
    &=\sigma_1^2+\sigma_2^2+2\sqrt{\sigma_1^2\sigma_2^2-t^2\sigma_1\sigma_2} \;.
    \end{align*}
    Let $\gamma=\sigma_1\sigma_2$ such that $\sigma_1^2+\sigma_2^2=(\sigma_1+\sigma_2)^2-2\sigma_1\sigma_2=1-2\gamma$. It then follows that 
    $$F(\rho,\rho_{\rm diag})=1-2\gamma+2\sqrt{\gamma^2-\gamma t^2}=1-2t^2\cdot \frac{\sqrt\gamma}{\sqrt\gamma+\sqrt{\gamma-t^2}}\geq 1-2t^2 \;,$$
    which completes the proof of Lemma~\ref{lemma:F-rho-rhodiag-taylor}.
\end{proof}
\subsection{Proof of Theorem~\ref{thm:reduction}}\label{sec:pf:thm:reduction}
We are now ready to prove Theorem~\ref{thm:reduction}.
\begin{proof}
    For product mixed state $\pi=\bigotimes_{j=1}^{n}\pi_j\in \cM_n$, we can decompose:
    $$\pi_j\equiv \frac{1}{2}(I+c_j\cdot \sigma_j),$$
    where $c_j\in \R^3$ with $\|c_j\|_2\leq 1$ and $\sigma_j=(X_j,Y_j,Z_j)$ are the Pauli operators on the $j$-th qubit. Now, consider the POVM $\left\{\frac{I+X}{2},\frac{I-X}{2}\right\}^{\otimes n}$ on $\pi$. The probability of each measurement outcome for the $j$-th qubit is: 
    \begin{align*}
p_{j,0}&\equiv \tr \pi_j  \frac{I+ X_j}{2}=\frac{1}{4}\tr ((I+c_{j,X} X+c_{j,Y} Y c_{j,Z} Z)(I+ X))=\frac{1+ c_{j,X}}{2} \;.
    \end{align*}
Then, the distribution of outcomes for the whole POVM is simply a binary 
product distribution, $q$, over $\{0,1\}^n$, where the mean of each coordinate is $\mu_j=\frac{1-c_{j,X}}{2}$. Measuring $\rho$ with this POVM gives a distribution, $p$, over the same hypercube. By Lemma~\ref{lemma:tvd-leq-trd}, we then know that $\tvd(p,p)\leq \trd(\rho,\pi)\leq \opt$. Applying our algorithm for Problem~\ref{prob:robust-mean}, 
we can recover a $\Tilde \mu$ such that:
$$\|\Tilde \mu-\mu\|_2\leq f_{\rm mean}(\opt) + \epsilon \; .$$
Doing this for the two other POVMs generated by replacing $X$ with $Y$ 
and then $Z$, for any $P\in \{X,Y,Z\}$, we recover 
an estimate $\Tilde c_{:,P}\in \R^n$ such that 
$\|\Tilde {c}_{:,P}-{c}_{:,P}\|_2\leq 2f_{\rm mean}(\opt)$ where $c_{:,P}=(c_{1,P},\dots,c_{n,P})$. This gives a matrix of coefficients $\Tilde c\in \C^{n\times 3}$. 
From these, we can construct $\Tilde{\pi}=\bigotimes_j \Tilde{\pi}_j\in \cM_n$, 
where 
$$\Tilde{\pi}_j\equiv \frac{1}{2}(I+\Tilde c_j\cdot \sigma_j)$$
with $\Tilde c_j=(c_{j,X},c_{j,Y},c_{j,Z})$ and $\sum_j \|\Tilde{\pi}_j-\pi_j\|_F^2\leq 6 (f_{\rm mean}(\opt) + \epsilon)^2$. 

If $\|\Tilde c_j\|_2\leq 1-\delta$ for all $j\in [n]$, one can show that 
this error already suffices to achieve small trace distance. 
Since this is not necessarily the case, we must correct 
our estimator $\pi_j$. Our motivation is to use $\Tilde c_j$ to construct 
a new basis in which $\pi_j$ is approximately diagonal. 
We can then learn the diagonal entries by measuring in this basis 
and applying robust density estimation for arbitrary binary product
distributions to learn the diagonal. 

Specifically, we decompose:
$$\Tilde{\pi}_j=(1-\Tilde \lambda_j) \ketbra{\Tilde u_j}{\Tilde u_j}+\Tilde \lambda_j\ketbra{\Tilde v_j}{\Tilde v_j} \;,$$
where $\ket{u_j},\ket{v_j}\in \C^2$ are the eigenvectors ordered by eigenvalue magnitude. In the $\{\ket{\Tilde u_j},\ket{\Tilde v_j}\}$ basis, we can write 
$$\Tilde{\pi}_j=\begin{bmatrix}
    1-\Tilde\lambda_j & 0 \\ 0 & \Tilde\lambda_j
\end{bmatrix}\quad \textrm{and} \quad \pi_j=\begin{bmatrix}
    1-\lambda_j & a_j \\ \overline{a}_j & \lambda_j
\end{bmatrix} \;,$$
where we have control of the off-diagonal via 
$\sum_j |a_j|^2\leq \frac 12 \sum_j \|\Tilde \pi_j-\pi_j\|_F^2\leq 3 (f_{\rm mean}(\opt) + \eps)^2$. 
Then, we can measure the POVM:
$$\bigotimes_{i=1}^{n} \{\ketbra{\Tilde u_j}{\Tilde u_j},\ketbra{\Tilde v_j}{\Tilde v_j}\}$$
such that the distribution of outcomes when applying the 
POVM to $\pi$ is a binary product distribution $q_{uv}$ 
over $\{0,1\}^n$ with mean vector $\lambda\in \R^n$. 
The effect of this POVM when actually applied to $\rho$ 
gives an arbitrary product distribution $p_{uv}$, 
which by Lemma~\ref{lemma:tvd-leq-trd} satisfies $\tvd(p_{uv},q_{uv})\leq \trd(\rho,\pi)\leq \opt$. 
Applying our algorithm for Problem~\ref{prob:robust-density}, 
we can then recover $\hat{\lambda}\in \R^{n}$ such that:
$$\tvd(\Bern(\hat{\lambda}),\Bern(\lambda))\leq f_{\rm density}(\opt) + \eps \;,$$
where $\Bern(\lambda)$ denotes the binary product 
distribution where the $j$-th marginal is $\Bern(\lambda_j)$. 
Thus, we construct:
$$\hat{\pi}=\bigotimes_{j=1}^{n}\hat{\pi}_j,\quad \hat{\pi}_j=\begin{bmatrix}
    1-\hat{\lambda}_j & 0 \\ 0 & \hat{\lambda}_j
\end{bmatrix}$$
written in the $\{\ket{u_j},\ket{v_j}\}$ basis, which we claim 
achieves the desired trace distance bound. 
To show this, let $\pi'$ be the diagonal portion of $\pi$ 
in the learned basis:
$$\pi'=\bigotimes_j \pi'_j,\quad \pi'_j=\begin{bmatrix}
    1-\lambda_j & 0\\ 0& \lambda_j
\end{bmatrix} \;.$$
By Fuchs-van de Graaf and Lemma~\ref{lemma:F-rho-rhodiag-taylor}, we have that 
\begin{align*}
    \trd(\pi,\pi')\leq \sqrt{1-\prod_{j=1}^{n}F(\pi_j,\pi'_j)}\leq  \sqrt{1-\prod_{j=1}^{n}(1-2|a_j|)^2}\leq \sqrt{2\sum_j |a_j|^2}\leq \sqrt 6\, f_{\rm mean}(\opt) + \eps \;,
\end{align*}
This shows that $\pi$ is sufficiently diagonal in the basis learned by the first round of measurement. Then, we have that 
\begin{align*}
    \trd(\pi',\hat\pi)=\tvd(\Bern(\lambda),\Bern(\hat\lambda))\leq f_{\rm density}(\opt) + \epsilon
\end{align*}
meaning that our diagonal estimate is a good estimate in trace distance. Thus, we conclude that 
\begin{align*}
    \trd(\pi,\hat{\pi})\leq \trd(\pi,\pi')+\trd(\pi',\hat{\pi})\leq \sqrt 6 \; f_{\rm mean}(\opt)+f_{\rm density}(\opt) + 2 \eps \; .
\end{align*}
which, up to adjusting constants in the $\eps$, completes the proof of Theorem~\ref{thm:reduction}. 
\end{proof}

\subsection{Semi-Agnostic Pure Product State Tomography}
\label{sec:pure}
Recalling the motivating related work of learning the closest product state in fidelity \citep{bakshi2025learning}, we ask if our reduction in Theorem~\ref{thm:reduction} can be modified for the original pure state setting. Fortunately, the answer is yes. Since pure product states are a subclass of the product mixed states which we have considered thus far, we show that the estimated product mixed state from Theorem~\ref{thm:reduction} can be \textit{rounded} to a pure product state while preserving the semi-agnostic error guarantee. We state this formally in the following theorem. 
Recall we let $\Pi_n$ denote the set of pure product states.

\begin{theorem}[Agnostic Pure Product State Tomography]\label{thm:agnostic-pure}
Let $\epsilon>0$, and let $\delta>0$. Let $\rho$ be an $n$-qubit density matrix,

 There is an algorithm, using only non-adaptively chosen, single-copy, single-qubit measurements, which given $N\geq N_0$ copies of $\rho$, where $N_0= O \left(\frac{n\log(1/\delta)}{\epsilon^2}\right)$, runs in $\poly(N)$ time, and outputs with probability at least $1-\delta$ a description of a pure product state $\ket{\hat \pi}$ such that $\trd(\rho,\ketbra{\hat \pi}{\hat \pi}) = O(\opt \sqrt{\log(1/\opt)} )+\eps$,  
where $\opt = \opt(\Pi_n) \eqdef \inf_{\ket{\pi} \in \Pi_n} \trd(\ketbra{\pi}{\pi}x,\rho)$.
\end{theorem}
\begin{proof}
    Let $\ketbra{\pi_j}=\frac 1 2 (I_j+c_j\cdot \sigma_j)$ be the Pauli decomposition of each qubit as before. Since $\ket{\pi_j}$ is a pure state, $\|c_j\|=1$. By Theorem~\ref{thm:reduction}, the first round of Algorithm~\ref{alg:reduction} outputs a matrix of coefficient estimates $\Tilde c\in \C^{n\times 3}$ such that 
    $\|\Tilde c -c\|_F^2\leq 12 \left( f_{\rm mean}^2(\opt) + \epsilon \right) $
    where $\Tilde \pi_j\equiv \frac 1 2 (I_j+\Tilde c_j\cdot \sigma_j)$. Now, let 
    $$\ketbra{\hat \pi_j}{\hat \pi_j}\equiv \frac 1 2 (I_j+\hat c_j \cdot \sigma_j)$$
    where $\hat c_j=\Tilde c_j /\|\Tilde c_j\|_2$ is the rounded Bloch vector. We claim that this rounded pure product state is close to $\pi$ in trace distance. Bounding the coefficient change due to rounding,
   \begin{align*}
   \|\hat c_j- c_j\|_2 &\leq \| \hat c_j - \Tilde c_j \|_2 + \| \Tilde c_j - c_j \|_2 \\
   &= 1 - \norm{\Tilde c_j}_2 + \norm{\Tilde c_j - c_j}_2 \\
   &\leq 2 \|\Tilde c_j-c_j\|_2 \; .
   \end{align*} 
    Considering the fidelity between the estimated and true pure state, 
    \begin{align*}
        F(\ket{\hat \pi_j}, \ket{\pi_j})&=|\braket{\hat \pi_j}{\pi_j}|^2 =\tr (\hat \pi_j \pi_j)=\frac{1+\hat c_j\cdot c_j}{2}=1-\frac{\|\hat c_j-c_j\|_2^2}{4} \; ,
    \end{align*}
    and thus
    \begin{align*}
        F(\ket{\hat \pi},\ket{\pi})&=\prod_{j=1}^{n}F(\ket{\hat \pi_j}, \ket{\pi_j})\geq 1-\frac 1 4 \sum_{j=1}^{n}\|\hat c_j-c_j\|_2^2 \geq 1- 12 (f_{\rm mean}(\opt) + \eps)^2 \; .
    \end{align*}
    Then, by Fuchs-van de Graaf,
    \begin{align*}
        \trd(\ketbra{\hat \pi}{\hat \pi}, \ketbra{\pi}{\pi})\leq \sqrt{1-F(\ket{\hat \pi}, \ket{\pi})}\leq 2\sqrt 3 \left( \;f_{\rm mean}(\opt) + \epsilon \right) \;.
    \end{align*}
    Using our robust estimation primitive for binary products from Theorem~\ref{thm:main-l2}, and by slightly adjusting the choice of $\epsilon$ in the proof, we obtain the desired semi-agnostic guarantee. 
\end{proof}
\noindent
Notably, Theorem~\ref{thm:agnostic-pure} only relies on the non-adaptive portion of Algorithm~\ref{alg:reduction} and performs no additional measurements.
In particular, the measurements required are non-adaptively chosen, single-qubit measurements.

In the special case where $\rho$ itself is also pure, note that our algorithm also implies an algorithm for semi-agnostic learning in fidelity: this is because for any two pure states $\ket{\psi}, \ket{\phi}$, we have that
\[
F(\ket{\psi}, \ket{\phi}) = 1 - \Theta (\trd (\ketbra{\psi}{\psi}, \ketbra{\phi}{\phi}))^2 \; .
\]
Combining this result with Theorem~\ref{thm:agnostic-pure}, we have:
\begin{corollary}
Let $\epsilon_0>0$ be a sufficiently small universal constant, and let $\delta>0$. Let $\ket{\psi}$ be an $n$-qubit pure state so that there is a pure product state $\ket{\pi}=\ket{\pi_1}\dots\ket{\pi_n}$ such that $F(\ket{\psi}, \ket{\pi}) \eqdef \opt_F (\ket{\psi}) = 1 - \eps^2$ for some $\epsilon\leq \epsilon_0$. 
There is a non-adaptive algorithm, using only single-copy, single-qubit measurements, which given $N\geq N_0$ copies of $\ket{\psi}$, where $N_0 = O \left(\frac{n\log(1/\delta)}{\epsilon^2}\right)$, runs in $\poly(N)$ time, and outputs with probability at least $1-\delta$ a description of a pure product state $\ket{\hat \pi}$ such that $F(\ket{\psi},\ket{\hat \pi}) \geq \opt_F (\ket{\psi}) - O( \epsilon^2 \sqrt{\log 1/\epsilon})$.
\end{corollary}
\noindent
The closest comparison between this result is the algorithm by \citep{bakshi2025learning} which gives an polynomial time algorithm under the same guarantee which achieves fidelity $1 - (1 - o(1)) \eps$.
In contrast, we provide a weaker semi-agnostic guarantee.
However, their algorithm requires somewhat more sophisticated measurements across each copy of the state $\rho$.
In particular, their algorithm requires first performing a qubit-wise unitary to every qubit, then learning the value of $\bra{e_i} \rho \ket{0^n}$, where $e_i$ is the string which is $1$ in position $i$ and $0$ otherwise.
This can be done with qubit-wise measurements, but would (at least na\"{i}vely) require an intermediate measurement of all but the $i$-th qubit, which may be challenging to implement on realistic architectures.
In contrast, our algorithm only requires single-qubit, non-adaptively chosen measurements, and thus may be substantially more practical to implement.

\section{Robustly Learning Binary Products Near-Optimally } \label{sec:robust-product}

In this section, we give our novel algorithm for robustly learning product distributions.

\subsection{Additional Technical Background} 
In this section, we will need several well-known facts from probability theory.
\begin{definition}[Hellinger distance]
    For two distributions $p, q$, the \emph{Hellinger distance} between $p$ and $q$ is defined to be 
    \[
    \hd (p, q) = \int \left( \sqrt{dp} - \sqrt{dq} \right)^2 \; .
    \]
\end{definition}
We will need the following basic facts about Hellinger distance:
\begin{fact}
\label{fact:hellinger}
    Let $p, q$ be two distributions. Then:
    \begin{itemize}
        \item \textbf{Hellinger upper bounds TV:} $\tvd(p, q) \leq \sqrt{2} \hd (p, q)$.
        \item \textbf{Subadditivity:} If $p = (p_n, \ldots, p_n)$ and $q = (q_n, \ldots, q_n)$ are both product distributions across the coordinates, then
        \[
        \hd(p, q)^2 \leq \sum_{i = 1}^n \hd (p_i, q_i)^2 \; . 
        \]
    \end{itemize}
\end{fact}
\noindent
By direct calculation, one can show that $p$ and $q$ are Bernoulli with means $\mu, \nu$ respectively, then
\begin{equation}
\label{eq:bernoulli-hellinger}
\hd(p, q)^2 = O \left( \frac{(\mu -\nu)^2}{\min (\mu, 1 - \mu)} \right) \; .
\end{equation}
\Cref{eq:bernoulli-hellinger} and Fact~\ref{fact:hellinger} together immediately imply:
\begin{corollary}
\label{cor:hellinger-bound}
Let $p, q$ be two binary product distributions with mean vectors $\mu, \nu \in \R^n$, and suppose $\mu_i \leq 2/3$ for all $i = 1, \ldots, n$.
Then,
\[
\tvd (p, q)^2 \leq \sum_{i = 1}^n \frac{(\mu_i - \nu_i)^2}{\mu_i} \; .
\]
\end{corollary}

\paragraph{The set $\sW_{n, \eps}$}
We will heavily leverage the standard \emph{filtering} framework for our upper bounds, and in particular, the weighted filter~\cite{dong2019quantum,diakonikolas2019recent,diakonikolas2023algorithmic}.
We will chiefly follow the presentation in~\cite{dong2019quantum}.
For simplicity of notation, we let $S = \{X_1, \ldots, X_N\}$, and we will associate indices with their associated points as necessary, i.e., we will say $i \in S_g$ if $X_i \in S_g$, etc.
We will assign to each point a nonnegative weight $w_i$, that we will evolve over the course of the algorithm.
Formally, we denote the set of allowable weights by $\Gamma_n:$
\begin{equation}
    \label{eq:gamma}
    \Gamma_{N} = \left\{ w \in \R^N: \sum_{i = 1}^N w_i \leq 1 \; \mbox{and} \; w_i \geq 0 \; \mbox{for all $i = 1, \ldots, N$} \right\} \; .
\end{equation}
For any set $T \subseteq [N]$, let $w(T) \in \Gamma_N$ be defined by $w(T)_i = \frac{1}{N} \cdot \mathbf{1}_{i \in T}$ for all $i = 1, \ldots, N$.
For two sets of weights $w, w' \in \Gamma_N$, we say $w \leq w'$ if $w_i \leq w'_i$ for all $i = 1, \ldots, N$.
We also define weighted notions of the mean and covariance: for any $w \in \Gamma_N \setminus \{0\}$, we let
\begin{equation}
    \mu(w) = \sum_{i = 1}^N \frac{w_i}{\|w\|_1} X_i \;, \qquad \mbox{and} \qquad \Sigma (w) = \sum_{i = 1}^N w_i (X_i - \mu(w)) (X_i - \mu(w))^\top \; .
\end{equation}
More generally, for any function $f$, we let $\E_{X \sim w} [f(X)] = \frac{1}{\norm{w}_1} \sum_{i \in S} w_i f(X_i)$.

Our algorithm will primarily work with the following set of weights:
\begin{equation}
    \label{wne}
    \sW_{N, \eps} = \left\{ w \in \Gamma_N: w \leq w(S) \;, \mbox{and} \; \norm{w - w(S)}_1 \leq \eps \right\} \; .
\end{equation}
The key invariant that we will need about these weights is the following.
For any vector $w$, let $\nnz(w)$ denote the number of nonzero entries of $w$.

\begin{lemma}[see, e.g.,~\cite{DKKLMS16,dong2019quantum}]
\label{lem:filtering}
    Let $\tau \in \R^N \setminus \{0\}$ be a entrywise non-negative, and let $w \in \Gamma_N$.
    Let $S = A \cup B$ for disjoint $A, B$ and assume that
    \[
    \sum_{i \in A} w_i \tau_i \leq \sum_{i \in B} w_i \tau_i \; .
    \]
    Consider the updated set of weights $w' \leq w$ given by
    \[
    w'_i = \left( 1 - \frac{\tau_i}{\tau_{\max}}\right) w_i \; ,
    \]
    where $\tau_{\max} = \max_{i \in [n]} \tau_i$.
    Then $w'$ satisfies $\nnz (w') < \nnz (w)$, and
        \[
        \sum_{i \in A} (w_i - w_i') < \sum_{i \in B} (w_i - w_i') \; .
        \]
\end{lemma}
Intuitively speaking, this lemma states that if there is a way to assign scores (the $\tau_i$) to the data points, in a way so that the weighted sum of the scores on $B$ exceeds that on $A$, then there is a way to update the weights in a way which decreases more mass on $B$ than on $A$.
This is the key point of the filtering procedure: roughly, larger scores will correspond to points which seem to be more suspicious.
If we can guarantee that the scores will satisfy this ``larger-on-average`` property on the bad points, then the lemma states that we are guaranteed to decrease more mass on the bad points then the good points.

\subsection{Simple Preprocessing Reductions} \label{ssec:robust-prep}

The following reductions from~\cite{DKKLMS16} will be useful.
First, as observed in Section 7.2.2 of~\cite{DKKLMS16}, if there is any coordinate $i$ so that $\mu (S)_i \leq \eps / n$ or $\mu(S)_i \geq 1 - \eps / n$, then there is a simple polynomial-time algorithm which can identify such coordinates, and which estimates the mean of these coordinate to be $0$ or $1$ respectively, and this will induce an TV error by at most $O(\eps / n)$.
Thus, by a triangle inequality, removing all such coordinates will affect the overall TV error by at most $O(\eps)$, so without loss of generality, we can assume that we have removed all such coordinates, and so we may assume that
\begin{equation}
\label{eq:regularity1}
    \frac{\eps}{n} \leq \mu(S)_i \leq 1 - \frac{\eps}{n} \; ,
\end{equation}
for all $i = 1, \ldots, n$. 
Next, we will use the following:
\begin{lemma}[Lemma 7.26 in~\cite{DKKLMS16}]
    Let $\pi \in \cP_n$ with mean vector $\mu$, and let $S$ be an $\eps$-corrupted set of samples from $\pi$ of size at least $\Omega (n)$.
    Then, with probability $1 - \exp (- \Omega (\eps n))$, there exists a product distribution $\pi'$ with mean vector $\mu'$ so that $S$ is an $1.2 \eps$-corrupted set of samples from $\pi'$, and moreover $\mu'$ satisfies $\mu(S)_i \geq \mu'_i / 3$ and $1 - \mu(S)_i \leq 1 - \mu'_i / 3$.
\end{lemma}
In other words by replacing $\pi$ with $\pi'$, this allows us to assume without loss of generality (by incurring a small constant blow-up in $\eps$) that
\begin{equation}
\label{eq:regularity2}
    \mu(S)_i \geq \frac{\mu_i}{3} \;, \qquad \mbox{and} \qquad 1 - \mu(S)_i \leq 1 - \frac{\mu_i}{3} \; ,
\end{equation}
for all $i = 1, \ldots, n$.
In light of these results, for the rest of the section, we will assume~\Cref{eq:regularity1} and~\Cref{eq:regularity2} hold deterministically.

Next, note that we can assume that $\mu_i \leq 2/3$ for all $i = 1, \ldots, n$.
This is because if $\mu_i \geq 2/3$, then $\mu(S)_i \geq 3/5 - \eps > 1/2$ except with exponentially small probability, and so if there is any coordinate $i$ so that $\mu(S)_i \geq 1/2$, we can simply flip the role of $0$ and $1$ in this coordinate, and this will guarantee that, except with vanishing probability, $\mu_i < 2/3$.

\paragraph{Reducing to known $\opt$ and constant $\delta$}
We next describe how to reduce to the case of known $\opt$.
The remainder of the section will be dedicated to the proof of the following theorem:
\begin{theorem}
\label{thm:robust-main}
    Let $\eps_0>0$ be some universal constant.
    There is an algorithm (Algorithm~\ref{alg:main}), which given an $\eps$-corrupted set of samples from an unknown product distribution $p \in \cP_n$, for $\eps \leq \eps_0$, of size $N \geq N_0$, where  $N_0 = \widetilde{O} \left( \frac{n^2 \log 1 / \delta}{\eps^2} \right)$, outputs with probability $0.99$ the mean vector for a product distribution $\hat{p}$ satisfying $\tvd(p, \hat{p}) \lesssim \eps \log 1 / \eps$.
    Moreover, the algorithm runs in time $\poly(N) = \poly (n, 1/\eps)$.
\end{theorem}
Before we show this, we first show how it is sufficient to prove Theorem~\ref{thm:robust-main-inf}. 

\begin{proof}[Proof of Theorem~\ref{thm:robust-main-inf} given Theorem~\ref{thm:robust-main}]
This is essentially the doubling argument described in~\cite{DKKLMS16}.
This argument is more or less standard in the literature (see e.g. Remark 2.24 in~\cite{DKKLMS16}), and so we will be somewhat terse here.
Assume there is an algorithm which, given the dataset $S$, $\eps >0$, and knowledge of $\opt$, outputs with probability $1 - \delta'$ an estimate $\pi$ so that $\tvd (p, \pi) \leq f(\opt) + \eps$.

Now, simply run this algorithm with internal value of $\eps$ set to $\eps, 2 \eps, 4 \eps, \ldots, 1/2$.
Each one generates a candidate solution $\pi_i$, for $i = 1, \ldots,L$, for $L = O(\log 1 / \eps)$.
Let $\mu_1, \ldots, \mu_L$ denote the associated mean vectors.
Round each coordinate of each mean vector to the closest integer multiple of $\eps^2 / n$.
By doing so, we ensure that each output of the algorithm is deterministically in a family of hypotheses $\mathcal{M}$ of size $(n / \eps^2)^n$, and clearly each such hypothesis can be sampled from in linear time.
By Equation~\ref{eq:regularity1} and Corollary~\ref{cor:hellinger-bound}, this changes the TV distance of each $\pi_i$ by at most $\eps$.
We can now apply Lemma~\ref{lem:robust-hypothesis-selection} on these hypotheses.
Let $i$ be the smallest index so that $2^i \eps \geq 1.01 \cdot \opt$.
By Fact~\ref{fact:contamination-to-corruption}, we know that except with $\exp (-c \eps n) \ll \delta$ probability, $S$ is an $2^i \eps$-corrupted set of points from some product distribution $\pi^*$ satisfying $\tvd (\pi^*, p) \leq 2^i \eps$, and hence $\tvd (\pi^*, p_i) \leq O(\opt)$, with probability $\delta'$.
Hence, by Lemma~\ref{lem:robust-hypothesis-selection} and a union bound, with probability $O(\delta' \cdot \log (1 / \eps))$, we can output a $\pi_j$ so that $\tvd (\pi^*, \pi_j) \leq O(\opt)$, whence $\tvd (p, \pi_j) \leq O(\opt)$.
The additional sample overhead is $O(\tfrac{\log |\mathcal{M}| + \log \log (1 / \eps) \log (1 / \delta)}{\eps^2}) = \widetilde{O} \left( \tfrac{n}{\eps^2} \right)$, so this does not add any additional overhead to our overall sample complexity.
Similarly, the runtime of the overall algorithm will still be polynomial.
\end{proof}

\noindent
Therefore, for the rest of this section, we will assume that the algorithm knows $\opt$.
In a slight abuse of notation, we will let $\opt = \eps$.

Finally, we note that by standard robust boosting techniques (see e.g. Lemma 2.23 in~\cite{DKKLMS16}), it suffices to show this for $\delta$ constant, from which we immediately obtain the overall bound.
Thus, for the rest of the section, we will show Theorem~\ref{thm:robust-main} for $\delta$ being a small constant.

Additionally, for the rest of the section, we let $S$ be our $\eps$-corrupted set of samples of size $n$ from $p \in \cP_n$ with mean $\mu$.

\subsection{Characterization of TV Distance between Product Distributions}
Previous work of~\cite{DKKLMS16} obtained suboptimal results for robust learning of binary product distributions, in large part because they did not have a tight characterization of the TV distance.

The first contribution here is to demonstrate such a tight characterization.
The key idea will be to use the following distance:
\begin{definition}\label{def:norm}
    For any vector $\mu \in \R^n$ with $0 \leq \mu_i \leq 2/3$ for all $i = 1, \ldots, n$, let 
    \begin{equation}
        \cT_\mu = \left\{ y \in \R^n: \norm{y}_\infty \leq 1 \; , \mbox{ and } \; \sum_{i = 1}^n \mu_i y_i^2 \leq 1 \right\} \; .
    \end{equation}
    We also denote the dual norm with respect to this set by $\norm{x}_{\mu} = \sup_{y \in \cT_\mu} \iprod{y, x}$.
\end{definition}
Intuitively, this set captures an ``intermediate'' set of test vectors, namely, test vectors which are both bounded in $\ell_\infty$, as well as which are bounded in some relative $\ell_2$ sense, relative to $\mu$.
The idea is that the former set of test vectors form the natural set of dual vectors to the $\ell_1$ norm, and the latter set of test vectors forms the set of dual vectors to some notion of $\chi^2$-divergence.
The idea is that in some coordinates, namely the unbalanced ones, the ``optimal'' witness to the statistical farness of two product distributions should use the $\ell_\infty$ bound, and in the others, the bound one can obtain from the $\chi^2$-divergence ought to be tight.
We can formalize this below:
\begin{theorem}
\label{thm:product-tv-bound}
    Let $\pi, \sigma$ be two Boolean product distributions with mean vectors $\mu, \nu $, and suppose that $0 \leq \mu_i \leq 2/3$ for all $i = 1, \ldots, n$.
    Then
    \begin{equation}
        \tvd (\pi, \sigma) \leq O \left( \min \left( 1, \norm{\mu - \nu}_{\mu}\right) \right) \; .
    \end{equation}
\end{theorem}
We note that one can in fact show that this bound is tight up to constant factors (in fact, the proof below also shows this), although we will not directly need this.
\begin{proof}[Proof of~\Cref{thm:product-tv-bound}]
Let $\delta_i = \mu_i - \nu_i$, and let $a_i = |\delta_i| / \mu_i$.
Sort the coordinates in decreasing order of $a_i$, so that without loss of generality, we assume that $a_1 \geq a_2 \geq \ldots \geq a_n$.

Let $k$ be the largest integer so that $\sum_{i \leq k} \mu_i \leq 1$.
Note that $\sum_{i \leq k} \mu_i \geq 1/3$ since each $\mu_i$ is at most $2/3$.
Let $\pi_{\leq k}, \sigma_{\leq k}$ denote the restriction of $\pi$ and $\sigma$ to these coordinates, and let $\pi_{> k}, \sigma_{> k}$ denote the restriction of $\pi$ and $\sigma$ to the remaining coordinates.
By sub-additivity of total variation distance for product distributions, we have that
\[
\tvd (\pi, \sigma) \leq \tvd (\pi_{\leq k}, \sigma_{\leq k}) + \tvd (\pi_{> k}, \sigma_{> k}) = O \left( \max \left(\tvd (\pi_{\leq k}, \sigma_{\leq k}),  \tvd (\pi_{> k}, \sigma_{> k}) \right) \right) \; .
\]
Hence, by a further application of sub-additivity and by Corollary~\ref{cor:hellinger-bound}, we have that
\begin{equation}
\label{eq:max-term}
    \tvd (\pi, \sigma) \leq O \left( \max \left\{ \sum_{i = 1}^k |\delta_i|, \left(\sum_{i \geq k} \frac{\delta_i^2}{\mu_i} \right)^{1/2} \right \} \right) \; .
\end{equation}
We now split into two cases, depending on which term on the RHS dominates.
First, suppose that 
\begin{equation}
 \left( \sum_{i \geq k} \frac{\delta_i^2}{\mu_i} \right)^{1/2}  \leq \sum_{i = 1}^k |\delta_i| \; .
\end{equation}
Then, by the definition of $k$, if we let $y_i = \mathrm{sign} (\mu_i - \nu_i)$ for $i \leq k$ and $y_i = 0$ otherwise, we have that $y \in \cT_\mu$, and so $\tvd(\pi, \sigma) \leq \sup_{y \in \mathcal T_\mu} \iprod{y, \mu - \nu}$, and so the theorem is true in this case.

Otherwise, suppose that
\begin{equation}
\label{eq:chi2-case}
 A = \left( \sum_{i \geq k} \frac{\delta_i^2}{\mu_i} \right)^{1/2}  \geq \sum_{i = 1}^k |\delta_i| \; .
\end{equation}
Note that
\begin{equation}
\label{eq:delta-sum-bound}
\sum_{i = 1}^k |\delta_i| = \sum_{i = 1}^k \mu_i a_i \geq a_{k + 1} \sum_{i = 1}^k \mu_i \geq \frac{1}{3} a_{k + 1} \; .
\end{equation}

In this case, let $c > 0$ be a sufficiently small universal constant, and define $y_i = \tfrac{\mu_i - \nu_i}{3 A \cdot \mu_i}$ for $i > k$, and $y_i = 0$ otherwise.
Observe that, by~\Cref{eq:delta-sum-bound}, we have that
\begin{align*}
    |y_i| \leq  \frac{|\delta_i|}{3 \mu_i \sum_{j = 1}^k |\delta_j|} \leq \frac{a_i}{a_{k + 1}} \leq 1 \; .
\end{align*}
We also have that
\begin{align*}
    \sum_{i = 1}^n \mu_i y_i^2 =  \frac{1}{9A^2} \sum_{i > k} \frac{(\mu_i - \nu_i)^2}{\mu_i} \leq 1 \; ,
\end{align*}
and so these together imply that $y \in \cT_\mu$.
Since 
\begin{align*}
    \iprod{y, \delta} = \frac{1}{3A} \sum_{i > k} \frac{\delta_i^2}{\mu_i} = \frac{A}{3} \; ,
\end{align*}
this implies that in this case, we have $\tvd (\pi, \sigma) \leq O (\sup_{y \in \cT_\mu} \iprod{y, \mu - \nu})$ as well.
This completes the proof.
\end{proof}

\paragraph{A convex relaxation}
We briefly recall the spectral filter for learning the mean of a balanced product distribution from ~\cite{DKKLMS16}. 
In that paper, the key point was that one could upper bound the deviation of the empirical mean by spectral properties of the empirical covariance with the diagonal zeroed out.
By running the filter to successively downweight points that are causing the empirical covariance to have large spectral norm, we can ensure that the resulting set of weighted points has bounded covariance, and moreover, must still have the vast majority of its weight on the good points.
Note that this step corresponds to filtering based the variance of linear test functions $x \mapsto \iprod{x, y}$, where $y$ is a unit vector.

However, to obtain total variation bounds, we should not consider tests based on unit vectors $y$, but rather tests based on vectors $y \in \mathcal T_\mu$, since such vectors witness the difference in TV distance directly.
However, finding a $y$ that maximizes the expectation of this test function over the dataset is computationally nontrivial. Instead, we will want to consider a convex relaxation of this set of test functions.
Formally, let $\mathbb{S}^n_{\geq 0}$ denote the set of symmetric $n \times n$ real-valued positive semi-definite matrices, and define the set
\begin{equation}
\label{eq:st-mu}
\sT_\mu = \left\{ M \in \mathbb{S}^n_{\geq 0}: |M_{ij}| \leq 1 \mbox{ for all $i, j$}, \sum_{i = 1}^n M_{ii} \mu_i \leq 1, \sum_{i,j} M_{ij}^2 \mu_i \mu_j \leq 1, \sum_i \mu_i \cdot \sup_{j} M_{ij}^2  \leq 1 \right\} \; .
\end{equation}
One can easily verify that for all $y \in \cT_\mu$, we have that $y y^\top \in \sT_\mu$. 
Intuitively, the idea is that since the set of $y \in \cT_\mu$ captures $y$ which are simultaneously dual to $\ell_1$ and to the $\chi^2$-divergence, to obtain a good relaxation of this set, we need to enforce all combinations of $\ell_1$ and $\chi^2$-divergences across all rows and columns.

Moreover, because all the constraints are either linear or sums of squares of nonnegative polynomials, this is a convex set.
Moreover, while \eqref{eq:st-mu} encodes exponentially many constraints, one can build a polynomial-time separation oracle for it, and thus by the classic theory of convex programming~\cite{grotschel2012geometric}, one can optimize over this set in polynomial time.

Similarly to before, we can also define the natural dual norm with respect to $\sT_\mu$.
Namely, for any matrix $A$, we let 
\begin{equation}
\norm{A}_\mu = \sup_{M \in \sT_\mu} \left| \iprod{A, M} \right| \; .
\end{equation}
Since $\left|\iprod{A,M}\right|$ can be written as the maximum of two linear objectives optimized over $\sT_\mu$, by standard tools in convex optimization, we can both optimize this objective and find its optimizer in polynomial time:
\begin{lemma}[see e.g.~\cite{grotschel2012geometric}]
\label{lem:convex-programming}
    For any $\delta > 0$, there is an algorithm which runs in time $\poly (n, \log 1/ \delta)$ and which, given $A \in \R^{n \times n}$, outputs $M \in \cT_\mu$ so that $\left| \iprod{A, M} \right| \geq \norm{A}_\mu - \delta$.
\end{lemma}
\noindent
We also need the following fact:
\begin{lemma}
\label{lem:rank-one-bound}
    For any vector $\delta \in \R^n$, we have that $\norm{\delta \delta^\top}_\mu = O(\norm{\delta}_\mu^2)$.
\end{lemma}
\begin{proof}
    From the proof of Theorem~\ref{thm:product-tv-bound}, and specifically~\Cref{eq:max-term} we know that 
    \[
    \norm{\delta}_\mu = \Theta \left( \max \left\{ \sum_{i = 1}^k |\delta_i|, \left(\sum_{i \geq k} \frac{\delta_i^2}{\mu_i} \right)^{1/2} \right \} \right) \; ,
    \]
    where we have taken the same ordering of coordinates and $k$ as in the proof of Theorem~\ref{thm:product-tv-bound}.
    Thus, it suffices to show that $\delta^\top M \delta$ can be upper bounded by the RHS for any $M \in \sT_\mu$.
    Since $M$ is PSD, we have that
    \[
    \delta^\top M \delta \leq 4 \sum_{i, j \leq k} M_{ij} \delta_i \delta_j + 4 \sum_{i, j > k} M_{ij} \delta_i \delta_j \; .
    \]
    The first term can be upper bounded by:
    \begin{align*}
    \sum_{i, j \leq k} M_{ij} \delta_i \delta_j &\leq \sum_{i, j \leq k} |\delta_i| |\delta_j| \leq \left( \sum_{i \leq k} |\delta_i| \right)^2 \; .
    \end{align*}

    On the other hand, we also have that
    \begin{align*}
    \sum_{i, j > k} M_{ij} \delta_i \delta_j &= \sum_{i, j > k} \sqrt{\mu_i \mu_j} M_{ij} \frac{\delta_i}{\sqrt{\mu_i}} \frac{\delta_j}{\sqrt{\mu_j}} \\
    &\leq \left( \sum_{i, j > k} \mu_i \mu_j M_{ij}^2 \right)^{1/2} \left( \sum_{i, j > k} \frac{\delta_i^2 \delta_j^2}{\mu_i \mu_j} \right)^{1/2} \\
    &\leq \sum_{i > k} \frac{\delta_i^2}{\mu_i} \; .
    \end{align*}
    Combining these two inequalities yields the final desired claim.
\end{proof}

\subsection{Regularity Conditions}
As is standard in robust statistics, we will condition on a set of deterministic conditions on the set of uncorrupted points $S_g$ that hold with high probability, and we will show that under these conditions, our algorithm succeeds, for any worst-case perturbation of $S_g$.
These conditions ensure that the empirical mean and variance of any of the types of test functions we will apply to the data are well-concentrated under the uncorrupted set of points.
One wrinkle is that because we have to use test functions from $\sT_\mu$, our regularity condition will also have to take this into account.
Formally:
\begin{definition}
    We say a set of points $T$ is \emph{$\eps$-good} with respect to a binary product distribution $\pi$ with mean $\mu$ if for all $\mu'$ satisfying $\mu'_i \geq \mu_i / 3$ for all $i = 1, \ldots, n$:
    \begin{itemize}
        \item We have that
        \begin{align}
        \label{eq:mean-conc}
            &\norm{\mu(T) - \mu}_\mu \lesssim \eps \log 1 / \eps \; , \mbox{ and} \\
            \label{eq:cov-conc}
            &\norm{\E_{X \sim T}  (X - \mu(T))(X - \mu(T))^\top  - \E_{X \sim \pi}  (X - \mu)(X - \mu)^\top }_{\mu'} \lesssim \eps \log^2 (1 / \eps) \; .
        \end{align}
        \item For all $w \leq w(T)$ with $\norm{w}_1 \leq \eps$, we have that 
        \begin{align}
        \label{eq:mean-dev}
       &\norm{\sum_{i = 1}^n w_i (X_i - \mu)}_\mu \lesssim \eps \log 1 / \eps \; , \mbox{ and} \\
       \label{eq:cov-dev}
       &\norm{ \sum_{i \in T} w_i (X_i - \mu)(X_i - \mu)^\top }_{\mu'} \lesssim \eps \log^2 (1 / \eps) \; .
        \end{align}
    \end{itemize}
\end{definition}

\noindent
The key statistical fact we will require is that a polynomial-sized set of samples from $\pi$ will be $\eps$-good with high probability.
For clarity of exposition, we defer the technical proof of this fact to~\Cref{sec:empirical-good-proof}:

\begin{lemma}
\label{lem:empirical-good}
    Let $\eps> 0$, and let $T = \{ X_1, \ldots, X_N \}$ be a set of $N \geq N_0$ independent samples from $\pi$, where $N_0 = \widetilde{O} \left( \frac{n^2 }{\eps^2} \right)$.
    Then, with probability $0.99$, $T$ is an $\eps$-good set of points for $\pi$.
\end{lemma}

\subsection{Key Geometric Lemma}
Before we state the geometric lemma, we will need the following operation:
\begin{definition}
    For any square matrix $A \in \R^{n \times n}$, let $\poff (A) \in \R^{n \times n}$ be given by $\poff (A) = A - \mathrm{diag} (A)$, i.e. the matrix $A$ with the diagonals zeroed out. 
\end{definition}
\noindent Note that $\poff$ is a projection onto a subspace, and is hence clearly linear.
We are now in a position to state the key lemma that forms the main structural basis of the algorithm, which states that deviations of the empirical mean in the $\norm{\cdot}_\mu$ norm can be controlled by deviations in the second second moment, after the diagonal has been zeroed out:
\begin{lemma}
\label{lem:key-geometric}
    Let $\pi$ be a binary product distribution with mean $\mu \in \R^n$ with $0 \leq \mu_i \leq 2/3$ for all $i = 1, \ldots, n$.
    Let $S = S_g \cup S_b \setminus S_r$ where $S_g$ is an $\eps$-good set of points for $\pi$, $S_r \subset S_g$, and $|S_b| = |S_r| = \eps |S|$, and suppose $S$ satisfies~\Cref{eq:regularity1} and~\Cref{eq:regularity2}.
    Let $w \in \sW_{N, \eps}$.
    Then
    \begin{equation}
    \norm{\mu(w) - \mu}_{\mu(w)} \leq \sqrt{\eps \cdot \sup_{y \in S_{\mu(w)}} y^\top \poff (\Sigma (w)) y} + O(\eps \log 1 / \eps) \; .
    \end{equation}
\end{lemma}
\begin{proof}
    Let $\eta = \norm{\mu(w) - \mu}_{\mu(w)}$, and let $y \in \cT_{\mu(w)}$ so that $\iprod{y, \mu(w) - \mu} = \norm{\mu(w) - \mu}_{\mu(w)} = \eta$.
    If $\eta \leq O(\eps \log 1 / \eps)$ then the inequality is trivial, so assume that $\eta = \omega(\eps \log 1 / \eps)$.
    Let $w_g, w_b$ be the restriction of $w$ to $S_g$ and $S_b$, respectively, and let $\left( \overline{w} \right)_i = 1/N - w_i$ for all $i$.
    Note that $\norm{\overline{w}}_1 \leq \eps$.
    We expand:
    \begin{align*}
        \eta &= \iprod{y, \mu(w) - \mu} = \E_{X \sim S_g} [\iprod{y, X - \mu}] + \norm{\overline{w}}_1 \E_{X \sim\overline{w} } \iprod{y, X - \mu} - \eps \E_{X \sim S_r} \iprod{y, X - \mu} \\
        &= O(\eps \log 1 / \eps) + \norm{\overline{w}}_1 \E_{X \sim \overline{w}} \iprod{y, X - \mu} \; ,
    \end{align*}
    by the $\eps$-goodness of $S_g$, and the observation that by~\Cref{eq:regularity2}, we have that $\tfrac{1}{3} y \in S_{\mu}$. 
    By Jensen's inequality, we next have that
    \begin{align}
    \label{eq:jensens}
       \E_{X \sim \overline{w}} \iprod{y, X - \mu}^2 \geq \left( \E_{X \sim\overline{w}} \iprod{y, X - \mu} \right)^2 \geq \left( \frac{\eta - O(\eps \log 1 / \eps)}{\eps} \right)^2 \gg \frac{\eta^2}{\eps^2} \; .
    \end{align}
    
\noindent
Next, observe that 
\begin{align*}
    y^\top \poff (\Sigma (w) ) y &= \E_{X \sim w} \iprod{y, X - \mu(w)}^2 - \sum_{i = 1}^n y_i^2 \mu(w)_i (1 - \mu(w)_i) \\
    &= \E_{X \sim w} \iprod{y, X - \mu}^2 - \iprod{y, \mu(w) - \mu}^2 - \sum_{i = 1}^n y_i^2 \mu(w)_i (1 - \mu(w)_i) \\
    &= \E_{X \sim w} \iprod{y, X - \mu}^2 - \sum_{i = 1}^n y_i^2 \mu(w)_i (1 - \mu(w)_i) - O(\eta^2) \; .
\end{align*}
We now further decompose the first term on the RHS:
\begin{align*}
    \E_{X \sim w} \iprod{y, X - \mu}^2 &= \E_{X \sim S_g} \iprod{y, X - \mu}^2 + \norm{\overline{w}}_1 \E_{X \sim \overline{w}} \iprod{y, X - \mu}^2 - \eps \E_{X \sim S_r} \iprod{y, X - \mu}^2 \\
    &= \E_{X \sim \pi} \iprod{y, X - \mu}^2 + \norm{\overline{w}}_1 \E_{X \sim \overline{w}} \iprod{y, X - \mu}^2 \pm O(\eps \log^2 (1 / \eps)) \\
    &= \sum_{i = 1}^n y_i^2 \mu_i (1 - \mu_i) + \norm{\overline{w}}_1 \E_{X \sim \overline{w}} \iprod{y, X - \mu}^2 \pm O(\eps \log^2 (1 / \eps)) \; .
\end{align*}
We also have that
\begin{align*}
    \left| \sum_{i = 1}^n y_i^2 \mu_i (1 - \mu_i) - \sum_{i = 1}^n  y_i^2 \mu(w)_i (1 - \mu(w)_i) \right| &\leq \left| \sum_{i = 1}^n y_i^2 (\mu_i - \mu(w)_i) \right| + \left| \sum_{i = 1}^n y_i^2 (\mu_i^2 - \mu(w)_i^2) \right| \\
    &\leq O(\eta) + \left| \sum_{i = 1}^n y_i^2 ((\mu_i + \mu(w)_i)) (\mu_i - \mu(w)_i) \right| \\
    &\leq O(\eta) \; ,
\end{align*}
where the last two lines follow because if $y \in S_{\mu(w)}$ it is easily verified that the vectors $y'$ and $y''$ defined by $y'_i = y_i^2$ and $y''_i = \frac{1}{2} y_i^2 ((\mu_i + \mu(w)_i))$ also belong to $S_{\mu(w)}$.
These calculations, along with~\Cref{eq:jensens}, imply that
\begin{equation}
    y^\top \poff (\Sigma (w)) y \geq \frac{\eta^2}{\eps} - O(\eta) - O(\eps \log^2 1 / \eps) \; ,
\end{equation}
which by rearranging immediately implies the desired claim.
\end{proof}

\subsection{Algorithm Description and Analysis}
\label{sec:alg-analysis}
We are now ready to state our algorithm.
\begin{algorithm}
\caption{A nearly-optimal robust learner for binary product distributions}\label{alg:main}
\KwIn{An $\eps$-corrupted set of samples from a product distribution $p \in \cP_n$}
\KwOut{A product distribution $\hat{p}$}
Let $C$ be a sufficiently large universal constant \\
$w \gets w(S)$ \\
\While{$\norm{ \poff (\Sigma (w))}_{\mu(w)} > C \eps \log^2 1 / \eps$}{
    Let $A \in \cT_{\mu(S)}$ be an $\delta$-approximate maximizer of $\iprod{A, \poff (\Sigma (w))}$ as per Lemma~\ref{lem:convex-programming}, where $\gamma = \poly(1/n, 1/ \eps)$.\\
    Let $\tau_i = (X_i - \mu(w))^\top A (X_i - \mu(w))$ for all $i \in S$ \\\
    Sort the $\tau_i$ in decreasing order. WLOG assume that $\tau_1 \geq \tau_2 \geq \ldots \tau_N$.\\
    Let $M$ be the first index so that $\sum_{i = 1}^M w_i > 2 \eps$. \\
    For every $i \leq M$, let
    \[
    w_i \gets \left( 1 - \frac{\tau_i}{\tau_1} \right) w_i \; .
    \]
    Let $S \gets \{i \in S: w_i \neq 0 \}$.
}
\textbf{return} The product distribution $\sigma$ with mean vector $\mu(w)$ 
\end{algorithm}

\begin{proof}[Proof of~\Cref{thm:robust-main}]
    First, note that the runtime is polynomial: by Lemma~\ref{lem:convex-programming} each loop of the algorithm runs in polynomial time, and since the loop removes at least one element of $i$, it can run for at most $n$ iterations.
    Moreover, since the quality of the approximation returned by the convex programming is so high, it is easily seen that it will not affect the downstream calculations, so for simplicity of exposition we will assume in the latter that we have an exact optimizer.

    We now turn our attention to correctness.
    Let $w^{(1)}, \ldots, w^{(T)}$ denote the sequence of weight vectors $w$ produced by the algorithm, so that $w^{(1)} = w([N])$, where we adopt the convention that $w^{(t)}_i = 0$ for $i \in S_r$ and all $i$ removed from $S$ by the algorithm.
    It suffices to show the following key invariant: for all $t \leq T-1$, we have that
    \begin{equation}
    \label{eq:invariant}
    \sum_{i \in S_g} w^{(t)}_i - w^{(t + 1)}_i \leq \sum_{i \in S_b} w^{(t)}_i - w^{(t + 1)} \; .
    \end{equation}
    This is because given~\Cref{eq:invariant}, by telescoping, this implies that $w^{(T)}$ is a set of weights with \[
    \norm{ \poff (\Sigma (w^{(T)}))}_{\mu (w)} \leq C \eps \log^2 1 / \eps
    \]
    and which satisfies $w^{(T)} \in \sW_{n, \eps}$, so by Lemma~\ref{lem:key-geometric}, we have that $\norm{\mu(w^{(T)}) - \mu}_{\mu(w^{(T)})} \leq O(\eps \log 1 / \eps)$, which by Theorem~\ref{thm:product-tv-bound} we have that $\tvd(\sigma, \pi) \leq O(\eps \log 1 / \eps)$, as claimed.

    To show~\Cref{eq:invariant}, we will proceed by induction.
    Fix some iteration $t \leq T - 1$, and suppose that~\Cref{eq:invariant} held for all $t' < t$.
    In particular, this implies that $w^{(t)} \in \sW_{n, \eps}$.
    Moreover, by Lemma~\ref{lem:filtering}, if we let $I^{(t)}$ denote the set of largest $\tau_i$ in this iteration, it suffices to show that
    \begin{equation}
        \sum_{i \in S_g \cap I^{(t)}} \tau_i w^{(t)}_i \leq \sum_{i \in S_b \cap I^{(t)}} \tau_i w^{(t)}_i \; .
    \end{equation}
    For the remainder of the proof, for clarity we will drop the subscript $t$, as we will only work with a single iteration.
    Let $w_g, w_b$ denote the restrictions of $w$ to $S_g$ and $S_b$, respectively.
    Observe that
    \begin{align*}
     \sum_{i \in S_g} w_i \tau_i &= \iprod{A, \Sigma (w_g) + (\mu(w_g) - \mu (w))(\mu (w_g) - \mu(w))^\top}  \\
         &= \iprod{A, \Sigma (w_g)} + O\left( \norm{\mu(w_g) - \mu(w)}_{\mu(w)}^2 \right) \\
        &= \iprod{A, \Sigma (w_g)} + O \left( \eps \iprod{A, \poff{\Sigma (w)}} + \eps \log 1 / \eps \right) \\
        &= \norm{w_g}_1 \iprod{A, \Sigma} + O \left( \eps \iprod{A, \poff{\Sigma (w)}} + \eps \log 1 / \eps \right) \; .
    \end{align*}
    Hence, we have that
    \begin{align*}
        \sum_{i \in S_b} w_i \tau_i &= \sum_{i \in S} w_i \tau_i - \sum_{i \in S_g} w_i \tau_i \\
        &= \iprod{A, \poff \Sigma (w)} + \iprod{A, \diag (\Sigma (w))} - \norm{w_g}_1 \iprod{A, \Sigma} \pm O \left( \eps \iprod{A, \poff{\Sigma (w)}} + \eps \log 1 / \eps \right) \; .
    \end{align*}
    By the same calculation as in the proof of Lemma~\ref{lem:key-geometric}, we have that
    \begin{align*}
    \left|\iprod{A, \diag (\Sigma (w))} - \norm{w_g}_1 \iprod{A, \Sigma} \right| &\leq O \left( \norm{\mu(w_g) - \mu(w)}_{\mu(w)} + \eps \right) \\
    &= O \left( \sqrt{\eps \iprod{A, \poff \Sigma (w)}} \right) \; ,
    \end{align*}
    and so since $\iprod{A, \poff \Sigma (w)} \geq C \eps \log^2 1 / \eps$, this implies that 
    \begin{align*}
        \sum_{i \in S_b} w_i \tau_i \geq \frac{3}{4} \cdot \iprod{A, \poff \Sigma (w)} \; .
    \end{align*}
    By our choice of $N$, we note that $|S_g \cap [M]| \geq \eps N$, as the bad points can only account for an $\eps$ amount of the mass.
    Therefore, by $\eps$-goodness and an application of Lemma~\ref{lem:key-geometric}, we have that 
    \[
    \sum_{i \in S_g \cap [M]} w_i \tau_i \leq O(\log^2 1/\eps +  \eps \iprod{A, \poff{\Sigma (w)}}) \; .
    \]
    In particular, by an averaging argument, since $\sum_{i \in S_g \cap [M]} w_i \geq \eps$, we conclude that 
    \[
    \tau_i \leq O(\log^2 1/\eps +  \eps \iprod{A, \poff{\Sigma (w)}})
    \]
    for all $i \geq M$.
    Thus, we have that 
    \begin{align*}
        \sum_{i \in S_b \cap [M]} w_i \tau_i &= \sum_{i \in S_b} w_i \tau_i - \sum_{i \in S_b \setminus [M]} \tau_i \\
        &\geq \frac{3}{4} \iprod{A, \poff \Sigma (w)} - \left( \sum_{i \in S_b} \tau_i \right) \cdot  O(\log^2 1/\eps +  \eps \iprod{A, \poff{\Sigma (w)}}) \\
        &\geq \frac{2}{3} \iprod{A, \poff \Sigma (w)} \\
        &\geq \frac{2}{3} \sum_{i \in S_g} w_i \tau_i \; ,
    \end{align*}
    and hence by Lemma~\ref{lem:filtering} we satisfy~\Cref{eq:invariant}, which completes the proof of the theorem.
\end{proof}

\subsection{Proof of Lemma~\ref{lem:empirical-good}}
\label{sec:empirical-good-proof}

    We split up the proof into several parts.
    Throughout this section, let $\eps, N_0, T$ be as in Lemma~\ref{lem:empirical-good}.
    We first prove the relevant statements for the concentration of the first moment, i.e.~\Cref{eq:mean-conc} and~\Cref{eq:mean-dev}:
    \begin{lemma}
        Suppose that $N \gtrsim \tfrac{n \log (1 / \delta)}{\eps^2}$.
        Then~\Cref{eq:mean-conc} and~\Cref{eq:mean-dev} hold simultaneously with probability $1 - \delta / 2$.
    \end{lemma}
    \begin{proof}
    Fix any $y \in \mathcal T_\mu$.
    By Bernstein's inequality, we have that if $X \sim \pi$, then for all $t > 0$, we have that
    \begin{equation}
        \Pr \left[ \left| \iprod{y, X - \mu} \right| > t \right] \leq \exp \left( - \frac{\frac{1}{2} t^2}{\sum_{i = 1}^n y_i^2 \mu_i + \frac{1}{3} t} \right) \leq \exp \left(-\Omega \left( \min \left( t, t^2\right) \right) \right) \; ,
    \end{equation}
    so in particular, the random variable $\iprod{y, X - \mu}$ is sub-exponential.
    Since the set of valid $y \in S_{\mu}$ is contained within the unit $\ell_\infty$ ball, by standard union bound arguments (see e.g.~\cite{vershynin2009high}), we have that for any $T' \subseteq T$, it holds that
    \begin{equation}
    \label{eq:mean-conc-1}
    \Pr \left[ \exists y \in \mathcal T_\mu: \left| \iprod{y, \mu(T') - \mu} \right| > t \right] \leq \exp \left( C_1 n \log (n / \eps) - c_1 |T'| \min(t, t^2) \right) \; ,
    \end{equation}
    for some universal constants $C, c$.
    In particular, this implies that $\norm{\mu(T) - \mu}_{\mu} \leq \eps$ with probability $1 - \delta$ so long as $N_0$ exceeds $O \left( \frac{n \log (n / \eps) + \log 1/ \delta}{\eps^2} \right)$.

    That~\Cref{eq:mean-dev} follows from~\Cref{eq:mean-conc-1} can then be easily shown using framework laid out in~\cite{li2018principled}, see e.g. the proof of Lemma 2.1.8 therein.
    \end{proof}

    We now turn to the proof of the bounds for the second moment, i.e.~\Cref{eq:cov-conc} and~\Cref{eq:cov-dev}.
    As it will not change anything in the argument, for simplicity of exposition in this proof we will replace all $\norm{\cdot}_{\mu'}$ with $\norm{\cdot}_{\mu}$.
    For any matrix $M \in \cT_\mu$, let $p_M (y) = y^\top M y - \E [(X - \mu)^\top (X  - \mu)]$, and let $Y_i = X_i - \mu$ for all $i \in T$.
    We first prove the following key inequality:
    \begin{lemma}
    \label{lem:quadratic-tail-bound}
        Let $M \in \cT_\mu$, and let $Y = X - \mu$, where $X$ is sample from the product distribution with mean $\mu$.
        There exists a universal constant $C$ so that for all $t \geq C$, we have
        \begin{equation}
        \label{eq:quadratic-tail-bound}
            \Pr \left[ | p_M (Y) | \geq t \right] \leq \exp \left( -\Omega (t^{1/2}) \right) \; .
        \end{equation}
    \end{lemma}
    \begin{proof}
        Let $Y = X - \mu$.
        We first break up the quadratic form into two terms:
    \[
    Y^\top M Y = \underbrace{\sum_{i = 1}^n M_{ii} Y_i^2}_{D} + \underbrace{\sum_{i \neq j} M_{ij} Y_i Y_j}_{O} \; .
    \]
    We control each term separately. 
    By Bernstein's inequality, we have that
    \begin{align*}
        \Pr \left[ \left| \sum_{i = 1}^n M_{ii} Y_i^2 - \sum_{i = 1}^n M_{ii} \mu_i (1 - \mu_i) \right| > t \right] &\leq \exp \left( - \frac{\frac{1}{2} t^2}{O \left( \sum_{i = 1}^n M_{ii}^2 \mu_i \right) + \frac{1}{3} t } \right) \\
        &\leq \exp \left( - \Omega (\min (t, t^2) \right) \; .
    \end{align*}
    The main challenge is controlling the off-diagonal term $O$.
    By standard decoupling results in Boolean analysis, see e.g.~\cite{dinur2006fourier,austrin2009randomly} or Theorem 2.4 in~\cite{diakonikolas2010bounding}, if we let $\sigma_i$ be new, independent, uniformly random $\{0, 1\}$-valued random variables, then 
    \begin{equation}
        \label{eq:cov-conc-1}
        \Pr \left[ \left| \sum_{i \neq j} M_{ij} Y_i Y_j \right| \geq t \right] \leq \Pr \left[ \left| \sum_{i \neq j} M_{ij} Y_i Y_j (1 - \sigma_i) \sigma_j \right| > 4 t \right] \; .
    \end{equation}
    Let $A$ denote the set of coordinates where $\sigma_i = 1$ and let $B$ denote the set of coordinates where $\sigma_i = 0$.
    Then, $A$ and $B$ form a random partition of $[n]$, and the random variable on the RHS of~\Cref{eq:cov-conc-1} is
    \[
     \sum_{i \neq j} M_{ij} Y_i Y_j (1 - \sigma_i) \sigma_j = \sum_{i \in A, j \in B} M_{ij} Y_i Y_j \; .
    \]

    Let $\Mbar$ denote the restriction of $M$ onto the indices of $A \times B$, and let $Y_A$ and $Y_B$ denote the vector of $Y_i$'s restricted to $A$ and $B$, respectively, so that we can write 
    \[
    \sum_{i \in A, j \in B} M_{ij} Y_i Y_j = Y_A^\top \Mbar Y_B \; .
    \]
    Fix some $\delta > 0$, and let $I_{\delta}$ denote the smallest set of elements of $A$ so that $\sum_{i \in I_\delta} \mu_i \geq \delta / 2$, that is $I_\delta$ consists of the top $\delta$ mass of elements of $A$, weighted by $\mu$.
    For any vector $v \in \R^A$, let $\gamma_\delta (v)$ denote the magnitude of the largest entry of $v$ in magnitude outside of $I_{\delta}$, and let $\nu_\delta (v) = \sum_{i \in I_\delta} |\mu_i v_i|$.
    We claim that for all $v \in \R^A$, there exists a universal constant $C$ so that for all $\delta$ sufficiently small, we have 
    \begin{equation}
    \label{eq:truncated-bernstein}        
    \Pr \left[ \left|\iprod{v, Y_A} \right| > C \left( \gamma_\delta (v) \log (1 / \delta) + \nu_\delta (v) + \sqrt{\log 1 / \delta} \cdot \left( \sum_{i \in A} \mu_i v_i^2 \right)^{1/2} \right) \right] \leq \delta \; .
    \end{equation}
    This is because
    \[
    \E \left[ \sum_{i \in I_\eps} X_i \right] \leq \delta / 2 \; ,
    \]
    and so by Markov's inequality, with probability $1 - \delta / 2$, all of the $X_i$ are $0$ for $i \in I_\delta$, and thus these terms contribute a $\nu_\delta (v)$ term to the sum.
    Then, if we condition on this event, the bound follows from Bernstein's inequality.
    Thus, to complete the proof, it suffices to show that if we let $v = \Mbar Y_B$, that the expression in~\Cref{eq:truncated-bernstein} is of order $\log^2 (1 / \delta)$ with probability $1 - O(\delta)$.
    Then the result follows by re-parameterizing $t = \log (1 / \delta)$.

    The fact that $\left| \nu_\delta (v) \right| \leq C \log (1 / \delta)$ with probability $\delta / 6$ follows directly from Bernstein's inequality.
    Next, we observe that
    \begin{align*}
        \left( \sum_{i \in A} \mu_i v_i^2 \right)^{1/2} &= \left( \sum_{i \in A} \mu_i \left(\sum_{j \in B} M_{ij} Y_j \right)^2 \right)^{1/2} \\
        &\leq  \sum_{i \in A} \left| \mu_i^{1/2} \sum_{j \in B} M_{ij} Y_j \right| \; ,
    \end{align*}
    but if we define $Z_i = \left| \mu_i^{1/2} \sum_{j \in B} M_{ij} Y_j \right|$, then by Bernstein's inequality, we have that
    \[
    \Pr \left[ Z_i \geq t \right] \leq \exp \left(- \Omega \left( \min \left( \frac{t^2}{\sum_{j \in B} M_{ij}^2 \mu_i \mu_j}, \frac{t}{\mu_i \cdot \sup_j |M_{ij}|} \right) \right) \right) \; ,
    \]
    and so by sub-exponential concentration (see e.g. Theorem 2.9.1 in~\cite{vershynin2009high}), and by the definition of $\cT_\mu$, we have that
    \[
     \Pr \left[ \sum_{i \in A} Z_i \geq t \right] \leq \exp \left( - \Omega (\min (t^2, t))\right) \; ,
    \]
    so in particular the probability this exceeds $ C \log (1 / \delta)$ is at most $\delta / 6$.

    We next show that $\gamma_\delta (v) \leq C \cdot \log 1 / \delta$ with probability $1 - \delta / 6$.
    To do this, by Markov's inequality, it suffices to show that for some constant $C > 1$ sufficiently large,
    \[
    \E \left[ \sum_{i \in A} \mu_i \cdot \mathbf{1} \left[ |v_i| \geq C \cdot \log 1 / \eps \right] \right] \leq \delta^{C} \; .
    \]
    Recalling that $v_i = \sum_{j \in B} M_{ij} Y_j$, by Bernstein's inequality, we have that for all $\zeta > 0$,
    \[
    \Pr \left[ |v_i | \geq \alpha_i \cdot \max \left( \sqrt{\log 1 / \zeta}, \log (1 / \zeta) \right) \right] \leq \zeta \; ,
    \]
    where 
    \[
    \alpha_i = \max \left( \left(\sum_{j \in B} M_{ij}^2 \mu_j\right)^{1/2}, \sup_j |M_{ij}| \right) \; .
    \]
    By the definition of $\cT_\mu$, the $\alpha_i$ satisfy that (1) $0 \leq \alpha_i \leq 1$ for all $i$, and (2) $\sum_{i \in A} \alpha_i^2 \mu_i \leq 1$.
    Now, let $A_k$ be the subset of rows of $A$ satisfying $\alpha_i \in [2^{-k}, 2^{- (k + 1)}]$, for $k = 0, \ldots, \infty$.
    By condition (2), we observe that $\sum_{i \in A_k} \mu_i \leq 4^k$, however, our tail bound implies that for any $i \in A_k$, and for all $\delta$ sufficiently small,
    \begin{align*}
    \Pr \left[ |v_i| \geq C \log 1 / \eps \right] \leq \delta^{C 2^k} \; ,
    \end{align*}
    and hence 
    \begin{align*}
        \E \left[ \sum_{i \in A} \mu_i \cdot \mathbf{1} \left[ |v_i| \geq C \cdot \log 1 / \eps \right] \right] &= \sum_{k = 0}^\infty \sum_{i \in k} \mu_i  \left[ |v_i| \geq C \log 1 / \eps \right] \\
        &\leq \sum_{k =0}^\infty 4^k \delta^{C 2^k} \ll \delta^{O(C)} \; ,
    \end{align*}
    so the desired claim follows from adjusting the choice of constant $C$, we conclude that with probability $1 - \delta / 2$, we have that
    \[
     \gamma_\delta (v) \log (1 / \delta) + \nu_\delta (v) + \sqrt{\log 1 / \delta} \cdot \left( \sum_{i \in A} \mu_i v_i^2 \right)^{1/2} \leq C' \log^2 (1 / \delta) \; .
    \]
    The result then follows by plugging in this bound into~\Cref{eq:truncated-bernstein}, and adjusting constants.
    \end{proof}
 
    We first prove~\Cref{eq:cov-conc}:
    \begin{lemma}
        Suppose that $N \gtrsim \tfrac{n^2}{\eps^2}$.
        Then~\Cref{eq:cov-conc} holds with probability $\geq 0.99$.
    \end{lemma}
    \begin{proof}
    We note that this condition is equivalent to the condition that for all $M \in \cT_\mu$, it holds that $\left| \iprod{M, \poff (\Sigma (T))} \right| \lesssim \eps \log^2 (1 / \eps)$.
    For any $j \neq \ell$, note that $\Sigma_{j, \ell} = 0$, so we have that $\E[(\Sigma(T)_{j, \ell})^2] = \tfrac{\mu_j \mu_\ell}{N}$, so consequently, we have that with probability at least $0.99$,
    \[
    \sum_{j \neq \ell} \frac{1}{\mu_j \mu_\ell} \Sigma(T)_{j, \ell}^2  \leq \frac{d^2}{N} \leq \eps^2 \; .
    \]
    Condition on this event holding.
    Then, for any $M \in \cT_\mu$, we have that
    \begin{align*}
        \left| \iprod{M, \poff (\Sigma (T))} \right| &= \left| \sum_{j \neq \ell} M_{j, \ell}\Sigma(T)_{j, \ell} \right| \\
        &\leq \left| \sum_{j \neq \ell} M_{j, \ell}^2 \mu_j \mu_\ell \right|^{1/2}  \left| \sum_{j \neq \ell} \frac{1}{\mu_j \mu_\ell} \Sigma(T)_{j, \ell}^2 \right|^{1/2} \\
        &\leq \eps \; ,
    \end{align*}
    by the definition of $\cT_\mu$.
    \end{proof}
    \noindent
    Finally, we prove~\Cref{eq:cov-dev}.
    Before we do so, we need the following result from VC theory, which follows since the VC dimension of degree-$2$ polynomial threshold functions in $n$ dimensions is $O(n^2)$:
    \begin{theorem}[VC inequality, see e.g.~\cite{devroye2001combinatorial}]
    \label{thm:vc}
        Let $D$ be any distribution over $\R^n$.
        Let $X_1, \ldots, X_N$ be a set of $N \gtrsim \tfrac{n^2 \log (n / \delta)}{\eps^2}$ independent samples from $D$.
        Then, with probability $1 - \delta$, the following holds: for all degree-$2$ polynomials $p: \R^d \to \R$, and all thresholds $\tau \in \R$, we have that
        \[
        \left| \frac{1}{N} \sum_{i = 1}^N \mathbf{1} [p(X_i) \geq \tau] - \Pr_{X \sim D} \left[ p(X) \geq \tau \right] \right| \leq \delta \; .
        \]
    \end{theorem}
    \begin{lemma}
        Suppose that $N \gtrsim \tfrac{n^2 \log (n / \delta)}{\eps^2}$, and that~\Cref{eq:cov-conc} holds.
        Then~\Cref{eq:cov-dev} holds with probability $1 - \delta / 4$.
    \end{lemma}
    \begin{proof}
        As before let, $Y_i = X_i - \mu$, for all $i \in T$.
        Since~\Cref{eq:cov-conc} holds,~\Cref{eq:cov-dev} is equivalent to the statement that for all $w \leq w(T)$ with $\norm{w} = 1 - \eps$, it holds that for all $M \in \cT_\mu$, we have that
        \[
        \left| \sum_{i \in T} w_ip_M (Y_i) \right| \lesssim \eps \log^2 (1 / \eps) \; .
        \]
        By convexity, it suffices to show that for all $T' \subset T$ satisfying $|T| = (1 - \eps) |T|$, we have that 
        \[
        \left| \frac{1}{|T'|} \sum_{i \in T'} p_M (Y_i) \right| \lesssim \eps \log^2 (1 / \eps) \; .
        \]
        Note that $p_M (y) \geq -\E_{X \sim \pi} \left[ (X - \mu)^\top M (X - \mu) \right] \geq -1$ deterministically.
        Let $\tau = C \log^2 (1 / \eps)$ for some constant $C$ to be specified later.

        We claim that the following holds: with probability $1 - \delta / 4$, it simultaneously holds for all $M$:
        \begin{align}
            \frac{1}{|T|}\sum_{i \in T} \mathbf{1} [p_M (Y_i) \geq \tau] &\leq \eps \; \mbox{ and} \label{eq:cov-dev-1} \\
            \left| \frac{1}{|T|} \sum_{i \in T} \min(p_M (Y_i), \tau) \right| &= O\left( \eps \log^2 (1 / \eps) \right) \label{eq:cov-dev-2} \; .
        \end{align}
        Suppose these two conditions hold.
        Let $T'$ be any set of size $(1 - \eps) |T|$, and let $M \in \cT_\mu$.
        We know that
        \begin{align*}
            \sum_{i \in T} p_M (Y_i) - \sum_{i \in T'} p_M (Y_i) \geq -\eps |T| \; ,
        \end{align*}
        which implies that $\tfrac{1}{|T'|} \sum_{i \in T'} p_M (Y_i) \leq \tfrac{1}{|T|} \sum_{i \in T} p_M (Y_i) + O(\eps)$.
        Since we also have that
        \begin{align*}
            \sum_{i \in T'} p_M (Y_i) &\geq \sum_{i \in T'} \min \left( p_M (Y_i), \tau \right) \\
            &\geq \sum_{i \in T} \min \left( p_M (Y_i), \tau \right) - \eps \tau \; ,
        \end{align*}
        we conclude that $\tfrac{1}{|T'|}  \sum_{i \in T'} p_M (Y_i) \gtrsim -\eps \log^2 (1 \ \eps)$.
        Together, these two claims imply that 
        \[
        \left| \frac{1}{|T'|}  \sum_{i \in T'} p_M (Y_i) \right| \lesssim \eps \log^2 (1 / \eps) \; ,
        \]
        which is what we wanted to show.

        Thus it suffices to show~\Cref{eq:cov-dev-1} and~\Cref{eq:cov-dev-2}.
        For any event $E$. let $\Pr_{T} [E]$ denote the fraction of elements in $T$ that satisfy $E$.
        Condition on the event that for all degree-$2$ polynomials $p$, we have that
        \[
        \left| \Pr_{T} [p \geq \tau]  - \Pr_{Y} [p(Y) \geq \tau] \right| \lesssim \eps \; .
        \]
        By~\Cref{thm:vc}, we know that this occurs with probability $1 - \delta / 4$.
        By~\Cref{lem:quadratic-tail-bound}, this immediately implies~\Cref{eq:cov-dev-1} holds.
        To show~\Cref{eq:cov-dev-2}, we have that
        \begin{align*}
             \frac{1}{|T|} \sum_{i \in T} \min(p_M (Y_i), \tau) &= \int_{-1}^0 \Pr_T [p \leq t] dt + \int_0^\tau \Pr_T [p \geq t] dt \\
             &= \int_{-1}^0 \Pr_Y [p(Y) \leq t] dt + \int_0^\tau \Pr_Y [p(Y) \geq t] dt \pm O(\eps \log^2 (1 / \eps)) \\
             &= - \int_{\tau}^\infty  \Pr_Y [p(Y) \geq t] dt \pm O(\eps \log^2 (1 / \eps)) \; , \\
        \end{align*}
        where the last line follows since $\E_Y [p(Y)] = 0$.
        To finish, we observe that by~\Cref{lem:quadratic-tail-bound}, we have
        \begin{align*}
            \left| \int_{\tau}^\infty  \Pr_Y [p(Y) \geq t] dt \right| \leq \int_\tau^\infty \exp (-t^{1/2}) dt = O(\eps \log^2 1 / \eps) \; ,
        \end{align*}
        as claimed.
    \end{proof}

\section{Non-Adaptive Lower bound for Single-Qubit Two-Outcome Projective Measurements} \label{sec:na-lb}
Notice that two-step adaptivity is crucial to the reduction in Theorem~\ref{thm:reduction}. Naturally, we ask if we can show that this adaptivity is inherent to the task at hand. We specifically do so for a restricted class of algorithms that are only permitted to perform non-adaptive \emph{single-qubit two-outcome projective measurements}, that is POVMs of the form:
$$\mathcal{M}=\bigotimes_{i=1}^{n}\{\ketbra{b_i}{b_i},\ketbra{b_i^\perp}{b_i^\perp}\}$$
This corresponds to separately measuring each qubit of each copy in some basis.
Specifically, we show the following lower bound.
\begin{theorem}\label{thm:lb-product-basis-meas} For some constant $c>0$,
    any algorithm for Problem~\ref{prob:agnostic-pmix} with $\epsilon = n^{-c}$ that achieves $1-\exp(-n^{c})$ error with probability at least ${\exp(-n^{c})}$ that uses measurements of the form $\mathcal{M}_1, \ldots, \mathcal{M}_N$, where the $\mathcal{M}_i$ are a set of non-adaptively chosen single-qubit two-outcome projective measurements, requires $N = {\exp(n^c)}$ copies. 
\end{theorem}
To prove this, we begin with a warmup that proves a lower bound for constant error by constructing a family of pairs of mixed states so that any given measurement has exponentially small probability of providing more than an exponentially small amount of information about which one it is observing.
\begin{proposition}\label{prop:main-lower-bound-prop}
For some constant $c>0$, and $\epsilon=n^{-c}$, there exists an ensemble of pairs of $n$-qubit mixed states $\rho_1,\rho_2$ that are each $\epsilon$-close in trace distance to product mixed states and $\Omega(1)$-far from each other in trace distance so that for any $\mathcal{M}=\otimes_{i=1}^n \mathcal{M}_i$ where $\mathcal{M}_i$ is a two-outcome projective measurement on the $i$-th qubit, with probability at least $1-\exp(-n^{c})$ over the choice of $\rho_1,\rho_2$, the distributions of measurement outcomes for $\mathcal{M}$ applied to $\rho_1$ and $\mathcal{M}$ applied to $\rho_2$ differ in trace distance by at most $\exp(-n^c)$.
\end{proposition}

In order to prove this hardness result, we will need to choose mixed states that can only be easily learned if the appropriate measurement basis is known. To do this, we will need to choose highly unbalanced product mixed states. In particular, we will pick a random, common product basis in which both distributions are diagonal and then construct $\rho_1$ and $\rho_2$ to be $\epsilon$-approximate products with respect to this basis. In particular, our mixed state will be equivalent to the stochastic process of sampling a bias parameter $t$ from some near-deterministic distribution and then independently setting each qubit in its respective unknown basis to be the first basis vector with probability $1-t$ and the second basis vector with probability $t$. Importantly, both mixed states when conditioned on $t$ are product mixed states.

For $\rho_1$, we will pick $t$ to be $m/n$ with high probability for some carefully chosen $m$, and for $\rho_2$, we will pick $t$ to be $(m+\sqrt{m})/n$ with high probability. This means that the number of qubits in the second basis vector for $\rho_1$ will be roughly $\mathrm{Poisson}(m)$, whereas for $\rho_2$ it will be roughly $\mathrm{Poisson}(m+\sqrt{m})$, which has constant total variational distance. This guarantees the separation between $\rho_1$ and $\rho_2$ in trace distance. Furthermore, in order to make these states hard to distinguish we will need our distributions over $t$ to match many moments. We construct these distributions using the following Lemma which follows from standard techniques in the literature on polynomial threshold functions and low-degree lower bounds.

\begin{lemma}\label{lemma:lb-moment-match}
    Let $m$ and $k$ be hyperparameters, and let $p_1,p_2,D_1,D_2$ be probability distributions such that
    $$p_1=(1-\epsilon)\delta_{\frac{m}{n}}+\epsilon D_1,\quad p_2=(1-\epsilon)\delta_{\frac{m+\sqrt m}{n}}+\epsilon D_2 \;,$$
    where $\epsilon$ is small. For any small positive constant $\beta$, there exists some constant $\gamma$ such that if $m=n^{\beta}=(k/\epsilon)^\gamma$, there exists a choice of $D_1$ and $D_2$ supported on $\bigl[0,\frac{2m}{n}\bigr]$ such that $\E_{t\sim p_1}[t^r]=\E_{t\sim p_2}[t^r]$ for all integers $0\leq r\leq k$. 
\end{lemma}
\begin{proof}
By translating the distributions in question by $m/n$, we note that it suffices to find distributions $D_1$ and $D_2$ so that for
    $$p_1'=(1-\epsilon)\delta_{0}+\epsilon D_1,\quad p_2'=(1-\epsilon)\delta_{\frac{\sqrt m}{n}}+\epsilon D_2$$
we have $\E_{t\sim p_1}[t^r]=\E_{t\sim p_2}[t^r]$ for all integers $0\leq r\leq k$. In particular, this means that $D_1$ and $D_2$ are distributions supported on $[-m/n,m/n]$ so that for $1\leq r \leq k$,
$$
\E[D_1^r] - \E[D_2^r] = \left(\frac{1-\epsilon}{\epsilon}\right)(\sqrt{m}/n)^r.
$$
If we let $D_1$ and $D_2$ have probability densities that differ by $p(x)dx$ for some function $p$ that we will chose, we need to find a $p$ with $\|p\|_1 \leq 2$ so that for $1\leq r\leq k$,
$$
\int_{-m/n}^{m/n} p(x)x^r dx = \left(\frac{1-\epsilon}{\epsilon}\right)(\sqrt{m}/n)^r
$$
and $\int_{-m/n}^{m/n}p(x)=0$ for $r=0$.
However, by Exercise 8.3 in~\cite{diakonikolas2023algorithmic}, this is possible so long as
$$
\mathrm{poly}(k)\max_{1\leq r \leq k} (\sqrt{m}/n)^r (n/m)^r/\epsilon < 1.
$$
Since $m>1$, this is equivalent to ${\rm poly}(k)(1/\sqrt{m})/\epsilon < 1$. Taking $m=(k/\epsilon)^\gamma$ for suitable $\gamma$, this is immediate.
\end{proof}

We can then use these two distributions to sample the shared bias parameter $t$ for each mixed state. Since Lemma~\ref{lemma:lb-moment-match} implies that $t=O(m/n)=O(n^{\beta-1})$ when $t\sim D_{\ell}$ for $\ell\in \{1,2\}$, conditioned on $t$, the resulting product mixed states will be very unbalanced. We now give a formal construction of our two mixed states which we show have constant separation in trace distance. 
\begin{lemma}\label{lemma:lb-states}
    Let $U$ be some product Haar unitary over $n$ qubits, meaning $U=\bigotimes_{i=1}^{n}U_i$ where $\{U_i\}_{i=1}^{n}$ are independent single-qubit Haar unitaries. Let $M(t)={\rm diag}(1-t,t)^{\otimes n}$ be a product mixed state with a shared bias $t$. For $\ell \in \{1,2\}$, define the mixed state:
    $$\rho_\ell=U\Tilde{\rho}_{\ell}U^\dagger,\quad \Tilde \rho_\ell =\E_{t\sim p_\ell}[M(t)] \;.$$
    Then, $\trd(\rho_1,\rho_2)=\Omega(1)$ for large $n$ and small $\epsilon$. 
\end{lemma}
\begin{proof}
    By unitary invariance,
    \begin{align*}
        \trd(\rho_1,\rho_2)&=\frac{1}{2}\|\Tilde \rho_1-\Tilde \rho_2\|_1=\tvd(P_1,P_2) \;,
    \end{align*}
    where $P_\ell={\rm Bern}^{\otimes n}(t)$ is conditionally a binary product distribution with a shared bias $t\sim p_\ell$. Let $Q_\ell={\rm Bern}^{\otimes n}(t)$ be similarly constructed with $t\sim D_\ell$. Then, by triangle inequality,
    \begin{align*}      
    \tvd(P_1,P_2)&\geq (1-\epsilon)\tvd\left({\rm Bern}^{\otimes n}\left(\frac{m}{n}\right),{\rm Bern}^{\otimes n}\left(\frac{m+\sqrt{m}}{n}\right)\right)-\epsilon \tvd(Q_1,Q_2) \\
    &\geq (1-\epsilon)\tvd\left({\rm Bin}\left(n,\frac{m}{n}\right), {\rm Bin}\left(n,\frac{m+\sqrt{m}}{n}\right)\right)-\epsilon \;.
    \end{align*}
    Denote these binomials as $B_1,B_2$ respectively. Let $A$ be the event that the outcome of the binomial is greater than $m+\sqrt{m}/2$. Then, for large $n$,
    \begin{align*}
       B_1\xrightarrow{d}Z_1\equiv \mathcal{N}\left(m,m(1-o(1)\right)\quad B_2\xrightarrow{d}Z_2\equiv \mathcal{N}\left(m+\sqrt{m},(m+\sqrt{m})\left(1-o(1)\right)\right) \;.
    \end{align*}
    Then, by the Berry-Esseen CLT, we obtain 
    \begin{align*}
        \Pr_{B_1}[A]&=\Pr\left[Z_1>\frac{\sqrt{m}/2}{\sqrt{m\left(1-o(1)\right)}}\right]\pm O(n^{-1/2})=1-\Phi(1/2)\pm o(1)\\
        \Pr_{B_2}[A]&=\Pr\left[Z_2>\frac{-\sqrt{m}/2}{\sqrt{(m+\sqrt{m})\left(1-o(1)\right)}}\right]\pm O(n^{-1/2})=\Phi(1/2)\pm o(1)
    \end{align*}
    Then, $\tvd(B_1,B_2)\geq |\Pr_{B_1}[A]-\Pr_{B_2}[A]|=2\Phi(1/2)-1\pm o(1)$ which is constant. Thus,
    $$\tvd(P_1,P_2)\geq (1-\epsilon)(2\Phi(1/2)-1\pm o(1))-\epsilon $$
    is also constant for small $\epsilon$ and sufficiently large $n$. 
\end{proof}
Let $\cM_{\rho_\ell}$ be the distribution over measurement outcomes achieved by measuring $\rho_\ell$ with $\cM$. We seek to show that $\tvd(\cM_{\rho_1},\cM_{\rho_2})=\exp(-n^c)$ is exponentially small. 

To analyze this, let $\mathcal{M}_{U,t}$ be the distribution on measurement outcomes when the measurement $\mathcal{M}$ is applied to the product mixed state $U \rho_t U^\dagger$ where $\rho_t = \mathrm{diag}(1-t,t)^{\otimes n}$ is the product mixed state in the standard basis with shared bias $t$. We note that $\cM_{\rho_\ell} = \E_{U,t\sim p_\ell}[\cM_{U,t}].$ We want to show that with probability $1-\exp(-n^c)$ over the choice of $U$ that $\| \E_{t\sim p_1}[\cM_{U,t}]- \E_{t\sim p_2}[\cM_{U,t}] \|_1 = \exp(-n^c).$

Suppose we measure $\rho_\ell$ with $\cM$ and get an outcome $F=\bigotimes_{i=1}^{n}F_i$ where $F_i=\ketbra{f_i}{f_i}$ such that $\ket{f_i}\in \{\ket{b_i},\ket{b_i^\perp}\}$. Let $\gamma_i=|\bra{0}U_i^\dagger \ket{b_i}|^2$ be the overlap between the random basis and the measurement basis. WLOG, $\gamma_i <1/2$ for all $i\in [n]$ by swapping the order of the basis elements for each qubit that violates this. The intuition is that any claimed algorithm's corresponding POVM will have low overlap with the random basis, hiding the approximate unbalanced product mixed state structure of the mixed state. 

We demonstrate this by arguing that conditioned on $t$, the probability of observing some measurement outcome $F$ can be written as a low-degree polynomial in $t$ plus a small error term. Since the expectation of low degree polynomials in $t$ are the same over $t\sim p_1$ and $t\sim p_2$, this will complete our proof.

In order to do this, we will want to analyze separately the coordinates for which the measurement basis and the true basis are very close. Noting that our state is a product mixed state when conditioned on $t$, we can separately consider the coordinates of low and high overlap as they are conditionally independent. 

Formally, let $\alpha$ be some small positive constant. We say a coordinate is \emph{good} if $\gamma_i > n^{-\alpha}$, and we say a coordinate is \emph{bad} if $\gamma_i \leq n^{-\alpha}$. Let $I$ denote the set of good coordinates. Then, by conditional independence,
$$
\cM_{U,t} = (\cM_{U,t})_{[n]-I} \times (\cM_{U,t})_{I}.
$$
Thus, we proceed by showing that the conditional probability is approximately low-degree in $t$, when restricting to each set of coordinates. We begin with the bad coordinates.

When approximating this distribution by one with polynomial entries, we might produce probabilities for the individual components that are neither strictly positive nor normalized. Thus, we define a pseudo-distribution to be simply a real-valued measure, and a (degree $k$) pseudo-distribution-valued polynomial to be a function from $\R$ to pseudo-distributions where the measure of any set is a degree at most $k$ polynomial.
\begin{lemma}\label{lemma:lb-bad}
If $\alpha\geq \beta$, with probability $1-\exp(-n^{\Omega(1)})$ over the choice of $U$, there exists a degree $k/2$ pseudo-distribution-valued polynomial $f_{[n]\setminus I}(t)$ so that for $t\in [0,2m/n]$ we have
$$
\|(\cM_{U,t})_{[n]\setminus I} - f_{[n]\setminus I}(t)\|_1 = \exp(-n^{\Omega(1)}). 
$$
\end{lemma}
\begin{proof}
    Since $\gamma_i \iid {\rm Unif}(0,1)$, there are $n_{\rm bad}=O(n^{1-\alpha})$ many bad coordinates with high probability. Since the distribution over measurement outcomes is obtained by applying a stochastic linear transformation to the distribution over states in the diagonal basis, it suffices to assume that the bad coordinates are perfectly bad, meaning $\gamma_i=0$ for $i\in [n]\setminus I$. Then, conditioning on $t$, the distribution of measurement outcomes is equivalent to ${\rm Bern}^{\otimes n_{\rm bad}}(t)$ if we appropriately label each qubit's measurement basis with $\{0,1\}$. Since this distribution is symmetric, it suffices to consider the distribution over the number of observed ones, denoted $s$. We have 
    \begin{align*}
    \Pr[{\rm Bin}(n_{\rm bad},t)=s]&=\binom{n_{\rm bad}}{s}t^s(1-t)^{n_{\rm bad}-s}\\
    &=\sum_{j=1}^{n_{\rm bad}-s}\binom{n_{\rm bad}}{s}\binom{n_{\rm bad}-s}{j}(-t)^{j+s}            \\
    &=\underbrace{\sum_{0\leq j \leq k/2-s}\binom{n_{\rm bad}}{s}\binom{n_{\rm bad}-s}{j}(-t)^{j+s}}_{f_{[n]\setminus I}(t)}+\underbrace{\sum_{k/2-s < j \leq n_{bad}-s}\binom{n_{\rm bad}}{s}\binom{n_{\rm bad}-s}{j}(-t)^{j+s}}_{\xi_{[n]\setminus I}(t)}
    \end{align*}
    Bounding the error term, 
    $$|\xi_{[n]\setminus I}|\lesssim \sum_{k/2-s < j \leq n_{bad}-s} (n_{\rm bad} t)^{j+s}.$$
    Since $t=O(m/n)$, $n_{\rm bad}t=O(n^{\beta-\alpha})$. If we set $\alpha\geq \beta$, 
    $$|\xi_{[n]\setminus I}|\leq e^{-n^{\Omega(1)}}$$
    since $k=\epsilon n^{\beta/\gamma}$. Summing this error over all $n_{\rm bad}+1$ possible values of $s$ gives our result.
\end{proof}
We now continue to the good coordinates. Expanding the conditional probability of producing a specific measurement outcome,
$$\Pr[\cM_{U,t}=\bigotimes_{i \in I} F_i]=\tr(F_I U\Tilde M(t) U^\dagger)=\prod_{i\in I} L_{F_i}(t)$$
where $L_{F_i}(t)=\tr( F_i U_i {\rm diag}(1-t,t) U_i^\dagger)=p_{F_i} (1-t)+q_{F_i} t$ with $p_{F_i}=|\bra{0}U_i^\dagger  \ket{f_i}|^2$ and $q_{F_i}=1-p_{F_i}$. Then,
$$L_{F_i}(t)=p_{F_i}+(q_{F_i}-p_{F_i})t$$
where $p_{F_i},q_{F_i} \geq  n^{-\alpha}$. Consider the logarithm of the probability after factoring out the leading term. 
$$\Pr\left[\cM_{U,t}=\bigotimes_{i \in I} F_i\right]=\exp(G_F(t))\prod_{i\in I}p_{F_i}$$
$$G_F(t)=\sum_{i\in I} \log(1+\eta_{F_i}  t),\quad \eta_{F_i}\equiv \frac{q_{F_i}-p_{F_i}}{p_{F_i}}\leq n^{\alpha}$$
We proceed by showing that $\exp(G_F(t))$ is approximately a low-degree polynomial in $t$ in two steps. First, we show in Lemma~\ref{lemma:Gf-moment} that all constant moments of $G_F(t)$ are approximately low-degree. Second, we show in Lemma~\ref{lemma:Gf-small} that with high probability over the observed $F$, $|G_F(t)|$ is small meaning that $\exp(G_F(t))$ is well-approximated by its Taylor expansion. We can then use both of these facts to handle $\exp(G_F(t))$ and thus the good coordinates.
\begin{lemma}\label{lemma:Gf-moment}
    If $\alpha+\beta < 1$, for any constant $r>0$, there exists a polynomial $g_F(t)$ of degree at most $r$ such that
    $$G_F(t)=g_F(t)+\zeta_F(t)$$
    where $|\zeta_F(t)|\leq n^{1-\Omega_{\alpha,\beta}(r)}$ when $t\in [0,2m/n]$. 
\end{lemma}
\begin{proof}
    If $\alpha+\beta <1$,
    $$|\eta_{F_i} t|\lesssim n^{\alpha}\cdot \frac{m}{n}=n^{\alpha+\beta-1}=n^{-\Omega_{\alpha,\beta}(1)}$$
    This justifies the Taylor expansion of each logarithm in $G_F(t)$. Bounding the $k$-th order truncation error, 
    $$\log(1+x)=T_{r}(x)+R_{r}(x) \text{ where } T_{r}(x)=\sum_{j=1}^{r}\frac{(-1)^{j+1}x^j}{j}\text{ and } R_{r}(x)=O\left(\frac{x^{r+1}}{k/2+1}\right)$$
    Summing this over all $\eta_{F_i}t$ for $i\in I$, we get
    $$
    G_F(t) = \sum_{i\in I} \log (1+\eta_{F_i} t) = \sum_{i \in I}(T_{r}(t)+R_{r}(t)) = \sum_{i \in I}T_{r}(t) + O(n\max_{i\in I}(|\eta_{F_i}t|)^{r+1}).
    $$
    Noting that $\max_{i\in I}(|\eta_{F_i}t|)=n^{-\Omega_{\alpha,\beta}(1)},$ letting
    $$
    g_F(t) = \sum_{i \in I}T_{r}(t),
    $$
    we have 
    $$|\zeta_F(t)| = O(n^{1-\Omega_{\alpha,\beta}(r)}).$$
\end{proof}
 At this point, we would like to argue that $|G_F|$ is small to justify the Taylor expansion of $\exp(G_F(t))$. However, $|G_F|$ could be large for arbitrary $F$ since $\eta_{F_i}\leq n^{\alpha}$ is only crudely bounded. Thus, we instead argue that $|G_F|$ is small with high probability over the observed measurement outcome $F$. 
\begin{lemma}\label{lemma:Gf-small}
    If $\alpha+\beta <1/2$ and $s\in[0,2m/n]$, then with probability $1-\exp(-n^{\Omega(1)})$ over the observed outcome $F\sim \cM_{U,s}$, there exists a small constant $\kappa>0$ such that for and $U,t$ with $t\in[0,2m/n]$, $|G_F(t)|\leq n^{-\kappa}$. 
\end{lemma}
    \begin{proof}
        Since $|\eta_{F_i} t|\leq n^{-\Omega(1)}$ for $i\in I$, we have that:
        $$|G_F|=\left|\sum_{i\in I}\log(1+\eta_{F_i} t)\right|=\left|\sum_{i\in I}(\eta_{F_i} t\pm O((\eta_{F_i} t)^2)\right|\leq t\left|\sum_{i\in I}\eta_{F_i}\right|+\sum_{i\in I}O(\eta_{F_i} t)^2$$
        Since $|\eta_{F_i}|\leq n^{\alpha}$, $\sum_{i\in I}(\eta_{F_i}t)^{2}\leq t^2\cdot n\cdot n^{2\alpha}=n^{2\alpha+2\beta-1}$, meaning the second order term is small if $\alpha+\beta <1/2$. For the first order term, we consider $\eta_{F_i}$ over the distribution of the observed measurement outcome $F\sim \cM_{\rho_\ell}$. We know that the probability of observing $\ketbra{b_i}{b_i}$ is:
        \begin{align*}
        &\tr \ketbra{b_i}{b_i}U_i M_i(t)U_i^\dagger =\gamma_i(1-t)+(1-\gamma_i)t=\gamma_i+(1-2\gamma_i)t
        =\gamma_i\pm O (m/n)
        \end{align*}
        Then, we have the following cancellation in the expectation of $\eta_{F_i}$. 
        $$\E[\eta_{F_i}]=(\gamma_i\pm O(m/n))\cdot \frac{(1-\gamma_i)-\gamma_i}{\gamma_i}+(1-\gamma_i\pm O(m/n))\cdot \frac{\gamma_i-(1-\gamma_i)}{1-\gamma_i}\leq O(m/n)\left(\frac{1}{\gamma_i}+\frac{1}{1-\gamma_i}\right)$$
        Since $\gamma_i \geq n^{-\alpha}$, $\E[\eta_{F_i}]\lesssim n^{\alpha+\beta-1}$. Thus,
        $$\E\left[\sum_{i\in I}\eta_{F_i}\right]\lesssim n^{\alpha+\beta}.$$

        On the other hand, the $\eta_{F_i}$ terms are independent and have absolute value at most $n^{\alpha}.$ Therefore, by Chernoff bounds for any constant $\nu>0$, with exponentially large probability, we have that:
        $$
        \left| \sum_{i\in I}\eta_{F_i}\right| = O(n^{\alpha+\beta})+O(n^{\alpha+\nu+1/2}).
        $$
        Thus, with exponentially large probability, we have that
        $$|G_F|\lesssim n^{\max(\alpha+2\beta-1,\alpha+\beta+\nu-1/2,2\alpha+2\beta-1)}.$$
        This means that as long as $2\nu+\kappa < 1-2\alpha-2\beta$, the claim holds.
         \end{proof}
Lemmas~\ref{lemma:Gf-moment} and \ref{lemma:Gf-small} gives us control over the moments and size of $G_F(t)$. We will now show that this suffices to control the good coordinates.
\begin{lemma}\label{lemma:lb-good}
    If $\alpha+\beta < 1/2$, there exists a pseudo-distribution-valued polynomial $f_{I}(t)$ of degree $k/2$ such that 
    $$\| (\cM_{U,t})_I - f_{I}(t)\|_1 = \exp(-n^{\Omega_{\alpha,\beta}(1)})$$
    for all $t\in[0,2m/n].$
\end{lemma}
\begin{proof}
    Fix any positive integer $r$. By Lemma~\ref{lemma:Gf-moment}, for a choice of $\alpha+\beta<1/2$, and small constant $\kappa>0$, if $F$ is such that $|G_F(t)|<n^{-\kappa}$ for all $t\in [0,2m/n]$, we have that:
    \begin{align*}
        \exp(G_F(t))&=\exp(g_F(t)+\zeta_F(t))\\
        & = \exp(g_F(t))\exp(\zeta_F(t))\\
        & = \exp(g_F(t))(1+n^{1-\Omega_{\alpha,\beta}(r)})\\
        & = \sum_{j=0}^\infty g_F(t)^j/j! + O(\exp(G_F(t))n^{1-\Omega_{\alpha,\beta}(r)})\\
        & =\sum_{j=0}^{\lfloor k/(2r) \rfloor} g_F(t)^j/j! + O(G_F(t))^{k/(2r)}+O(\exp(G_F(t))n^{1-\Omega_{\alpha,\beta}(r)}) \;,
    \end{align*}
    with exponentially large probability over $F$ we have that $|G_F(t)| < n^{-\kappa}.$ In this case, let $f_I(t)$ assign $F$ probability $\left(\sum_{j=0}^{\lfloor k/(2r) \rfloor} g_F(t)^j/j!\right) \prod_{i\in I}p_{F_i}$, the difference between this and the probability that $\cM_{U,t}$ assigns to $F$ is $(n^{1-\Omega_{\alpha,\beta}(r)}+n^{-k\kappa/(2r)})\prod_{i\in I}p_{F_i}.$ Letting $r=\sqrt{k}$, gives error $\exp(-n^{\Omega_{\alpha,\beta,\kappa}(1)})\prod_{i\in I}p_{F_i}$.

    Let $f_I$ assign these probabilities to $F$ for all $F$ where $|G_F(t)|<n^{-\kappa}$ for all $t\in [0,2m/n]$ and $0$ to all other $F$. The contribution to $\|(\cM_{U,t})_I-f_I(t)\|_1$ from the outcomes where $|G_F|$ is always small is at most $\exp(-n^{\Omega_{\alpha,\beta,\kappa}(1)})$ times the sum over $F$ of $\prod_{i\in I}p_{F_i}$. Since the latter is just the probability that $(\cM_{U,0})_I$ assigns to $F$, these sum to  $\exp(-n^{\Omega_{\alpha,\beta,\kappa}(1)})$. The contribution to the error coming from $F$'s where $|G_F(t)|$ is large is at most the probability that $(\cM_{U,t})_I$ assigns to these $F$'s, which by Lemma \ref{lemma:Gf-small} is also exponentially small.

    This completes our proof.
\end{proof}

We can now complete the proof of Proposition \ref{prop:main-lower-bound-prop}:
\begin{proof}
    We note that (with high probability over $U$)
    \begin{align*}\cM_{U,t} & = (\cM_{U,t})_I \times (\cM_{U,t})_{[n]\setminus I}\\
    & = f_I(t)\times f_{[n]\setminus I}(t) + ((\cM_{U,t})_I - f_I(t))\times (\cM_{U,t})_{[n]\setminus I} +  (\cM_{U,t})_I\times ((\cM_{U,t})_{[n]\setminus I} - f_{[n]\setminus I}(t))\\
    & \ \ \ \ \ \ \ \ \ \ - ((\cM_{U,t})_I - f_I(t))\times ((\cM_{U,t})_{[n]\setminus I} - f_{[n]\setminus I}(t)).
    \end{align*}
    The first term here has the same expectation over $t\sim p_1$ and $t\sim p_2$ and the latter terms all have exponentially small $L^1$ norms. As $\cM_{\rho_\ell} = \E_{t\sim p_\ell}[\cM_{U,t}]$, this completes the proof.
\end{proof}

We are now prepared to prove Theorem \ref{thm:lb-product-basis-meas}.
\begin{proof}
    Firstly, we note that if we only wanted to obtain a constant error lower bound with constant probability of error, we could do this easily directly from Proposition \ref{prop:main-lower-bound-prop}. In particular, if we feed the algorithm a random $\rho_\ell$ for a random $\ell$, with high probability over the choice of $(\rho_1,\rho_2)$, all of the measurements that the algorithm makes will have nearly identical output distributions over $\rho_1$ and $\rho_2$, and thus the measurements will provide the algorithm almost no evidence as to which is the correct density and will at best have to guess (as no $\rho$ can be close to both $\rho_1$ and $\rho_2$ in trace norm).

    In order to get exponentially close to $1$ error with exponential probability, we will alter this construction to have many copies of $\rho_1$ and $\rho_2$ and force the algorithm to guess almost all of them.

    In particular, for $t$ being some small power of $n$, we will construct $\epsilon$-noisy product densities on $nt$ qubits in the following way:

    Firstly, choose a set $S$ of $2^{\Omega(t)}$ many strings in $\{0,1\}^t$ so that any two strings in $S$ differ in at least $t/3$ of their coordinates.

    Next, create two $n$-qubit densities $(\rho_1,\rho_2)$ as in Proposition \ref{prop:main-lower-bound-prop}.

    Finally, pick a uniformly random $s\in S$ and let
    $$
    \rho := \rho_s := \bigotimes_{i=1}^t \rho_{s_i}.
    $$
    Note that since each $\rho_{s_i}$ is $n^{-c}$ close to a product mixed state, $\rho$ is $tn^{-c}$-close. In particular, so long as $t<n^{c/2}$, $\rho$ will be polynomially close to a product density.

    Suppose that an algorithm makes $\exp(n^{c'})$ many product measurements for $c'>0$ sufficiently small. By Proposition \ref{prop:main-lower-bound-prop}, with exponentially high probability over the choice of the pair $(\rho_1,\rho_2)$, restricting each such measurement to any block of $n$-coordinates would have measurement outcomes on $\rho_1$ and $\rho_2$ that are $\exp(-n^c)$-close. Assuming that $c'$ is small enough relative to $c$, this implies that if the algorithm is run on many copies of $\rho_s$ for different $s\in S$ that the distributions over the full outputs are exponentially close in total variational distance. This means that up to inverse exponential error, the output of the algorithm is independent of $s.$

    Thus, we merely need to prove that for any mixed state $\rho$ that, with exponentially high probability over a uniform random $s\in S$, the trace distance between $\rho$ and $\rho_s$ is at least $1-\exp(-n^{c'}).$ After we show this, altering the values of $n$ and $c$ appropriately gives our result. 

    For this, we note that since $\rho_1$ and $\rho_2$ have constant trace distance and because any $s,s'\in S$ with $s\neq s'$ disagree in $\Omega(t)$ coordinates, we have that the trace distance between $\rho_s$ and $\rho_s'$ is exponentially close to $1$. In particular, this implies that 
    $\tr(\rho_s \rho_{s'}) < \exp(-n^{c''})$ for some $c''>0.$ However, if there is a $\rho$ close to many of the $\rho_s$, we will need to have a $\rho$ so that $\tr(\rho\rho_s) > \exp(-n^{\Omega(c')})$ for exponentially many $s\in S$. However, letting $T$ be the set of such $s$ we have that
    $$
    \sum_{s\in T}\tr(\rho\rho_s) \geq |T| \exp(-n^{\Omega(c')}).
    $$
    On the other hand
    \begin{align*}
        \sum_{s\in T}\tr(\rho\rho_s) & = \tr\left( \rho \sum_{s\in T} \rho_S\right) \\
        & \leq \| \sum_{s\in T} \rho_S \|_F \\
        & = \sqrt{\sum_{s,s'\in T}\tr(\rho_s\rho_{s'})}\\
        & = \left(\sum_{s\in T}\tr(\rho_s^2) +  \sum_{s,s'\in T, s\neq s'}\tr(\rho_s\rho_{s'})\right)^{1/2}\\
        & = O(\sqrt{|T|}) + O(|T|) \exp(-n^{c''}).
    \end{align*}
    This leads to a contradiction if $c'$ is small enough and $|T|>|S|\exp(-n^{c'}).$

    This completes our proof.
\end{proof}

\section{Near-Optimal SQ Lower Bound for Robustly Learning Binary Products}
\label{sec:sq}

The body of this section is devoted to the proof of Theorem~\ref{thm:sq-informal}. We start with the basics of the SQ model. 

\paragraph{Background on SQ model}

The Statistical Query (SQ) model, introduced by \cite{Kearns:98} and extensively studied since, 
see, e.g.,  \cite{FGR+13},
considers algorithms that, 
instead of drawing individual samples from the target distribution, 
have indirect access to the distribution using an appropriate oracle.
A \emph{Statistical Query algorithm} is an algorithm 
whose objective is to learn some information about an unknown 
distribution $D$ by making adaptive calls 
to the following $\mathrm{STAT}$ oracle.

\begin{definition}[STAT Oracle] \label{def:stat-oracle}
Let $D$ be a distribution on $\R^n$. 
A Statistical Query is a bounded function $f: \R^n \to [-1,1]$. 
For $\tau > 0$, the $\mathrm{STAT}(\tau)$ oracle responds to the query $f$ 
with a value $v$ such that $|v - \E_{X \sim D}[f(X)]| \leq \tau$. 
We call $\tau$ the tolerance of the statistical query.
\end{definition}

\noindent The complexity of an SQ algorithm for a learning problem 
is quantified by the total number of queries to the $\mathrm{STAT}$ oracle 
(viewed as a measure of the algorithm's running time) and the maximum 
simulation complexity of any such query (viewed as a measure of the algorithm's sample complexity).
An \emph{SQ lower bound} for a learning problem is an unconditional statement 
that any SQ algorithm for the problem either needs to perform a large number $q$ of 
queries, or at least one query with very small tolerance $\tau$. 
By standard Chernoff bounds, a query of tolerance $\tau$ 
is implementable by non-SQ algorithms by drawing $O(1/\tau^2)$ 
samples and averaging them. Thus, an SQ lower bound intuitively 
serves as a tradeoff between runtime of $\Omega(q)$ and sample complexity of 
$\Omega(1/\tau)$.

We will use the framework of Statistical Query (SQ) algorithms for problems over distributions
introduced in~\cite{FGR+13}.
Before we get into the formal statement of our SQ lower bound, 
we formulate our task as a decision problem as follows:

\begin{definition}[Decision/Testing Problem over Distributions]\label{def:ht}
Let $D$ be a distribution and $\D$ be a family of distributions over $\R^n$. 
We denote by $\mathcal{B}(\mathcal{D},D)$ the hypothesis testing 
problem in which the input distribution $D'$ is promised to satisfy either 
(a) $D'=D$ or (b) $D'\in\mathcal{D}$, and the goal of the algorithm is to distinguish between these two cases.
\end{definition}

We will also need the following definition. 

\begin{definition}[Pairwise Correlation] \label{def:pc}
The pairwise correlation of two distributions with probability mass functions
$D_1, D_2 : \{0,1\}^n \to \R_+$ with respect to a distribution with mass $D:\{0,1\}^n \to \R_+$,
where the support of $D$ contains the supports of $D_1$ and $D_2$,
is defined as 
$$\chi_{D}(D_1, D_2) + 1 := \sum_{\bx\in\{0,1\}^n} D_1(\bx) D_2(\bx)/D(\bx) \;.$$
{When $D_1 = D_2$, the correlation $\chi_{D}(D_1, D_1)$ is identified with the 
$\chi^2$-divergence between $D_1$ and $D$, i.e., $\chi_{D}(D_1, D_1) = \chi^2(D_1, D)$. 
We will typically use the notation $\chi_{D}(D_1): = \chi_{D}(D_1, D_1).$}
\end{definition}

We say that a set of $s$ distributions $\mathcal{D} = \{D_1, \ldots , D_s \}$ over $\{0,1\}^n$
is $(\gamma, \beta)$-correlated relative to a distribution $D$ if
$|\chi_D(D_i, D_j)| \leq \gamma$ for all $i \neq j$, and $|\chi_D(D_i, D_i)| \leq \beta$.
With this notation, we are ready to define the notion of SQ dimension.

\begin{definition}[SQ Dimension] \label{def:sq-dim}
For $\gamma ,\beta> 0$ and a decision problem $\mathcal{B}(\mathcal{D},D)$, 
where $D$ is fixed and $\mathcal{D}$ is a family of distributions over $\{0,1\}^n$,
let $s$ be the maximum integer such that there exists a set of distributions
$\mathcal{D}_D \subseteq \D$ such that $\D_D$ is $(\gamma,\beta)$-correlated
relative to $D$ and $|\D_D|\ge s$.
We define the {\em Statistical Query dimension} with pairwise correlations $(\gamma, \beta)$
of $\mathcal{B}$ to be $s$ and denote it by $\mathrm{SD}(\mathcal{B},\gamma,\beta)$.
\end{definition}

\noindent The connection between SQ dimension and SQ lower bounds is captured
by the following lemma.

\begin{lemma}[\cite{FGR+13}] \label{lem:sq-from-pairwise}
Let $\mathcal{B}(\D,D)$ be a decision problem, where $D$ is the reference distribution 
and $\D$ is a class of distributions over $\R^n$. 
For $\gamma, \beta >0$, let $s= \mathrm{SD}(\mathcal{B}, \gamma, \beta)$. 
Any SQ algorithm that solves $\mathcal{B}$ with probability at least $2/3$ 
requires at least $s \cdot \gamma /\beta$ queries to the
$\mathrm{STAT}(\sqrt{2\gamma})$ oracles.
\end{lemma}

We note that the hypothesis testing problem of Definition~\ref{def:ht}
may in general be information theoretically hard. In particular, if some distribution
$D'\in\mathcal{D}$ is very close to the reference distribution $D$,
it will be hard to distinguish between $D'$ and $D$.
On the other hand, if $D'$ is far from the reference distribution $D$ in total variation distance
for any $D'\in\D$, then one can straightforwardly reduce the hypothesis testing problem
to the problem of learning an unknown $D'\in\D$ to small accuracy 
(see, e.g., Lemma~8.5 in Chapter 8 of \cite{diakonikolas2023algorithmic}).

\subsection{Generic SQ Lower Bounds Against Unbalanced Products} \label{ssec:gen-sq}

As is standard in the context of SQ lower bounds, we will establish SQ-hardness for a related
hypothesis testing problem that is efficiently reducible to our learning problem. 

Establishing an SQ lower bound in this setting 
essentially boils down to proving lower bounds 
for the corresponding SQ dimension (Definition~\ref{def:sq-dim}). 
In our case, this amounts to constructing large families of $\eps$-corrupted 
binary product distributions that have pairwise small $\chi^2$-inner product 
with respect to some given base distribution. 

In the most related prior work~\cite{diakonikolas2022optimal}, 
the base distribution was selected to be the uniform distribution over the hypercube.
Such a choice inherently fails in our setting in the sense that there is an SQ {\em upper bound}  
contradicting our desired lower bound. As already mentioned in our technical overview,
we choose the base distribution to be the $p$-biased binary product distribution $U_p^n$, 
where the parameter $p$ will eventually be chosen to be very close to $0$. 
At a high-level, this choice rules out using any of the moment-matching techniques of 
~\cite{diakonikolas2022optimal}, which crucially relied on the fact that the Binomial distribution
$\Bin(m, 1/2)$ is well-approximated by a Gaussian. Here we are in the ``Poisson approximation" 
regime, where we need to address the case of a Binomial with tiny success probability, namely $\Bin(m, 1/m)$.

The following proposition, which can be viewed as a generalization of an analogous result in~\cite{diakonikolas2022optimal}, 
encapsulates our generic discrete SQ lower bound construction. 
At a high-level, suppose that we have constructed a distribution $A$ 
of the appropriate type over $\{0,1\}^m$---for a carefully selected value of $m$ 
substantially smaller than the ambient dimension $n$---so 
that $A$ matches its low-degree moments with $U_p^m$. 
One can then use $A$ to obtain a large family of different distributions over $\{0,1\}^n$ 
by embedding it as a junta of the coordinates and using the distribution $U_p$ 
over the remaining coordinates. We show below 
that this allows us to construct many nearly orthogonal distributions, thereby implying
an SQ lower bound for the corresponding hypothesis testing problem.

\begin{proposition}[Generic SQ Lower Bound Construction] \label{prop:generic-SQ}
For $p > 0$, let $U_p$ be the $p$-biased Bernoulli distribution. For 
$k, m \in \Z_+$ with $k \leq m$, let $A$ be a distribution on $\{0,1\}^m$ 
that matches its first $k$ moments with $U_p^m$. 
For an injective function $f: [m] \to [n]$, where $n \gg m$, 
let $P^A_f$ be the distribution on $\{0,1\}^n$
that is equal to $A$ on the coordinates in the image of $f$ and equal to 
an independent $U_p^{n-m}$ on the remaining coordinates. 
Then any SQ algorithm that distinguishes between $P^A_f$, 
for randomly selected $f$, and $U_p^n$ requires either number of queries or inverse tolerance 
at least $(n/m)^{\Omega(k)}/ \chi_{U_p^m}(A).$
\end{proposition}

\begin{proof}
Our analysis will require the use of Fourier analysis for distributions on $\{0,1\}^n$. 
To do this, we will need an appropriate basis. 
In particular, let $X$ be the $1$-bit pseudo-distribution 
that assigns $\alpha$ to $0$ and $-\alpha$ to $1$, 
where $\alpha = \sqrt{p(1-p)}$. 
We note that $X$ and $U_p$ form an orthonormal basis 
of the distributions on $\{0,1\}$ with respect to the inner product $\chi_{U_p}(-,-)$. 

For a set $T \subset [n]$, define the pseudo-distribution $X_T$ 
to be the product of $X$ over the coordinates in $T$ and $U_p$ over the coordinates not in $T$. 
Note that the $X_T$'s form an orthonormal basis for distributions 
on $\{0,1\}^n$ with respect to $\chi_{U_p^n}(-,-)$.

As a consequence of the above, we can write
$A = \sum_{T \subseteq [m]} a_T X_T$
for some constants $a_T$ satisfying $\sum_{T \subseteq [m]} a_T^2 = \chi_{U_p^m}(A)$, 
$a_{\emptyset} = 1$, and $a_T = 0$ for $0 < |T| \leq k$ (where the latter conditions follow from
the assumed moment matching property of $A$).

By the definition of the distributions $P^A_f$, it is also easy to see that
$P^A_f = \sum_{T \subseteq [m]} a_T X_{f(T)}$.
Thus, for two such functions $f, g:  [m] \to [n]$, we can write 
\begin{align*}
\chi_{U_p^n}(P^A_f,P^A_g) &=  \chi_{U_p^n}\left(\sum_{T\subseteq [m]} a_T X_{f(T)} , \sum_{R\subseteq [m]} a_R X_{g(R)} \right) \\
&= \sum_{T,R \subseteq [m] : f(T) = g(R)} a_T a_R \\ 
&= 1 +  \sum_{T, R \subseteq [m] : f(T) = g(R), |T|>k, |R| > k} a_T a_R \;,
\end{align*}
where we used the moment-matching condition and the 
orthonormal property of the basis functions. 

To establish our desired SQ lower bound, we wish to bound the expectation (over random choices of $f$ and $g$) 
of the quantity 
$|\chi_{U_p^n}(P^A_f,P^A_g) - 1|$. 
This is at most
\begin{align*}
&\sum_{T, R \subseteq [m], |T|>k, |R| > k} |a_T a_R| \, \Pr_{f, g}\left[f(T) = g(R)\right] \\ 
&\leq \sum_{T, R \subseteq [m], |T| = |R| = s > k}  \frac{|a_T a_R| }{\binom{n}{s}} \\ 
&= \sum_{s>k} \left( \left( \sum_{T \subseteq [m], |T| = s} |a_T|\right)^2 / \binom{n}{s} \right) \\ 
& \leq \sum_{s>k} \chi_{U_p^m}(A) \binom{m}{s} / \binom{n}{s}   \\
&\leq  \chi_{U_p^m}(A) \sum_{s > k} (m/n)^s \;,
\end{align*}
where the penultimate inequality is Cauchy-Schwarz. 
{Since $n \gg m$, the latter sum is at most $(m/n)^{k}$.}

Note that the hypothesis testing problem that we are considering is simply 
$\mathcal{B}(\D,D)$ of Definition~\ref{def:ht} with $D=U_p^n$ and $\D=\{P_f^A\}_f$.
By standard arguments, the above correlation bound immediately implies a lower bound 
on the SQ dimension of $\mathcal{B}(\D,D)$.
Then an application of Lemma~\ref{lem:sq-from-pairwise} completes the proof of Proposition~\ref{prop:generic-SQ}. 
\end{proof}

Given our generic SQ lower bound result, 
it suffices to prove the existence of a moment-matching distribution $A$ over $\{0, 1\}^m$ 
such that an algorithm that learns $\eps$-corrupted binary products 
within total variation error $o(\eps \log(1/\eps)/\log\log(1/\eps))$ 
can distinguish between the hypotheses $P^A_f$, 
for randomly selected $f$, and $U_p^n$. By Proposition~\ref{prop:generic-SQ}, this would imply 
the desired SQ lower bound and prove Theorem~\ref{thm:sq-informal}.

This construction is shown in the following subsections.
Specifically, we will set the bias parameter $p$ to $1/m$. 
We prove that there exists a distribution $A$ that (i) matches its low-degree moments with $U^m_{1/m}$, 
and (ii) is an $\eps$-corrupted version of $U^m_{(1+\delta)/m}$
for any $\delta = o(\eps \log(1/\eps)/\log\log(1/\eps))$. Since the total variation distance 
between $U^m_{(1+\delta)/m}$ and $U^m_{1/m}$ is proportional to $\delta$, 
a robust mean estimation algorithm achieving total variation error $\ll \delta$ 
solves our hypothesis testing problem.

\subsection{Moment Matching Distributions over Integers} \label{ssec:mm-integers}

The main technical ingredient in our construction is the following result that 
we believe may be of broader interest: 

\begin{proposition} \label{prop:int-mm}
Let $k$ be a sufficiently large integer and 
let $d$ be a positive integer less than a sufficiently small constant power of $k$. 
Let $a_0,a_1,\ldots,a_d$  be real numbers. 
There exists a function $f$ on $\{-k,-k+1,\ldots,k\}$ so that
\begin{enumerate}
\item $\sum_i |f(i)| < \poly(d) \max_t( |a_t| O(d/kt)^t)$. 
\item $\sum_i f(i) i^t = a_t$ for $0 \leq t \leq d$. 
\end{enumerate}
\end{proposition}
\begin{proof}
We will prove that such a function $f$ exists by an LP duality technique. 
In particular, to prove the existence of such a function, it suffices to show that 
the following LP is feasible. 

Find $f$ on  $\{-k,-k+1,\ldots,k\}$ such that
\begin{enumerate}
\item $\|f\|_1 \leq B := \poly(d) \max_t( |a_t| O(d/kt)^t )$.
\item For any degree at most $d$ polynomial $p(x) = b_0 + b_1 x + \ldots + b_d x^d$, it holds that  
$\sum_i p(i)f(i) = \sum_t a_t b_t$.
\end{enumerate}
By LP duality, the above system of constraints has a solution unless 
there exists a polynomial $p$ such that $B \sup_{i \in \{-k, \ldots,k\}} |p(i)| < \sum_t a_t b_t$.

To make progress, we will compare this discrete problem 
to the corresponding real version. 
In particular, in the real version, the goal is to find a measure $\mu$ on $[-k,k]$ such that
\begin{enumerate}
\item $\|\mu\|_1 \leq B$.
\item For any degree at most $d$ polynomial $p$, $\E[p(\mu)] = \sum_t a_t b_t$.
\end{enumerate}
By LP duality, the above system has a solution unless 
there exists such a polynomial $p$ such that 
$$B \sup_{x \in [-k,k]} |p(x)| < \sum_t a_t b_t \;.$$

The real valued version of this problem is reasonably well-studied (see for example \cite{diakonikolas2023algorithmic} Lemma 8.18), but there doesn't seem to be a general statement in the literature with the correct concrete bounds. Thus, we show:
\begin{lemma}\label{real moments lemma}
    Let $a_0,\ldots,a_d$ be real numbers and $k>0$. There exists a function $p:[-k,k]\rightarrow \R$ so that
\begin{enumerate}
    \item $\sup_{x\in [-k,k]}|p(x)| < \poly(d) \max_t(|a_t|O(d/kt)^{t+1}).$
    \item $\int_{-k}^k f(x)x^t dx = a_t$ for $0\leq t \leq d.$
\end{enumerate}
\end{lemma}
\begin{proof}
  First, we note by homogeneity that it suffices to prove the statement for $k=1$ by finding a $g$ so that
  $\int_{-1}^1 g(x) x^t dx = a_t/k^t$ and then letting $f(x)=g(kx).$

  We can replace condition 2 above with 
  $$
  \int_{-1}^1 f(x)p(x) dx = \sum_{t=0}^d a_tb_t
  $$
  for any polynomial $p(x) = \sum_{t=0}^d b_t x^t.$ By linearity, it suffices to check this condition for $p(x) = P_t(x)$ for $P_t(x) = \sum_{s=0}^t d_{t,s}x^s$ the $t^{th}$ Legendre polynomial. In particular, if we take
  $$
  f(x) := \sum_{t=0}^d c_t P_t(x)
  $$
  where $c_t = \frac{2n+1}{2}\sum_{s=0}^t a_s d_{t,s}$, then condition 2 will hold by the standard orthogonality relations of the Legendre polynomials.

  To show the first condition, we note that $|P_t(x)| \leq 1$ when $|x|\leq 1$, so it suffices to bound $$\sum_{t=0}^d |c_t| \leq \poly(d) \max_{t,s} |a_s||d_{t,s}|.$$ However, given the representation
  $$
  P_n(x) = 2^{-n} \sum_{k>n/2}^n (-1)^{k+n} \binom{n}{k}\binom{2k}{2k-n} x^{2k-n}
  $$
  it is easy to see that $|d_{t,s}| = O(t/s)^s,$ from which our proof follows.
\end{proof}

Lemma \ref{real moments lemma} implies that the real-valued version of our problem has a solution with $\|\mu \|_1 < B/2$, and therefore the system 
with the $a$'s twice as large still has a solution. 
In particular, by LP duality this implies that for all such polynomials $p$,  
$\sup_{x \in [-k,k]} |p(x)| < 2 \sum_t a_t b_t$. 

To show that our original discrete dual program does not have a solution, 
it suffices to prove the following lemma.

\begin{lemma} \label{lem:disc-cont}
If $p$ is a polynomial of degree at most $d$, 
then $\sup_{x \in [-k,k]} |p(x)| < 2 \sup_{x \in \{-k,\ldots,k\}} |p(x)|$.
\end{lemma}
\begin{proof}[Proof of Lemma~\ref{lem:disc-cont}]
First, by making a change of variables, we note that this is equivalent to the following: 
for all polynomials $p$ of degree at most $d$, we have that
$$\sup_{x \in [-1,1]} |p(x)| < 2 \sup_{x \in \{-1+2/(2k+1),-1+4/(2k+1),\ldots,1\}} |p(x)| \;.$$
Let $m$ be a sufficiently large constant multiple of $d$ 
and define the intervals $$I_j = [\cos(\pi j/m), \cos(\pi(j-1)/m)]$$ for $j \in [m]$. 
Note that these intervals form a partition of $[-1,1]$. 
Note also that the length of any $I_j$ is $\Omega(1/m^2)$.  
So, if $k$ is at least a large enough multiple of $d^2$, 
for each $j$ there must be some element $x_j \in I_j \cap  \{-1+2/k,-1+4/k,\ldots,1\}$. 
Pick such an $x_j$ and define $r(x)$ to be the piecewise constant function on $[-1,1]$ 
defined by $r(x)$ is $p(x_j)$ on all of $I_j$. 
We note that it is now sufficient to prove that
$\|p\|_{\infty} < 2 \|r\|_{\infty}$,
or that $\|p-r\|_{\infty} < \|p\|_{\infty}/2$. 
Interestingly, this statement follows immediately from Lemma 2.1 of~\cite{KaneKP17},
completing the proof of Lemma~\ref{lem:disc-cont}.
\end{proof}

This completes our proof of Proposition~\ref{prop:int-mm}.
\end{proof}

\subsection{Moment Matching Distribution for Robustly Learning Products} \label{ssec:SQ-A}

We can now leverage Proposition~\ref{prop:int-mm} to construct
our moment-matching distribution $A$. Specifically, we show:

\begin{proposition}[Moment-Matching Corrupted Binary Product in Low Dimensions] \label{prop:final-A}
Let $\eps > 0$ and $L > 0$. 
There exists $\delta > \Omega(\eps \log(1/\eps)/\log\log(1/\eps))/L$ 
and $d = \tilde{\Omega}(\log(L))$ such that for positive integers 
$m \gg \log^2(1/\eps)$ there exists a distribution $A$ on 
$\{0,1\}^m$ satisfying the following: 
\begin{itemize}
\item[(i)] $A$ is $\eps$-close to $U_{(1+\delta)/m}^m$ in total variation distance, and 
\item[(ii)] $A$  matches $d$ moments with $U_{1/m}^m$.
\end{itemize}
\end{proposition}

\begin{proof}
We will choose $A$ to be a symmetric distribution, i.e.,  
$A$ is determined by the distribution over the sum of its coordinates.
Considering its distribution over weights, we need $A'$ 
to be $\eps$-close to $\Bin(m,(1+\delta)/m)$ and match $d$ 
moments with $\Bin(m,1/m)$. 

We choose an even integer $T$ to be a sufficiently small constant multiple 
of $\log(1/\eps)/\log\log(1/\eps)$, so that 
$\Pr(\Bin(m,1/m) = i) > \eps^{1/2}$ for each $i \leq T$. 
We will let $A'$ be $\Bin(m,(1+\delta)/m)$ plus some 
pseudo-distribution $\mu$ supported on $\{0,1,\ldots,T\}$. 
We note that it suffices to have $\|\mu\|_1 < \eps$ 
and
$$\sum_i \mu(i) i^t = \E[ \Bin(m,1/m)^t ] -  \E[ \Bin(m,(1+\delta)/m)^t ]$$
for $t \leq d$.
Note that the $L_1$ bound on $\mu$ would imply that it is pointwise less than $\eps^{1/2}$. 

Letting $\nu = \mu - T/2$, this is equivalent to finding a $\nu$ supported on $\{-T/2,\ldots,T/2\}$ 
so that for $t\leq  d$,
$$\sum_i \nu(i) i^t = \E[ (\Bin(m,1/m)-T/2)^t ] -  \E[ (\Bin(m,(1+\delta)/m)-T/2)^t ] \;.$$
This difference above is
$$\sum_s \binom{t}{s} (T/2)^{t-s} [  \E[ \Bin(m,1/m)^s ] -  \E[ \Bin(m,(1+\delta)/m)^s ]] \;.$$
Note that the $s=0$ terms cancel, the $ \binom{t}{s}$ terms 
sum to at most $2^d$ and the maximum remaining $(T/2)^{t-s}$ term is at most $T^{t-1}$. 
Thus, this sum is at most
$$2^d T^{t-1} \max_s  [ \E[ \Bin(m,1/m)^s ] -  \E[ \Bin(m,(1+\delta)/m)^s ]] \;.$$
Note that the ratio of the probabilities of the events $\Bin(m,1/m) = x$ 
and $\Bin(m,(1+\delta)/m) = x$ is $1+O(x \delta)$. 
Thus, the difference in expectations above is at most 
$O(\delta) \E[\Bin(m,1/m)^{s+1}]$, 
which is at most $\delta d^{O(d)}$.

Thus, applying Proposition~\ref{prop:int-mm}, we can find a $\mu$ of the form we want 
with $L_1$-norm at most $\delta/Td^{O(d)}$. 
Given our choice of parameters, this is sufficient.
\end{proof}

\subsection{Putting Everything Together} \label{ssec:final-SQ}

Given $n$, let $m$ be approximately $\sqrt{n}$ and define $A$ as in Proposition~\ref{prop:final-A}. 
Note that $A$ is $\Omega(\delta)$-far from $U_{1/m}^m$ in total variation distance. 
By Proposition~\ref{prop:generic-SQ}, it is $n^{\tilde{\Omega}(\log(L))}$-hard in SQ to distinguish 
between $P^A_f$, for random $f$, and $U_{1/m}^n$, 
which are two $\eps$-corrupted products which are 
$\Omega(\eps \log(1/\eps)/\log\log(1/\eps)L)$-far from each other.
This completes the proof of Theorem~\ref{thm:sq-informal}. \qed

\section{Quantum SQ Lower Bounds for Agnostic Tomography of Mixed Product States} \label{sec:qsq}

In this section, we give the proof of Theorem~\ref{thm:qsq-intro}. 
Specifically, we establish the following more detailed statement. 

\begin{theorem} \label{thm:qsq-body}
    Any QSQ algorithm that learns a product mixed state $\pi$ on $n$ qubits to trace distance $o(\eps \log (1/ \eps) / \log \log (1 / \eps))$, given QSQ access to a state $\rho$ satisfying $\trd (\rho, \pi) = \eps$, either requires $n^{\omega (1)}$ many quantum statistical queries, or must make at least one query of tolerance inverse super-polynomial in $n$.
\end{theorem}
\begin{proof}
    We claim that any such algorithm immediately implies a SQ algorithm for robustly learning product distributions over the hypercube with the same parameters, from which the claim follows from Theorem~\ref{thm:sq-informal}.

    Any classical distribution $D$ on $\{0, 1\}^n$ can be canonically encoded as a mixed state $\rho_D$ on $n$ qubits which is diagonal in the computational basis, namely
    \[
    \rho_D = \sum_{x \in \{0, 1\}^n} D(x) \ketbra{x}{x} \; .
    \]
    Moreover, it is straightforward to verify that $\tvd (D, D') = \trd (\rho_D, \rho_{D'})$ for all distributions $D, D'$.

    Suppose we have a QSQ algorithm for agnostically learning product mixed states.
    We will directly construct an SQ algorithm for robustly learning product distributions using the same number of queries and tolerance.
    First, note that trivially, any QSQ algorithm for agnostically learning product mixed states also implies a QSQ algorithm for agnostically learning diagonal product mixed states.
    Then, note that any QSQ algorithm for agnostically learning diagonal states can without loss of generally be replaced by one that only makes queries to diagonal observables.
    Additionally, observe that if $O$ is diagonal, then $\tr (O \rho_D) = \E_{X \sim D} [f(X)]$, where $f_O (x) = \bra{x} O \ket{x}$ satisfies $|f(x)| \leq \norm{O}_2$ for all $x \in \{0, 1\}^n$.

    The reduction is then as follows.
    To construct our SQ algorithm, we simply invoke our QSQ algorithm for agnostically learning diagonal mixed product states.
    Whenever the QSQ algorithm queries some diagonal observable $O$, we simply replace it with an SQ query to $f_O$.
    By the reasoning above, the behavior of the two oracles on $\rho_D$ and $D$ are the same, for all distributions $D$.
\end{proof}

\bibliographystyle{alpha}
\bibliography{refs}

\appendix

\section*{Appendix}

\section{Sample Near-Optimal Efficient Algorithm for Robustly Learning Product Distributions in $\ell_2$-Norm} \label{app:sample-opt-l2}

In this section we show the following result, which we crucially require to obtain the nearly optimal copy complexity in Section~\ref{sec:pure}:

\begin{theorem}
\label{thm:main-l2}
    Let $\eps_0 > 0$ be some universal constant.
    There is an algorithm which given an $\eps$-corrupted set of samples from an unknown product distribution $p \in \cP_n$ with mean $\mu$, for $\eps \leq \eps_0$, of size $N \geq N_0$, where $N_0 = \widetilde{O} \left( \tfrac{n}{\eps^2} \right)$, outputs with probability $0.99$ a mean vector $\widehat{\mu}$ so that $\norm{\muhat - \mu}_2 \lesssim \eps \sqrt{\log 1 / \eps}$.
    Moreover, the algorithm runs in time $\poly(N)$.
\end{theorem}
\noindent
As in Theorem~\ref{thm:robust-main}, we can easily boost the success probability to $1 - \delta$ by paying an additional $\log (1 / \delta)$ in the sample complexity and runtime.
We believe this result is essentially folklore in the community, but to our knowledge, has not been written down, and so we include it here for completeness.

The algorithm is again based on the filtering method, and is essentially a significantly simpler version of the algorithm required in Theorem~\ref{thm:robust-main}.
At a high level, because we now only insist on $\ell_2$ closeness, we do not need to use the complicated $\norm{\cdot}_\mu$ norm, and can use more classical spectral techniques.

We will crucially require the following notion of goodness.
For a symmetric matrix $M$, we let $\norm{M}_2$ denote the spectral norm of $M$.
\begin{definition}
    We say a set of points $T$ is \emph{$\eps$-Euclidean good} with respect to a product distribution $\pi$ with mean vector $\mu \in [0, 1]^n$ if:
    \begin{itemize}
        \item We have that
            \begin{align*}
                &\norm{\mu(T) - \mu}_2 \lesssim \eps \sqrt{\log 1 / \eps} \; , \; \mbox{and} \\
                &\norm{\E_{X \sim T} (X - \mu(T)) (X - \mu(T))^\top - \E_{X \sim \pi} (X - \mu) (X - \mu^\top)}_2 \lesssim \eps \log 1 / \eps \; .
            \end{align*}
        \item For all $w \leq w(T)$ with $\norm{w}_1 \leq \eps$, we have that
            \begin{align*}
                &\norm{\sum_{i \in T} w_i (X_i - \mu)}_2 \lesssim \eps \sqrt{\log 1 / \eps} \; , \: \mbox{and} \\
                & \norm{\sum_{i \in T} w_i (X_i - \mu) (X_i - \mu)^\top} \lesssim \eps \log 1 / \eps \; .
            \end{align*}
    \end{itemize}
\end{definition}
\noindent We have the following concentration inequality:
\begin{lemma}
\label{lem:l2-sample-complexity}
Let $\eps > 0$, and let $T = \{X_1, \ldots, X_N\}$ be a set of $N \geq N_0$ independent samples from $\pi$, where $N = O\left( \tfrac{n}{\eps^2} \right)$.
Then, with probability $0.99$, $T$ is an $\eps$-Euclidean good set of points for $\pi$.
\end{lemma}
\begin{proof}
    The key observation is that for any unit vector $v$, the random variable $Z = \iprod{v, X - \mu}$ for $X \sim \pi$ is sub-gaussian with variance proxy $1$ by Hoeffding's inequality.
    Therefore, the desired bound follows immediately from the same analysis as Appendix C of~\cite{dong2019quantum}.
\end{proof}
Crucially, notice that this is satisfied with a number of samples which is linear in $n$.
Intuitively, this is because the univariate projections are all sub-gaussian, and we only have to union bound over a net of size exponential in $n$, whereas before our random variables were sub-exponential, and our union bound was over a larger set of test matrices.

\subsection{Algorithm Description and Analysis}
We are now ready to state our algorithm.
The algorithm is very similar to Algorithm~\ref{alg:main}.
As a first step, we do the same preprocessing steps as in Algorithm~\ref{alg:main}; it is readily verified that these can be done in $\widetilde{O} (n / \eps^2)$ samples.
Thus, without loss of generality, we will assume that our data points satisfy~\Cref{eq:regularity1} and~\Cref{eq:regularity2}.

The main distinction is that rather than using the $\norm{\cdot}_\mu$ norm, we simply check the largest eigenvalue of the covariance (with diagonals zeroed), and we set the score to be the variance in the direction of the largest eigenvector.
We give the formal pseudocode in Algorithm~\ref{alg:main-l2}

The key geometric fact we require is the following, which is the analog of~\Cref{lem:key-geometric} for this setting:
\begin{lemma}
\label{lem:key-geometric-l2}
    Let $\pi$ be a binary product distribution with mean $\mu \in \R^n$ with $0 \leq \mu_i \leq 2/3$ for all $i = 1, \ldots, n$.
    Let $S = S_g \cup S_b \setminus S_r$ where $S_g$ is an $\eps$-Euclidean good set of points for $\pi$, $S_r \subset S_g$, and $|S_b| = |S_r| = \eps |S|$, and suppose $S$ satisfies~\Cref{eq:regularity1} and~\Cref{eq:regularity2}.
    Let $w \in \sW_{N, \eps}$.
    Then
    \begin{equation}
    \norm{\mu(w) - \mu}_2 \leq \sqrt{\eps \cdot \norm{ \poff (\Sigma (w))}_2} + O(\eps \sqrt{\log 1 / \eps}) \; .
    \end{equation}
\end{lemma}
\begin{proof}
    The proof is very similar to the proof of~\Cref{lem:key-geometric}.
    We may assume without loss of generality that $\eta \gtrsim \eps \sqrt{\log 1 / \eps}$, as otherwise the claim is trivially true.
    Let $y$ be a unit vector so that $\iprod{y, \mu(w) - \mu} = \norm{\mu(w) - \mu} = \eta$.
    By the same calculation as in~\Cref{lem:key-geometric}, but now using $\eps$-Euclidean goodness, we obtain that 
    \begin{align*}
    \E_{X \sim w} \iprod{y, X - \mu}^2 &= \sum_{i = 1}^n y_i^2 p_i (1 - p_i) + \norm{\bar{w}}_1 \E_{X \sim \bar{w}} \iprod{y, X - \mu}^2 \pm O(\eps \log 1 / \eps) \\
    &\geq \sum_{i = 1}^n y_i^2 \mu_i (1 - \mu_i) + \frac{\eta^2}{\eps} - O(\eps \log 1/ \eps) \; .
    \end{align*}
To bound the first term, we observe that 
\begin{align*}
    \left| \sum_{i = 1}^n y_i^2 \mu_i (1 - \mu_i) - \sum_{i = 1}^n  y_i^2 \mu(w)_i (1 - \mu(w)_i) \right| &\leq \left| \sum_{i = 1}^n y_i^2 (\mu_i - \mu(w)_i) \right| + \left| \sum_{i = 1}^n y_i^2 (\mu_i^2 - \mu(w)_i^2) \right| \\
    &\leq 2 \norm{\mu - \mu(w)}_2 = O(\eta) \; .
\end{align*}
Note in fact this bound is actually usually very loose.
Given this, we have that
\[
y^\top \poff(\Sigma (w)) y \geq \frac{\eta^2}{\eps} - O(\eta) - O(\eps \sqrt{\log 1 / \eps}) \; ,
\]
which by rearranging implies the claim, since $\eta \gtrsim \eps \sqrt{\log 1 / \eps}$.
\end{proof}

\begin{proof}[Proof of Theorem~\ref{thm:main-l2}]
    The analysis of this is almost identical to the proof of~\Cref{thm:robust-main} in~\Cref{sec:alg-analysis}; the only difference is that we all $\norm{\cdot}_{\mu(w)}$ with $\norm{\cdot}_2$, we replace $A$ with $vv^\top$, where $v$ is the top eigenvector, and we invoke Lemma~\ref{lem:key-geometric-l2} instead of Lemma~\ref{lem:key-geometric}.
\end{proof}

\begin{algorithm}
\caption{A nearly-optimal robust learner for binary product distributions}\label{alg:main-l2}
\KwIn{An $\eps$-corrupted set of samples from a product distribution $p \in \cP_n$}
\KwOut{A product distribution $\hat{p}$}
Let $C$ be a sufficiently large universal constant \\
$w \gets w(S)$ \\
\While{$\norm{ \poff (\Sigma (w))}_{2} > C \eps \log 1 / \eps$}{
    Let $v$ be the unit eigenvector of $\poff(\Sigma (w))$ corresponding to its largest eigenvalue in absolute value.\\
    Let $\tau_i - \iprod{v, X_i - \mu(w)}^2$ for all $i \in S$ \\
    Sort the $\tau_i$ in decreasing order. WLOG assume that $\tau_1 \geq \tau_2 \geq \ldots \tau_N$.\\
    Let $M$ be the first index so that $\sum_{i = 1}^M w_i > 2 \eps$. \\
    For every $i \leq M$, let
    \[
    w_i \gets \left( 1 - \frac{\tau_i}{\tau_1} \right) w_i \; .
    \]
    Let $S \gets \{i \in S: w_i \neq 0 \}$.
}
\textbf{return} $\mu(w)$ 
\end{algorithm}

\end{document}